\documentclass[11pt,letterpaper]{article}
\usepackage[margin=1in]{geometry}

    \usepackage[english]{babel}
    \usepackage{amsmath, amssymb, amsthm, mathtools, multirow, array, amsfonts,thm-restate, bbm, tikz, csquotes}    
    \usetikzlibrary{matrix, positioning, decorations.pathreplacing, arrows.meta, calc}
    \usepackage[T1]{fontenc}
    \usepackage[colorlinks=true, allcolors=blue]{hyperref}
    \usepackage{cleveref}
    \usepackage[shortlabels]{enumitem}
    \usepackage{thmtools, thm-restate}
    \usepackage[style=alphabetic]{biblatex}
    \declaretheorem[name=Lemma, parent=section]{lemma}
    \declaretheorem[name=Definition, parent=section, sibling=lemma]{definition}

    \declaretheorem[name=Claim,parent=section, sibling=lemma]{claim}
    \declaretheorem[name=Observation, parent=section, sibling=lemma]{observation}

    \DeclareMathOperator*{\best}{best}
    \DeclareMathOperator*{\suppopp}{supp}
    \newcommand{\supp}[1]{{\suppopp\left(#1\right)}}
    \DeclarePairedDelimiter\ceil{\lceil}{\rceil}
    \DeclarePairedDelimiter\floor{\lfloor}{\rfloor}
    \newcommand{\poset}[1]{2^{#1}}
    \newcommand{\denote}{\coloneq}
    \newcommand{\expect}[2][]{\mathbb{E}\ifthenelse{\not\equal{}{#1}}{_{#1}}{}\!\left[{\def\givenn{\middle|}#2}\right]}
    \newcommand{\prob}[2][]{\mathbb{P}\ifthenelse{\not\equal{}{#1}}{_{#1}}{}\!\left[{\def\givenn{\middle|}#2}\right]}
    \newcommand{\apx}{\rho}
    \newcommand{\harmonic}[1]{H_{#1}}

    \newcommand{\MMS}{\textrm{MMS}}
    \newcommand{\PROP}{\textrm{PROP}}
    \newcommand{\TPS}{\textrm{TPS}}
    \newcommand{\hyp}{\textrm{-}}

    \newcommand{\M}{{\mathcal{M}}}
    \newcommand{\goods}{\M}
    \newcommand{\N}{{\mathcal{N}}}
    \newcommand{\agents}{\N}
    \newcommand{\V}{{\mathcal{V}}}
    \newcommand{\Vadd}{\mathcal{V}_{\textrm{add}}}
    \newcommand{\D}{{{D}}}
    \newcommand{\Dj}{{{D^{(j)}}}}
    \newcommand{\Dnum}[1]{{{D^{(#1)}}}}
    \newcommand{\Dprime}{{{D'}}}
    \newcommand{\Real}{{\mathbb{R}}_{\geq 0}}

    \newcommand{\ix}{x}
    
    \newcommand{\allocs}{\mathcal{A}_{\M,\N}}
    \newcommand{\allocset}{A}
    \newcommand{\alloc}{\vec{A}}
    \newcommand{\allock}{{\vec{A}^{(k)}}}
    \newcommand{\allocj}{{\vec{A}^{(j)}}}
    \newcommand{\allocki}{{A_i^{(k)}}}
    
    \newcommand{\pk}{{p^{(k)}}}

    \newcommand{\fracallocs}{\mathcal{F}_{\M,\N}}
    \newcommand{\falloc}{\vec{f}}
    \newcommand{\f}{{f}}
    \newcommand{\fix}{{f_{i,x}}}
    \newcommand{\galloc}{\vec{\g}}
    \newcommand{\g}{{g}}
    \newcommand{\gix}{{g_{i,x}}}
    \newcommand{\fracD}{{{\vec{f}}^D}}
    \newcommand{\fracDix}{{f^D_{i,x}}}

    \newcommand{\mech}{M}
    \newcommand{\fracMech}{{M_{\textrm{frac}}}}
    \newcommand{\distMech}{{M_{\textrm{dist}}}}
    \newcommand{\randMech}{{M_{\textrm{rand}}}}
    \newcommand{\univMech}{{M_{\textrm{universal}}}}
    
    \newcommand{\cuq}{\textrm{CUQ}}
    \newcommand{\cuqalloc}{\overrightarrow{\cuq}}
    \newcommand{\cuqix}{{\cuq_{i,\ix}}}

    \newcommand{\pga}[1]{\galloc^{(#1)}}
    \newcommand{\pgix}[1]{{\g^{(#1)}_{i,x}}}
    \newcommand{\pgixj}{{\pgix{j}}}
    \newcommand{\pg}[1]{\g^{(#1)}}
    \newcommand{\dfa}[1]{\falloc^{(#1)}}
    
    \newcommand{\dfix}[1]{{\f^{(#1)}_{i,x}}}
    \newcommand{\dfixj}{\dfix{j}}
    
    \newcommand{\ip}{{i_{+1}}}
    \newcommand{\pitem}[2]{{y_{{#1}{#2}}}}
    \newcommand{\pitemij}{\pitem{i}{j}}
    \newcommand{\Yitems}[1]{Y_{#1}}
    \newcommand{\Yj}{\Yitems{j}}
    \newcommand{\Yip}{\Yitems{\ip}}
    \newcommand{\rank}[2]{{r_{{#1}{#2}}}}
    \newcommand{\rij}{\rank{i}{j}}
    \newcommand{\zset}[2]{{Z_{{#1}{#2}}}}
    \newcommand{\zsetij}{\zset{i}{j}}
    \newcommand{\fset}[2]{{F_{{#1}{#2}}}}
    \newcommand{\fsetij}{\fset{i}{j}}

    \newcommand{\defapx}{\alpha}
    \newcommand{\defi}{{d}}
    \newcommand{\weight}{w}
    \newcommand{\wapx}[1]{{\weight}^{#1}}
    \newcommand{\walpha}{\wapx{\defapx}}
    \newcommand{\Wcyc}{W}
    \newcommand{\Wcycapx}{\Wcyc^{\defapx}}

    \newcommand{\val}{v}
    \newcommand{\valp}{v'}
    \newcommand{\valprof}{\vec{v}}
    \newcommand{\valprofp}{\vec{v}'}
    \newcommand{\Vprofs}{\V^n}
    \newcommand{\Vaddprofs}{{\V_{\textrm{add}}^n}}

    \newcommand{\Pers}[1]{\Pi(#1)}
    \newcommand{\PerN}{\Pers{\N}}
    \newcommand{\per}{\pi}
    \newcommand{\percyc}[1]{\per^{(#1)}}
    \newcommand{\percycij}{{\percyc{j}_i}}
    \newcommand{\pinv}{\pi^{-1}}

    \newcommand{\matr}{K}
    \newcommand{\port}{B}
    \newcommand{\portij}{{B_{ij}}}
    \newcommand{\portmat}{{B}}
    \newcommand{\portmats}{{\mathcal{B}_n}}

    \newcommand{\en}[1]{\vec{e}^{(#1)}}
    \newcommand{\ei}{\en{i}}
    \newcommand{\eiat}[1]{e^{(i)}_{#1}}

\title{Truthful-in-Expectation Mechanisms for MMS Approximation}

\usepackage{datetime}
\date{\today}
\author{
    Moshe Babaioff
        \thanks{The Hebrew University of Jerusalem. {Emails: \texttt{\{moshe.babaioff, noam.manakermorag\}@mail.huji.ac.il}}}
    \and Uriel Feige
        \thanks{Weizmann Institute of Science. {Email: \texttt{uriel.feige@weizmann.ac.il}}}
    \and Noam Manaker Morag
        \footnotemark[1]
}

\begin{document}
\maketitle

\begin{abstract} 
We study fair allocation of indivisible goods among strategic agents with additive valuations. Motivated by impossibility results for deterministic truthful mechanisms, we focus on randomized mechanisms that are \emph{Truthful-in-Expectation (TIE)}. From a fairness perspective, we seek to guarantee every agent a large fraction of their \emph{Maximin Share (MMS)} ex-post. Among other results, Bu~and~Tao~[FOCS 2025] presented a TIE mechanism that guarantees $\frac{1}{n}$-MMS ex-post. First, we present an ordinal TIE mechanism that guarantees $\frac{1}{H_{n-1} + 2}$-MMS ex-post, where $H_k$ is the $k$-th harmonic number ($H_k \simeq \ln k$). This is nearly best possible for ordinal mechanisms, as even non-truthful ordinal allocation algorithms cannot obtain an approximation better than $\frac{1}{H_n}$. We then show that with just a small amount of additional cardinal information, the ex-post guarantee can be improved to $\Omega(\frac{1}{\log\log n})$-MMS, at the cost of relaxing the incentive requirement to $(1-\varepsilon(n))$-TIE for negligible $\varepsilon(n)$. Finally, for two agents, we present a TIE mechanism that is $\frac{2}{3}$-MMS ex-post.

All our mechanisms are ex-ante proportional (thus also providing ``Best-of-Both-Worlds'' results) and run in polynomial time. Moreover, all our results extend to the truncated proportional share (TPS), which is at least as large as the MMS. Our two-agent $\frac{2}{3}$-TPS result is best possible for the TPS.
\end{abstract}

\newpage

\section{Introduction}

The field of \emph{fair division} studies the problem of fair allocation of resources among equally-entitled agents with heterogeneous preferences over the resources.
A central problem within this field, which has attracted significant research effort, is the fair allocation of \emph{indivisible goods} among equally-entitled agents with additive valuations; see, for example, \cite{M-04, LMMS-04, B-11, KPW-18,CGH-19, PR-20, CKMS-21, AABFLMVW-23, CGM-24}. In this paper, we study this fundamental problem in the framework of \emph{Algorithmic Mechanism Design}, when agents are strategic and their preferences are private information, and the goal is to design truthful fair mechanisms. 

Many fairness criteria have been suggested for this problem. These generally fall into two categories: \emph{envy-based} and \emph{share-based}. Envy-based notions, such as Envy-Freeness (EF) and its relaxations (e.g., EF1, EFX), focus on comparative satisfaction, requiring that no agent prefers another agent's bundle to their own, possibly up to the value of some item. In contrast, share-based notions take a different approach, seeking to guarantee each agent a minimum value, which depends only on the valuation function of the agent, and not on how the remaining goods are allocated to the other agents. 
In this paper we focus on share-based fairness notions.

For agents with additive valuations, a natural share is the  \emph{proportional share}, that for a setting with $n$ agents is defined to be a $1/n$ fraction of the agent's value for all goods. Unfortunately, as goods are indivisible, it is not possible to guarantee every agent her proportional share (i.e., when there is only a single valuable item). An alternative share is the \emph{Maximin Share ($\MMS$)},  introduced by Budish~\cite{B-11}. The $\MMS$ represents the value an agent can guarantee for herself if she was asked to partition the goods into $n$ bundles, knowing she would receive the lowest-valued bundle among them. While the $\MMS$ provides a compelling fairness benchmark, it cannot always be guaranteed: For agents with additive valuations,  $\MMS$ allocations always exist for $n=2$ agents, but it was shown \cite{KPW-18}  that  $\MMS$ allocations do not always exist for $n\geq 3$ agents. Consequently, a significant body of work has focused on the challenge of {establishing} existence of allocations where every agent gets a high constant-factor fraction of the $\MMS$ (an approximation). Currently, the best approximation of the $\MMS$ is $\frac{7}{9}$ \cite{HZ-25}.

The above results assume that agents' valuations are public knowledge. 
In many real-world scenarios, valuations are private information, and an agent has the freedom to report any valuation. To address the problem that optimizing the report might be challenging (as it might depend on the behavior of others), we follow the line of work which seeks to design fair mechanisms in which agents have simple dominant strategies (being ``truthful''), relieving them of the burden of strategic reasoning and making optimal participation straightforward.

Ideally, we would like to find truthful mechanisms that give a good fairness guarantee (a large fraction of the $\MMS$). Unfortunately, there are strong impossibility results for \emph{deterministic} truthful mechanisms. For the case of two agents, it was shown \cite{ABCM-17} that no deterministic truthful mechanism that partitions all goods can guarantee  any reasonable fairness property. In particular, they showed that any such mechanism cannot guarantee every one of the two agents more than $O(1/m)$ fraction of their $\MMS$, where $m$ is the number of goods. This matches the poor approximation achieved by the trivial truthful mechanism in which the first agent picks one good, and the second agent receives all the rest. This impossibility extends to any number of agents under the additional assumptions of non-bossiness and neutrality \cite{BMM-25}.

Motivated by the strong negative results for fair deterministic truthful mechanisms, a new line of work considers \emph{randomized truthful mechanisms} \cite{BT-24, GNS-25}. Crucially, the mechanisms should still ensure fairness for every realized allocation (ex-post), and randomization is introduced within the mechanism in order to address the issue of incentives. In this setting, we can define two different incentive compatibility requirements. The first and stronger requirement is \emph{universal-truthfulness}, which requires that truthful reporting is a dominant strategy even if the agent knows the outcomes of the random coin-tosses in advance. Equivalently, it is just a randomization over deterministic truthful mechanisms. As such, all the previous impossibility results for the deterministic setting hold for universally-truthful mechanisms as well. Thus, both prior work and our paper consider the relaxed incentive requirement of \emph{Truthfulness in Expectation (TIE)}, which requires that reporting her true preference maximizes an agent's utility in expectation over the randomness of the mechanism, regardless of the reports of others. Equivalently, TIE mechanisms are simply truthful (dominant strategy) for expectation maximizing agents (i.e., risk-neutral agents).

For expectation-maximizing agents with additive valuations, the approach of using randomized TIE mechanisms was successful in circumventing the strong impossibility results that hold for deterministic mechanisms. Indeed, Bu~and~Tao \cite{BT-24} have recently presented several TIE mechanisms for $n$ additive agents, including, among others, a TIE mechanism that guarantees every agent at least a $1/n$ fraction of her $\MMS$ ex-post ({see \Cref{sec:related} for further discussion of their results}). While this approximation is much better than $O(1/m)$ when the number of goods $m$ is large, it deteriorates quickly as the number of agents increases. {We seek to improve this approximation result. We present ex-post $\MMS$ approximations that are independent of $m$, and do not deteriorate as quickly when $n$ grows.}

Up to this point we have discussed fairness ex-post (for every realized allocation). But agents might find such guarantee by itself unsatisfactory, as in the case of a single good that is always allocated to one specific agent. To overcome this, we also make sure that the distribution over allocations is proportional in expectation (ex-ante). Our simultaneous ex-post and ex-ante guarantees fit into the framework of  \emph{Best-of-Both-Worlds (BoBW)} distributions \cite{A-19, FSV-20, A-20, BEF-22, HSV-22, AGM-23, FMNP-24, AAGW-15, HPPS-20, BEF-21, AzizFSV24, BF-26}, but we also make sure the mechanisms that generate these distributions are truthful.

Finally, we also care about computational aspects. Our mechanisms are all poly-time computable, and moreover, {can be implemented with a low communication burden, which depends only on $n$ and $m$} and not on the representation of the valuations.

\subsection{Our Results}
\label{subsec:results}

We study the problem of fairly  allocating a set of $m$ goods $\goods=\{1,\ldots,m\}$ to a set of $n\geq 2$ strategic agents $\agents=\{1,\ldots,n\}$, each with a private additive valuation over the goods (``additive agents''). We assume that agents are expectation maximizers (risk-neutral), and we seek mechanisms which incentivize agents to report truthfully and that always allocate the goods in a fair way.

Our first main result is a poly-time randomized TIE mechanism for $n$ additive agents that guarantees every agent at least a $\frac{1}{\harmonic{n-1}+2}>\frac{1}{\ln n+3}$ fraction of her $\MMS$ ex-post. In fact, the mechanism does not just output an allocation chosen at random (a sample from the underlying distribution), but rather outputs an explicit representation of the distribution (from which one can sample to get a randomized mechanism). We call such a mechanism a \emph{distributional mechanism}. Distributional mechanisms have several advantages over randomized mechanisms, most notably that they allow each agent to verify that she indeed gets her proportional share ex-ante (see \Cref{subsec:prelim-mech} for further discussion). Moreover, this mechanism, as well as all our other mechanisms, are ex-ante proportional.

\begin{restatable}{theorem}{lognapx}
    \label{thm:logn-apx}
    Consider the problem of allocating $m$ indivisible goods to $n$ \emph{additive} agents. There exists a distributional mechanism which is \emph{truthful-in-expectation (TIE)}, \emph{ex-ante proportional}, and $\frac{1}{\harmonic{n-1}+2}$-$\MMS$ ex-post. Moreover, this mechanism outputs, in polynomial time, an explicit representation of the distribution, which is supported on at most $n\cdot m$ allocations. 
\end{restatable} 

Since the harmonic number $\harmonic{n}$ is roughly $\ln n$, this result represents an exponential improvement over the best known truthful $\MMS$ approximation of $\frac{1}{n}$ \cite{BT-24}. In fact, our exact approximation factor of $\frac{1}{\harmonic{n-1}+2}$ strictly exceeds $\frac{1}{n}$ for all $n \ge 4$. 

The above mechanism has the additional advantage of being implementable in dominant strategies as an \emph{ordinal} mechanism. The mechanism asks each agent to report a strict order over the goods (a ranking of the goods, rather than their specific numerical values), and it is a dominant strategy for each agent to report such an order that is consistent with her weak preference over the goods. Such an implementation has communication that is independent of the representation of the valuation functions. Our $\MMS$ approximation is essentially best possible for an ordinal mechanism, as even non-truthful ordinal allocation algorithms cannot obtain an approximation better than $\frac{1}{H_n}$ \cite{ABM-16}. Thus, to improve upon this logarithmic approximation, it is essential to utilize cardinal information.

In our second main result, we show that with just a small amount of additional cardinal information, the ex-post guarantee can be improved to $\Omega(\frac{1}{\log\log n})$-MMS, at the cost of relaxing the incentive requirement to $(1-\varepsilon(n))$-TIE, for negligible $\varepsilon(n)$. We note that similar relaxations of truthfulness have been explored in other settings in prior literature; see, for example, \cite{ArcherPTT, S-04}.  Notably, in our result, $\varepsilon(n)=n^{-\log n}$ is a \emph{negligible function} of the number of agents (decaying faster than any polynomial). This implies that as the number of agents grows, 
the incentive for untruthful reporting decreases rapidly.

\begin{restatable}{theorem}{loglognapx}
    \label{thm:loglog-apx}
    Consider the problem of allocating $m$ indivisible goods to $n$ \emph{additive} agents. There exists a {poly-time} $(1-\varepsilon(n))$-TIE randomized mechanism (with negligible $\varepsilon(n)=n^{-\log n}$) which is \emph{ex-ante proportional}, and $\Omega\left(\frac{1}{\log\log n}\right)$-$\MMS$ ex-post.
\end{restatable}

The mechanism that we present to prove Theorem~\ref{thm:loglog-apx} first samples a random ordering of the agents, and then behaves exactly like a distributional mechanism: it outputs a distributional allocation which is ex-ante proportional and is supported on at most $n\cdot m$ allocations. The relaxed incentive constraint is obtained by showing that, with high probability $\left(1 - \frac{1}{n}\varepsilon(n)\right)$, the sampled ordering allows for a TIE mechanism with the desired fairness guarantees. Only with probability $\frac{1}{n}\varepsilon(n)$ the ordering may be ``bad'', in which case we use a non-TIE mechanism instead, and so misreporting may lead to higher expected value (by a factor of at most $n$). As this happens rarely, the expected improvement from misreporting is negligible.

The mechanism in \Cref{thm:loglog-apx} is presented as a direct revelation mechanism, yet it can also be implemented in $(1-\varepsilon(n))$-dominant strategies (strategies with a small bounded loss) with agents only reporting a strict ordinal preference over the goods, coupled with just a small amount of additional cardinal information. The additional information that an agent is asked to supply is only the number of items out of her top $n$ items, that {do not} have value at least some threshold (some fraction of the $\MMS$), which we call the agent's \emph{deficiency}. In \Cref{prop:constant-with-public-defi}, we show that when agents' deficiencies are public knowledge, it is possible to design a TIE mechanism with $\frac{1}{4}$-$\MMS$ approximation ex-post.

Finally, we turn back to mechanisms that are exactly TIE. Recall that \Cref{thm:logn-apx} provides an $\MMS$ approximation factor of $\frac{1}{\harmonic{n-1}+2}$. While  $\frac{1}{\harmonic{n-1}+2}$ is strictly larger than $\frac{1}{n}$ {(the ratio in~\cite{BT-24})} for all $n \ge 4$, it is not so for $n=2,3$. We focus on the case of two agents and present a TIE mechanism that provides a $\frac{2}{3}$ approximation to the $\MMS$ ex-post.

\begin{restatable}{theorem}{twoagents}
\label{thm:two-agents}
    Consider the problem of allocating $m$ indivisible goods to two ($n=2$) \emph{additive} agents. There exists a distributional mechanism which is \emph{truthful-in-expectation (TIE)}, \emph{ex-ante proportional} (and thus also \emph{ex-ante envy-free}), and $\frac{2}{3}$-$\MMS$ ex-post. Moreover, this mechanism outputs, in polynomial time, an explicit representation of the distribution, which is supported on \emph{exactly two} allocations. 
\end{restatable}

All the positive results in this paper were presented as approximating the $\MMS$ ex-post. Yet, these results actually extend (with the exact same approximations) to hold for the truncated proportional share ($\TPS$), originally presented in \cite{BEF-22}. This work also proved that the $\TPS$ is at least as large as the $\MMS$, and is poly-time computable (unlike the $\MMS$). Consequently, our result for two agents (\Cref{thm:two-agents}) is the best possible, in the sense that there are allocation instances with two agents in which every allocation gives some agent at most $\frac{2}{3}$ of her TPS ex-post \cite{BEF-22}.

\subsection{Overview of our proofs}
\label{subsec:overview}

In this section we present an overview of the main ideas used in our proofs. Further details and refinements can be found in subsequent sections. Notation and terminology in this overview is simplified, and need not coincide with that used in later sections.

\subsubsection{The High Level Approach}

We design randomized TIE mechanisms that are ex-post fair, by implementing fractional mechanisms derived from universally-truthful mechanisms. We start with a brief review of fractional allocations, their implementation by distributions over (integral) allocations, and their potential use when agents have additive valuations.

A fractional allocation $\falloc$ assigns the goods fractionally to agents, where for each agent $i$ and item $x$, the fraction of item $x$ allocated to agent $i$ is $\f_{i,x}$. We require all fractions to be non-negative ($\f_{i,x} \ge 0$), and all goods to be fully allocated ($\sum_i \f_{i,x} = 1$ for every $x \in \goods$). Additive valuations extend naturally to fractional allocations, by $v_i(\falloc) = \sum_{x \in \goods} \f_{i,x}\cdot v_i(x)$.

To relate fractional allocations to integral allocations, one may interpret $\f_{i,x}$ as the probability with which agent $i$ receives item $x$. A distribution $\D$ over integral allocations (where allocation $\allock$ is selected with probability $\pk$) {\em implements} fractional allocation $\falloc$ (and then we also say that $\falloc$ is {\em induced} by $D$) if for every agent $i$ and item $x$ it holds that $\sum_{k \mid x \in \allock_i} \pk = \f_{i,x}$. Different distributions may implement the same fractional allocation, a fact that we make use of in our allocation mechanisms. 

For additive valuations, if $D$ implements $\falloc$, then $\sum_k \pk\cdot v_i(\allocki) = v_i(\falloc)$. In other words, the expected value that agent $i$ receives under $D$ exactly equals her value for the induced fractional allocation $\falloc$. Thus, for additive valuations, there is an equivalence between two notions.  One notion is that of a truthful allocation mechanism that outputs a fractional allocation. Here, every agent that wishes to maximize the value of the fractional bundle that she gets has a dominant strategy. The other notion is a truthful-in-expectation (TIE) allocation mechanism that outputs a distribution over allocations. Here, every agent that wishes the distribution to maximize the expected value of the random bundle that she gets has a dominant strategy.

Our template for designing TIE mechanisms with good ex-post guarantees is inspired by a template used previously~\cite{AY-13} for settings with divisible items. Initially, a truthful fractional mechanism is fixed. Then, given an input valuation profile, the template runs in two main steps:

\begin{enumerate}
    \item Using the truthful fractional mechanism, generate a fractional allocation $\falloc$. 
    \item Implement $\falloc$ by a distribution $\D$ over (integral) allocations, where $\D$ is supported only on allocations that offer good ex-post guarantees.
\end{enumerate}

For the second step, a known technique, referred to as {\em faithful implementations} in~\cite{BEF-22} {(similar techniques have been used earlier, for example, in the conference versions of~\cite{AzizFSV24}),} is a natural choice. It guarantees that each agent (with an additive valuation) will receive ex-post the same value that she receives in the fractional allocation, up to the value of one item. See \Cref{lem:faithful} for more details.

For the first step, a natural choice is to  use the proportional fractional mechanism, which, irrespective of the agents' reports, outputs fractional allocation $\f_{i,x} = \frac{1}{n}$ for every $i \in \agents$ and $x \in \goods$. It has the advantage of being truthful (as it does not depend {at all} on the valuations of the agents), and of giving each agent her proportional share. Given that it offers each agent at least her MMS ex-ante, one may hope that the proportional fractional allocation can be implemented by a distribution over allocations whose support only includes allocations that provide a good approximation to the MMS ex-post. Unfortunately, this is not true. It was already observed in~\cite{BEF-22} that standard techniques for achieving faithful implementations cannot guarantee more than $\frac{1}{n}$-MMS (not even if all agents have the same valuation function, a case referred to as {\em identical valuations}). The fact that they do guarantee at least $\frac{1}{n}$-MMS was noted in~\cite{BT-24}. We show (see Proposition~\ref{prop:n-approximation-impossibility}) that
in fact, no implementation of the proportional fractional allocation can guarantee more than $\frac{1}{n}$-MMS ex-post.

\subsubsection{Our Fractional Allocation Mechanism -- Cyclic-Unit-Quota} 

Our approach for designing the first step of our template is to obtain the fractional allocation $\falloc$ as the one induced by a distribution $\D$ over allocations, where each allocation in the distribution is obtained by a truthful deterministic mechanism. Such a randomized mechanism is referred to as {\em universally-truthful}. Allocations in the support of $\D$ might have very poor ex-post guarantees, as they suffer from the limitations of deterministic truthful allocation mechanisms~\cite{ABCM-17}.
However, we shall only use the fact that universally-truthful mechanisms are truthful-in-expectation, and consequently the induced fractional allocation is truthful (for expectation-maximizing agents with additive valuations). {We formalize this approach in \Cref{prop:universal-to-tie}.}
In our second step we find a different implementation of $D$, this time by allocations that have good ex-post guarantees.

Such an approach has been used before in the context of allocating a divisible good when valuations are piecewise linear~\cite{AY-13}. In that context, the second step (corresponding to implementing a fractional allocation by a distribution over integral allocations) involves no loss. That is, the ex-post guarantee is the same as the ex-ante one. In contrast, in our context of indivisible goods, this second step does involve a loss (as illustrated above for the proportional fractional allocation). This makes it much more challenging to find a fractional allocation that on the one hand, can be induced by a distribution over truthful deterministic mechanisms, and on the other hand, can be implemented by allocations that have good ex-post guarantees. 

The class of deterministic truthful allocation mechanisms that we use is the so-called {\em serial-quota mechanisms}. In these mechanisms,  there are $n$ rounds, and each round $r$ has a quota $q_r$, with $\sum_r q_r \le m$. There is an arbitrary order over the agents. In each round $r$, the $r$-th agent in the order selects $q_r$ of the remaining items, where a dominating strategy is to select a set of $q_r$ items of highest value. The remaining $m - \sum_r q_r$ items are all given to agent $n$. We focus on the ``unit-quota'' case, in which we fix $q_r = 1$ for all $r \le n$. We refer to the $n$ selected items as the \emph{prefix of the allocation}, and to the $m - n$ additional items given to agent $n$ as the \emph{suffix of the allocation}.

The ex-post guarantee of any serial quota mechanism is poor, not more than $O(\frac{1}{m-n})$~\cite{ABM-16}. However, by taking the $n$ cyclic shifts of a given ordering (\Cref{def:cyclic-shifts}), we get $n$ different ``unit quota'' allocations. The uniform distribution over them induces a fractional allocation, that we refer to as the {\em Cyclic-Unit-Quota} allocation, or CUQ. It is not difficult to see that the value of CUQ to each agent is at least her proportional share {(\Cref{lem:cuq-tie})}. 

\subsubsection{Overview of Proof of Theorem~\ref{thm:logn-apx}}
\label{sec:overviewLog}

Recall that Theorem~\ref{thm:logn-apx} applies to ordinal mechanisms for agents with additive valuations (see the discussion after the statement of that theorem). The input provided by each agent $i$ is a total order over the items. An order is {\em consistent} with $v_i$ if for every two items $e$ and $e'$, if $e$ precedes $e'$ in this order then $v_i(e) \ge v_i(e')$. Somewhat abusing terminology, we say that agent $i$ is truthful if she provides an order that is consistent with $v_i$. Clearly, a consistent order suffices in order to implement a dominant strategy in any serial quota mechanisms. Importantly, it also suffices in order to produce faithful implementations of fractional allocations. {(This is a known fact. See \Cref{lem:faithful}).} 

In the proof of Theorem~\ref{thm:logn-apx} we use CUQ as the fractional allocation in the first step of our template. This guarantees the TIE property, and also gives the desired proportional ex-ante guarantee. It remains to describe the second step of our template, implementing CUQ as a distribution over allocations that offer good ex-post guarantees. We do this in three steps (which are the content of \Cref{sec:gen-imp} and \Cref{sec:logn}): 

\begin{enumerate}
    \item {\bf The main step.} 
    We replace each of the $n$ unit quota allocations $\allocj$ that make up CUQ by a  fractional allocation $\dfa{j}$, so that the fractional allocation induced by averaging these $n$ fractional allocations remains equal to CUQ. Moreover, each $\dfa{j}$ satisfies two properties. 
    \begin{enumerate}
    \item Each $\dfa{j}$ coincides with the corresponding $\allocj$ on the prefix of $\allocj$. Consequently, only items in the suffix of $\allocj$ are allocated fractionally.
    \item Each $\dfa{j}$ offers relatively high value to each agent. That is, $v_i(\dfa{j}_i) \ge \frac{1}{\harmonic{n-1}+1}\cdot \MMS_i$.
    \end{enumerate}
    \item {\bf Implementation.} We implement each $\dfa{j}$ by a distribution $D^{(j)}$ using the faithful implementation technique {(see \Cref{lem:faithful})}.
    The combination of properties 1(a) and 1(b) above, together with properties of faithful allocations, can be used to show that every allocation in the support of $\Dj$ guarantees $\frac{1}{\harmonic{n-1} + 2}$-MMS. Thus, the overall distribution that implements CUQ (that is supported on the allocations of all distributions $\Dj$, for $1 \le j \le n$) satisfies the desired ex-post guarantees.
    \item {\bf Support reduction.} Using well known principles (the concept of {\em basic solutions} for linear programs), we reduce the support of the overall distribution to a size of at most $m(n-1) + 1 \le n\cdot m$. {(See Lemma~\ref{lem:support-reduction}).}
\end{enumerate}

We provide more details on the main step above. 
First, it replaces each $\allocj$ by an initial fractional allocation $\falloc^{(j,1)}$ that guarantees each agent at least $\frac{1}{\harmonic{n}}$-MMS. Afterwards it modifies each  $\falloc^{(j,1)}$ to $\falloc^{(j)}$, so as to achieve the additional property that for every agent $i$ and item $x$,  $\frac{1}{n} \sum_j \falloc^{(j)}_{i,x} = CUQ_{i,x}$. This is done with a small loss in the guarantees: each $\falloc^{(j)}$ is a {$\frac{1}{\harmonic{n-1} + 1}$}-MMS partial fractional allocation. 

{\bf The construction of $\falloc^{(j,1)}$ (for $1 \le j \le n$).} 
Recall that $\allocj$ is generated using a unit quota mechanism with some ordering $a^{(j)}_1, \ldots, a^{(j)}_n$ over the agents. We keep the prefix of this allocation unchanged. Suppose now that an agent $a^{(j)}_i$ gets a $\frac{1}{n-i+1}$ fraction of each item in the suffix. Then the total value that she gets from her item in the prefix and her fractions in the suffix is at least her MMS. This is because, excluding the first $i-1$ items of the prefix, the remaining items still have total value at least $(n-i+1)$-MMS for her. Of this value, she secures at least a fraction of $\frac{1}{n-i+1}$. For the prefix this is true because she gets the best of the $n - i + 1$ items remaining in the prefix. For the suffix this is true because she gets a $\frac{1}{n-i+1}$ fraction of each item.

If we were to give for every $i$ a fraction of $\frac{1}{n-i+1}$ of each item in the suffix to agent $a^{(j)}_i$, then the total amount of fractions allocated from each item in the suffix would be  $\sum_{i=1}^n \frac{1}{n-i+1} = H_n$. As a fractional allocation should contain one unit of each item, we instead give each agent $a^{(j)}_i$ a fraction of $\frac{1}{H_n}\cdot \frac{1}{n-i+1}$ of each item in the suffix. This serves as $\falloc^{(j,1)}$, and clearly gives each agent at least $\frac{1}{H_n}$-MMS.

{\bf The need to modify the fractional allocations $\falloc^{(j,1)}$.} In $\allocj$, agent $a^{(j)}_n$ receives all items in the suffix. In $\falloc^{(j,1)}$, agent $a^{(j)}_n$ keeps only a $\frac{1}{H_n}$ fraction of each such item, and gives up the rest. Also, for other $k \not= j$, agent $a^{(j)}_n$ gets in $\falloc^{(k)}$ fractions of items that were in the suffix of $\allock$. We need the contributions to agent $a^{(j)}_n$ of all suffixes of the $\falloc^{(k)}$ (for $1 \le k \le n$) to exactly match the contents of the suffix of $\allocj$, so that the induced fractional allocation will equal the original CUQ. However, this might not hold, for two reasons. 

One reason is that a suffix for some $\allock$ might contain an item $x$ that is not in the suffix of $\allocj$. In this case, $a^{(j)}_n$ should not receive any fraction of $x$. To address this, we modify $\falloc^{(k)}$ by simply changing the fraction of $x$ received by $a^{(j)}_n$ to~0. This temporarily causes a {\em deficit} in $\falloc^{(k)}$, leaving item $x$ not fully allocated in $\falloc^{(k)}$. This deficit will be filled in later.

The other reason is that a suffix for some $\allock$ might not contain an item $y$ that is in the suffix of $\allocj$. In this case, $a^{(j)}_n$ gets from the combination of all $\falloc^{(k)}$ ($1 \le k \le n$) less than one unit of $x$, causing a deficit for agent $a^{(j)}_n$, compared to $CUQ_{a^{(j)}_n, y}$. 

The construction of the $\falloc^{(k)}$ starts with all the $\falloc^{(k,1)}$, and fills up all their deficits. This can be done in a relatively straightforward way, by finding an appropriate fractional matching between deficits of agents and deficits of fractional allocations. {(See Section~\ref{subsec:completing-partial}).} The end result is a collection of fractional allocations, $\{\falloc^{(k)}\}_{k \in [n]}$ that implements CUQ.

{\bf Analyzing the value offered by each $\falloc^{(k)}$.} From the point of view of agent $a^{(j)}_n$, the loss of value in $\falloc^{(k)}$ compared to $\falloc^{(k,1)}$ is limited to the loss caused by not getting fractions from those items $x$ that are in the suffix of $\allock$ but not in the suffix of $\allocj$. Let $\ell$ be the rank of $a^{(j)}_n$ in the order that was used in the unit quota mechanism that generated $\allock$. Then the fact that this order is a cyclic shift of that used for $\allocj$, together with the assumption that each of the unit quota mechanisms was executed using the same total order over items provided by each agent, can be used in order to infer that the number of such items $x$ is at most $n - \ell$ (See Lemma~\ref{lem:union-picking-size}). In $\falloc^{(k,1)}$, the fraction of each item given to $a^{(j)}_n$ is $\frac{1}{n - \ell + 1}$. Consequently, in modifying $\falloc^{(k,1)}$ to $\falloc^{(k)}$, agent  $a^{(j)}_n$ lost value equivalent to {less than} one item from the suffix of $\allock$. But as $a^{(j)}_n$ holds in $\falloc^{(k)}$ an item from the prefix of $\allocj$ (which has value at least as high as any item in the suffix), this loss has relatively small effect. It changes the fraction of the MMS received by $a^{(j)}_n$ from at least $\frac{1}{\harmonic{n}}$ in  $\falloc^{(k,1)}$ to at least {$\frac{1}{\harmonic{n-1} + 1}$} in $\falloc^{(k)}$. As stated above, the faithful implementation step incurs an additional loss equivalent to at most one item, reducing the final ex-post integral bound to $\frac{1}{\harmonic{n-1} + 2}$.

\subsubsection{Overview of Proof of Theorem~\ref{thm:loglog-apx}}

For the purpose of proving Theorem~\ref{thm:loglog-apx}, we use the following terminology. 
For a parameter $\alpha > 0$, we say that an item $x$ is $\alpha$-{\em sufficient} for agent $i$ if $v_i(x) \ge \alpha \cdot \MMS_i$. Unless stated otherwise, $\alpha$ will be set to the desired ex-post approximation ratio for the MMS, so that if an agent receives an item that is $\alpha$-sufficient for her, she already reaches her promised ex-post guarantee. 
The $\alpha$-{\em deficiency} of  agent $i$ is the number of items, among the first $n$ items of her total order, that are not $\alpha$-sufficient.

The $\alpha$-deficiencies of the agents can be very useful in implementing the CUQ fractional allocation as a distribution over allocations that are $\alpha$-MMS. 

The allocation mechanisms in this section ask each agent $i$ to report a total order among the items (consistent with $v_i$), and also to report her $\alpha$-deficiency. 

{\bf Non-truthful constant approximation.} To illustrate the usefulness of $\alpha$-deficiency, we first show that given a total order over items from each agent (consistent with her valuation function) and her true $\alpha$-deficiency for $\alpha = \frac{1}{4}$, an allocation algorithm can output a $\frac{1}{4}$-MMS allocation {(See \Cref{prop:constant-with-public-defi})}. Recall that with only ordinal information (that is, without knowledge of the $\alpha$-deficiency), the best approximation ratio is roughly $\frac{1}{\ln n}$.

The allocation algorithm first orders the agents based on their $\alpha$-deficiencies, from smallest to largest. We refer to this order as the {\em deficiency order}, to the agents in this order as $a_1, \ldots, a_n$, and to the deficiencies of the agents as $d_1 \le \cdots \le d_n$. Run the cyclic unit quota mechanism (with the $n$ cyclic shifts over the deficiency order), obtaining $n$ allocations $\alloc^{(1)}, \ldots, \alloc^{(n)}$. Recall that in Section~\ref{sec:overviewLog}, these $\allocj$ were first transformed to fractional allocations $\falloc^{(j,1)}$, where the fractional allocations involved a scaling factor of $\frac{1}{\harmonic{n}}$. Here, we show that using the deficiency information, the scaling factor can be improved to $\frac{1}{2}$ {(See \Cref{lem:ordered-weight})}.

Consider $\alloc^{(1)}$, the unit quota allocation that uses the deficiency order (without permuting it cyclically). Recall that in $\falloc^{(1,1)}$ we wish to fractionally allocate the suffix of $\alloc^{(1)}$ to the agents, such that each agent gets a good approximation of her MMS. 

Let $t$ be the smallest index for which $t \ge n -d_t+1$. Every agent $a_i$ with $i < t$ will receive in the prefix of $\alloc^{(1)}$ an item that is $\alpha$-sufficient for them. For each such agent, we say that she has no demand on the suffix. For every agent $a_i$ with $i \ge t$, the total value of the prefix is at most $n-d_i$ for the first $n-d_i$ items, plus $\alpha \cdot d_i$ for the remaining items. Hence, the total value of the suffix is at least $(1 - \alpha)\cdot d_i$. By giving them $\frac{1}{d_i}$ of each item in the suffix, agent $i$ is guaranteed a value of at least $(1 - \alpha)$.  For each such agent $i$, we say that she has a demand of $\frac{1}{d_i}$ on the suffix. As there are only $n-t+1$ agents with index $i\ge t$, each with a demand of $\frac{1}{d_i}\le \frac{1}{d_t}\le\frac{1}{n-t+1}$, the total demand on the suffix is at most $\frac{1}{n-t+1} \cdot (n-t+1) = 1$. Because the total demand is at most $1$, there is no need to scale the demands of agents at all.

For $\allocj$ with $j > 1$, an argument similar to the above shows that the total demand on the suffix is at most~2. This uses the fact that the order induced by a cyclic shift of the deficiency order can be partitioned into two segments, each with increasing deficiencies. Consequently, a scaling factor of $\frac{1}{2}$ suffices. Each agent either gets value at least $\alpha$ from the prefix of $\falloc^{(j,1)}$, or value at least $\frac{1 - \alpha}{2}$ from the suffix.

Using arguments similar to those of Section~\ref{sec:overviewLog}, by modifying the $\falloc^{(j,1)}$ to $\falloc^{(j)}$, the value that agents get from the suffix becomes at least $\frac{1 - 2\alpha}{2}$.

Finally, doing the faithful implementation of each $\falloc^{(j)}$, those agents that got a value of at least $\alpha$ from the prefix maintain this value. Those who get value at least $\frac{1 - 2\alpha}{2}$ from the suffix, maintain this value up to one item. However, as they also get an item from the prefix, they maintain at least $\frac{1 - 2\alpha}{2}$. With $\alpha = \frac{1}{4}$, in either case each agent gets at least $\frac{1}{4}$-MMS.

Observe that in the above argument, deficiency was used in two ways. One is to determine the order over agents, the deficiency order. The other is to determine the demand requirements of each agent $i$, when replacing $\allocj$ by $\falloc^{(j)}$ (the demand becomes either $\frac{1}{d_i}$ or $0$, depending on the location of the agent in the order). The first aspect affects the choice of fractional allocation CUQ. As CUQ determines the ex-ante guarantees, expectation maximizing agents might find it useful to misreport their deficiencies, which is incompatible with our goal of having TIE mechanisms. The second aspect has no effect on the ex-ante guarantees (it only improves the ex-post ones), and so it can be incorporated in a TIE mechanism. In our proof of Theorem~\ref{thm:loglog-apx}, we extensively use the second aspect, but make only minimal use of the first aspect. 

{\bf A TIE mechanism.} We introduce a notion of a {\em weight} of an ordering $\pi$ over agents. This weight is equal to the sum of demands that agents have on the suffix (what fraction of the suffix each agent needs in order to attain its MMS), for the allocation generated using $\pi$ in a unit quota mechanism. In the proof of Theorem~\ref{thm:logn-apx} this weight was $\harmonic{n}$, whereas for any cyclic shift of the deficiency order, the weight is at most~2.  Using  a parameter $w = \Theta(\log\log n)$ that serves as a threshold on weights, our allocation mechanism for the proof of Theorem~\ref{thm:loglog-apx} is as follows.

\begin{enumerate}
\item Ask each agent $i$ to report a strict ordering over the items (consistent with her $v_i$) and her $\alpha$-deficiency $d_i$, for $\alpha = \frac{1}{w+2}$.
    \item Pick a uniformly random ordering $\pi$ over the agents.
    \item If the weight of $\pi$ (as computed from the reported demands) and the weight of each of $\pi$'s cyclic shifts is at most $w$, then use $\pi$ in a mechanism similar to that of Theorem~\ref{thm:logn-apx}, but using the reported deficiencies to set the demands of agents on suffixes of allocations {(See \Cref{lem:general-weight})}.  This gives a distribution over allocations that is ex-ante proportional, and ex-post gives each agent at least $\frac{1}{w +2}$-MMS (Here $w$ plays the role that $\harmonic{n}$ plays in the proof of Theorem~\ref{thm:logn-apx}).
    \item Else, use the deficiency order instead of $\pi$. This gives a distribution over allocations that is ex-ante proportional, and ex-post gives each agent at least $\frac{1}{4}$-MMS. 
\end{enumerate}

The above mechanism satisfies the desired ex-ante and ex-post guarantees, but is not TIE. By a careful choice of what to report as her deficiency $d_i$, an agent might be able to affect whether step~3 or step~4 are executed, and these steps potentially offer her different ex-ante guarantees. However, it is $(1 - \varepsilon(n))$-TIE for every small $\varepsilon(n)$ ($O(n^{-\log n})$, or even smaller, depending on the leading constant in the choice of $w = c\log\log n$). This is a consequence of {\Cref{lem:random-ordering},} 
which shows that for every sequence of reported demands $d_1, \ldots, d_n$, with overwhelming probability, $1 - \frac{\varepsilon(n)}{n^2}$, the weight of a random ordering $\pi$ is below $w - 1$. As different reports of $d_i$ can affect the weight by at most~1, with overwhelming probability, the report $d_i$ has no effect at all on the ex-ante properties of the resulting distribution over allocations. (The term $\frac{1}{n^2}$ in $\frac{\varepsilon(n)}{n^2}$ allows us to take a union bound over all cyclic shifts of $\pi$, and to take into account that switching between the allocations of step~3 and~4 might affect the ex-ante guarantee of the agent by a large factor, but not more than $n$).

\subsubsection{Overview of Proof of Theorem~\ref{thm:two-agents}}
\label{sec:overviewTwoAgents}

For the case of two agents, the fractional allocation that we implement as a distribution over allocations is again the CUQ fractional allocation. For two agents, CUQ is of one of two simple forms, each leading to a different implementation.

If the two agents disagree on the identity of the best item, then CUQ gives each of them her top item in full, and for each remaining item, it gives each agent half the item. The implementation needs to only address the remaining items. Given that all fractions in the fractional allocation are half and $n=2$, the standard approach for designing faithful implementations simply results in the partition of this set into two disjoint sets $S$ and $T$, such that for each agent, the difference in value between the two subsets is at most the value of one item from these subsets. The distribution over allocations contains only two allocations that each have probability $\frac{1}{2}$, where in one allocation agent~1 gets $S$  and agent~2 gets $T$ (in addition to their respective top item), and in the other the roles of $S$ and $T$ are reversed.  
A simple case analysis then shows that every agent receives at least $\frac{2}{3}$-MMS in each allocation.

If the two agents agree on the identity of the first item, then CUQ is simply the proportional fractional allocation, giving half of each item to each agent. In this case, we show that $\goods$ can be partitioned into two disjoint sets $S$ and $T$ that each give each agent a value of at least $\frac{2}{3}$-MMS. Thus, in this case the proportional fractional allocation is implemented as a distribution selecting each of the two allocations $(A_1 = S, A_2 = T)$ and $(A_1 = T, A_2 = S)$ with probability half.
However, the above partition into $S$ and $T$ is not simply the result of the standard faithful implementation technique. Rather, the contents of $S$ and $T$ are determined by a case analysis. First, each agent $i$ is asked whether the top item $e_1$ satisfies $v_i(e_1) \ge \frac{2}{3} \cdot \MMS_i$. If both agents answer {\em yes}, then we set $S = \{e_1\}$ and $T = \goods \setminus \{e_1\}$. If both agents answer {\em no}, then we use the standard faithful implementation. If one agent (say, agent~1) answers {\em yes} and the other agent (agent~2) answers {\em no}, then agent~2 is required to supply additional information about her valuation (beyond the order over items and her one bit {\em no} answer). 

The construction of the partition is as follows. {(\Cref{fig:two-agent-hard-case} provides a visual representation of the construction).} Consider the order of items supplied by agent~1 (starting with $e_1$). Partition the sequence of items into consecutive triples of items. Starting with the first triple, from each triple place in $S$ the item of highest value for agent~2 (this information is available from the order supplied by agent~2). The process is stopped after $s$ triples, where $s$ is determined from $v_2$ so that at that point, $\frac{2}{3} \cdot \MMS_2 \le v_2(S) \le \frac{4}{3} \cdot \MMS_2$. All remaining items are placed in $T$. Thus, also $v_2(T) \ge \frac{2}{3} \cdot \MMS_2$. As to agent~1, $v_1(S) \ge \frac{2}{3} \cdot \MMS_1$ because $e_1 \in S$, and $v_1(T) \ge \frac{2}{3} \cdot \MMS_1$ because $v_1(T) \ge \frac{2}{3} \cdot v_1(\goods \setminus \{e_1\})$.

\subsubsection{Some Additional Comments}
\label{sec:overviewComments}

Here we present some additional comments that were omitted from the overview above, so as not to make it more complicated than it already is.

\begin{itemize}
    \item All ex-post guarantees  hold (and the proofs apply as is) not just with respect to the MMS, but also with respect to the truncated proportional share (TPS), that is always at least as large as the MMS. 

    \item None of our mechanisms require agents to send their full valuation functions. They only require ordinal information (a strict order over the items, consistent with the valuation function of the agent), and relatively few additional bits of information. For Theorem~\ref{thm:logn-apx}, no additional information is needed. For Theorem~\ref{thm:loglog-apx}, only the deficiency {of every agent} is needed, roughly $\log n$ bits per agent. In fact, instead of sending the exact deficiencies, agents may send their deficiencies rounded up to the nearest power of~2. This takes only roughly $\log\log n$ bits per agent. The mechanism can be adapted to use this approximate information, but in the process, it loses a factor of~2 in the MMS approximation. (Further details are omitted).
    
    For Theorem~\ref{thm:two-agents}, each agent needs to send one additional bit specifying whether her top item has value at least $\frac{2}{3}$-MMS. In addition, one of the agents (say, agent~2) might need to send an integer $s$ in the range $[1, m]$ (determining in which triple to stop). If we allow communication to take two rounds and be public, then a suitable $s$ can be reported in the second round, using only $\log\log m + O(1)$ bits. (See Section~\ref{app:twoAgentsComplexity} for details. {It is also shown that if we insist on a one-round mechanism, then $O((\log m)^2)$ bits suffice.}) Having two rounds in the mechanism does not affect the TIE property, because the fractional allocation which determines the ex-ante value that agents get is fixed after the first round of communication.

\item All our mechanisms run in time polynomial in their inputs. The inputs reported to the mechanism, as explained above, are independent of the number of bits required to represent values of the valuation function. The total size of all inputs is $O(nm\log m)$ (mainly used to specify $n$ strict orders over the $m$ items). Consequently, the running time is polynomial in $n$ and $m$. 

Also, each expectation maximizing agent $i$ can compute the reports of her dominant strategy (or $(1 - \varepsilon(n))$-dominant strategy, in Theorem~\ref{thm:loglog-apx}) in time polynomial in the description of the additive valuation $v_i$. Here it is useful that our mechanisms approximate the TPS (and the approximation ratio automatically applies to the MMS as well). Unlike the MMS, the TPS can be computed in polynomial time, implying that the deficiencies (and the parameter $s$ in the proof of Theorem~\ref{thm:two-agents}) can also be computed in polynomial time. 

\end{itemize}

\subsection{Additional Related Work}
\label{sec:related}

The problem of fair allocation has been studied extensively from both algorithmic and game-theoretic perspectives. The reader is referred to the books \cite{BT-96,M-04} and survey \cite{AABFLMVW-23} for background. In this section we discuss some papers that are most related to ours and that we did not discuss already.

\subsubsection{Truthful Fair Mechanisms}
A key line of work in indivisible goods seeks to allocate items fairly in strategic environments, when agents' preferences are private information. We focus on papers on truthful mechanisms most related to our work.
For deterministic mechanisms, there are multiple impossibility results for achieving truthfulness in conjunction with most desirable fairness notions such as minimum-envy, EF1, and MMS \cite{LMMS-04, GKKK-09, ABM-16, ABCM-17}. Stronger impossibility results have been proven with additional restrictions added to the mechanisms such as non-bossiness, neutrality, and Pareto-efficiency \cite{P-00, P-01, H-09, HL-15, BMM-25}.

To circumvent these deterministic impossibilities, a prominent line of research \cite{AFCMM-15, BT-24, GNS-25} utilizes randomization, relaxing truthfulness to \emph{truthfulness-in-expectation (TIE)}, which ensures that truth-telling maximizes an agent's expected utility over the randomness of the mechanism. This is equivalent to truthfulness for risk-neutral agents. Most relevant to our work, Bu and Tao \cite{BT-24}  design TIE mechanisms with desirable fairness properties, including (not as the main result in that paper), a TIE mechanism that achieves PROP1 and $1/n$-MMS ex-post. This TIE mechanism is a faithful implementation of the proportional fractional allocation, which was previously noted to be PROP1~\cite{AzizFSV24, BEF-22}, and~\cite{BT-24} note that it is also $\frac{1}{n}$-MMS.

As noted and discussed in \cite{BF-22}, share-based fairness guarantees (such as an ex-post MMS approximation) that an allocation mechanism offers to an agent hold whenever the agent herself is truthful, regardless of whether other agents are truthful. Consequently, truth-telling (in a non-truthful mechanism) can be referred to as a {\em safe strategy}, meaning that the agent can ``safely'' follow this strategy if she is satisfied with the guarantees offered by the mechanism.  Safe strategies are a weaker notion than dominant strategies.  However, their existence implies that in every Nash equilibrium of the allocation mechanism, all agents get their guarantees. Reaching an equilibrium might be hard, so we construct dominant-strategies mechanisms (which do not have this drawback).

There is also work on truthful mechanisms in the  setting where goods are divisible. Most related to our work is Aziz~and~Ye \cite{AY-13}, which inspired our approach for obtaining TIE mechanisms from universally truthful ones, as discussed in Section~\ref{subsec:overview}. The mechanism in~\cite{AY-13} is based on a serial quota mechanism, with all quotas being $\frac{m}{n}$. For $m \ge 2n^2$, there are input instances (a variation on those appearing in the proof of Proposition~\ref{prop:n-approximation-impossibility}) for which no implementation of the resulting fractional allocation guarantees more than $\frac{1}{n}$-MMS. In contrast, we choose the quotas to all be~1 for all agents but the last, which allows us to give much stronger ex-post guarantees.

An elegant class of truthful mechanisms for fractional allocations, based on \emph{competitive scoring rules}, was introduced by Freeman et al.~\cite{FreemanWVP24}. A natural question is whether for the indivisible setting, these fractional allocations can be implemented on a support of allocations with desirable ex-post fairness guarantees (to get a BoBW result with a TIE mechanism). For two agents, this framework encompasses all non-wasteful truthful mechanisms, and thus includes our CUQ fractional mechanism. For $n \ge 3$, however, the CUQ mechanism does not belong to this class. While CUQ is ordinal, scoring-rule-based mechanisms can utilize cardinal information, so they have the potential of becoming the basis for designing TIE mechanisms with better ex-post fairness guarantees. However, adapting them to the indivisible goods setting presents a complex challenge. The design of our mechanism benefits greatly from the CUQ's derivation from universally-truthful mechanisms, a property that structurally facilitates implementation over a support of integral allocations. Since scoring-rule-based fractional mechanisms do not inherently share this property, the decomposition of their fractional allocations over a support of fair integral allocations seems challenging, and remains a direction for future research.

Some works relax exact truthfulness to obtain results with strategic agents. Mechanisms with no obvious manipulation were studied in \cite{PV22}. In \Cref{thm:loglog-apx}, we relax TIE to $(1 - \varepsilon(n))$-TIE. Similar in spirit relaxations of truthfulness, referred to as the
\emph{incentive ratio}, were previously studied in the context of fair allocation~\cite{XiaoL20, TaoY24}, though involving incentive ratios that are bounded away from~1.

\subsubsection{Share-based ex-post Fairness}
Our work focuses on share-based fairness, using both the Maximin Share ($\MMS$) \cite{B-11} and the Truncated Proportional Share ($\TPS$) \cite{BEF-22}.  

A significant body of work  {attempts to determine the best approximation ratio for the MMS that allocations can guarantee when valuations are additive} \cite{KPW-18, BK-20, GHSSY-18, GMT-19, 
AMNS-17, GT-21, AGST-23, AG-24, HKSS-25}.  Currently, {it is known that $\frac{7}{9}$-$\MMS$ allocations always exist}~\cite{HZ-25}. On the other hand, it was shown in \cite{FST-21} that for $n=3$ agents, {there are allocation instances in which no allocation offers more than} $\frac{39}{40}$-MMS. Another line of research examines the ordinal setting \cite{ABM-16,HS-21}, where the algorithm can only access the agents' ordering over the items, while the exact cardinal (additive) valuations of the agents remain unknown. Allocations based only on ordinal information cannot guarantee better than $1/\harmonic{n}$ approximation to the $\MMS$ ~\cite{ABM-16}. It was shown in~\cite{HS-21} that there exists an ordinal algorithm that obtains $1/(2\harmonic{n})$ approximation to the $\MMS$. Our \Cref{thm:logn-apx} improves upon this in two respects, first being TIE, and second, improving the approximation ratio to  $1/(\harmonic{n-1}+2)$. (Removing the TIE requirement, our bound improves to $\frac{1}{\harmonic{n} + 1}$. Details omitted).

The $\TPS$ was studied both in the context of equal entitlements \cite{BEF-22} and unequal entitlements \cite{BabaioffEF24, BF-25}.
Our proofs make use of the fact that TPS is poly-time computable and is always at least as large as the MMS. 

\subsubsection{Best-of-Both-Worlds and Fractional Mechanisms}

A recent line of work in fair division focuses on designing algorithms in the framework of \emph{Best-of-Both-Worlds (BoBW)}, which aim to simultaneously achieve strong ex-ante and ex-post fairness guarantees for randomized allocations \cite{A-19, FSV-20, A-20, BEF-22, HSV-22, AGM-23, FMNP-24, AAGW-15, HPPS-20, BEF-21, AzizFSV24, BF-26}. Aziz \cite{A-19} studied various settings, including voting, allocation, and matching. The  term ``Best-of-Both-Worlds''  appears to have been introduced in this context by Freeman, Shah, and Vaish \cite{FSV-20}, designing a polynomial-time algorithm that is envy-free (EF) ex-ante and envy-free up to one good (EF1) ex-post. Aziz \cite{A-20} has presented a simpler combinatorial algorithm with the same  guarantees, making use of the Birkhoff-von-Neumann algorithm to implement the fractional allocation ex-post. Babaioff, Ezra, and Feige  \cite{BEF-22} extend the Best‑of‑Both‑Worlds framework to include share‑based fairness, designing a polynomial-time algorithm that is proportional ex-ante and $1/2$-$\TPS$ ex-post. Babaioff~and~Frosh~\cite{BF-26} have designed refined Best-of-Both-Worlds algorithms for the special cases of two and three agents, obtaining, among other properties, proportionality ex-ante along with $9/10$-$\MMS$ ex-post for three agents.  While all the work is purely algorithmic, we obtain a BoBW result in a strategic setting, with a  mechanism that is truthful-in-expectation.

\subsection{Discussion}

Previous to our work, it seemed that there is a large price to pay in terms of the ex-post MMS approximation, if one compares the guarantees of TIE mechanisms with the best possible existential bounds. Known TIE mechanisms offered only a $\frac{1}{n}$ approximation to the MMS~\cite{BT-24}, whereas it is known that $\frac{7}{9}$-MMS allocations exist~\cite{HZ-25}.

Our results paint a very different picture.  For two agents we show a TIE mechanism that offers $\frac{2}{3}$-TPS, matching the best possible existential result. For ordinal mechanisms and general $n$, we show a TIE mechanism that offers $\frac{1}{\harmonic{n-1} + 2}$-MMS, nearly matching the best possible existential results  with only ordinal information~\cite{ABM-16}. Moreover, the notion of deficiency that we introduce, and our mechanisms that make use of it, suggest that much better approximations can be achieved by cardinal TIE mechanisms. 

With the insights obtained from our work, we find it reasonable to conjecture that there are TIE mechanisms that offer each agent at least a constant fraction of her MMS ex-post. Moreover, similarly to all the mechanisms presented in this paper, we may hope to achieve this simultaneously with guaranteeing the proportional share ex-ante.

\section{Preliminaries}
\label{sec:prelim}

We study the problem of fairly and truthfully allocating a set of $m$ items $\goods=\{1,\ldots,m\}$ to a set of $n$ strategic agents $\agents=\{1,\ldots,n\}$ with private preferences over the items. We assume that the preference of each agent $i\in\N$ can be expressed as an \emph{additive valuation function,} a function $\val_{i}:\poset{\M}\rightarrow\Real$ satisfying $\val_{i}(S)=\sum_{\ix\in  S}\val_{i}(\{\ix\})$ for every $ S\subseteq \M$, and $\val_{i}(\emptyset)=0$.
We furthermore assume that items are \emph{goods}, that is $\val_{i}(\ix) \geq 0$ for all $\ix\in\M$, where we write $\val_{i}(\ix)\coloneq\val_{i}(\{\ix\})$ for simplicity of notation. We denote by $\Vadd$ the set of all such valuations. A \emph{valuation profile} is a vector of $n$ valuations $\valprof\in \Vaddprofs$.

\begin{definition}
    A (deterministic) \emph{integral allocation} of the goods $\M$ is a vector of pairwise disjoint bundles $\alloc=(\allocset_1,\ldots,\allocset_{n})\in \left(2^\M\right)^{n}$ such that agent $i\in \agents$ receives the bundle $A_i$, and such that all goods are allocated ($\bigcup_{i\in\N}\allocset_i=\M$). We denote the set of all allocations of goods $\M$ among agents $\N$ by $\allocs$.
\end{definition}

For brevity, we often simply use the term \emph{allocation} to refer to integral allocations. We also consider fractional allocations, in which each good can be split fractionally among several agents.

\begin{definition}
    A \emph{fractional allocation} is a matrix $\falloc\in {[0,1]}^{n\times m}$ such that $\sum_{i\in\N}\fix=1$ for every item $\ix\in\M$. 
    We denote the set of all fractional allocations of goods $\M$ among agents $\N$ by $\fracallocs$. 
\end{definition}

A \emph{partial fractional allocation} is a vector $\falloc\in {[0,1]}^{n\times m}$ such that  $\sum_{i\in\N}\fix \leq 1$. 
Fractional allocations naturally arise from distributions over (integral) allocations.   

\begin{definition}
An (explicitly represented) \emph{distributional allocation} is a distribution $\D$ supported on $K$ allocations $\D = \{(\pk,\allock)\}_{k\in[K]}$, where for each $k\in [K]$, allocation $\allock\in \allocs$ has probability $\pk$ 
($\pk > 0 \ \forall k\in [K]$ and $\sum_{k\in [K]} \pk = 1$). We use $\supp{\D}$ to denote the \emph{support} of the distribution, that is, $\supp{\D}=\{ \allock \}_{k \in [K]}$.

Every distributional allocation \emph{induces} a fractional allocation $\fracD\in\fracallocs$ of the goods, which for each good $\ix\in \M$
allocates to each agent $i\in\N$  a fraction $\fracDix\in [0,1]$, which equals to the probability of $i$ receiving $\ix$ when the allocation is sampled from $\D$: 
    \[\fracDix\denote\prob[\alloc\sim\D ]{\ix\in \allocset_i}=\sum_{k=1}^K \pk\cdot 1 [\ix\in\allocki]\]
We denote the set of all distributional allocations of goods $\M$ among agents $\N$ by $\Delta(\allocs)$.
\end{definition}

We use $\PerN$ to denote the set of all permutations of the set of agents $\N$. Each permutation $\per\in\PerN$ represents an ordering of the agents, where for each agent $i\in\N$, we use $\per_i$ to denote the position of agent $i$ in the ordering. Inversely, the first agent is $\pinv_1$, continuing, up to $\pinv_n$.

Throughout the paper we use $\harmonic{n}$ to denote the $n$-th \emph{harmonic number}, that is, $\harmonic{n}=\sum_{i=1}^n\frac{1}{i}$.

{
In this paper we often focus on ordinal preferences. A strict ordinal preference is an ordering $\vec{e}=(e_1,\ldots,e_m)$ of all the goods in $\M$, where earlier goods are considered better. For an agent $i\in\N$ with additive valuation $\val_i\in\Vadd$, we say strict ordinal preference  $\ei=(\eiat{1},\ldots,\eiat{m})$ is \emph{consistent} with $\val_i$, if for every pair of goods $\ix,y\in \M$ such that $\val_i(x)>\val_i(y)$, item $x$ is placed before $y$ in $\ei$. We say a profile of strict ordinal preferences $(\en{1},\ldots,\en{n})$ is consistent with an additive profile $\valprof$, if ordering $\ei$ is consistent with its respective valuation $\val_i$, for every $i\in\N$.
}

\subsection{Mechanisms}
\label{subsec:prelim-mech}

We consider settings in which the valuation of each agent is \emph{private} information, known only to her. We seek to design \emph{dominant strategy incentive compatible (DSIC)} mechanisms, in which every agent has a dominant strategy, a strategy that is best irrespective of the behavior of the other agents. As usual, we assume that an agent that has a dominant strategy will play such a strategy.
 
We define several key types of mechanisms which we consider in this paper. To simplify our presentation, we present our DSIC mechanisms as \emph{direct-revelation} truthful mechanisms. In such mechanisms, each agent is requested to report a valuation, and truthful bidding is a dominant strategy. Notably, truthfulness is without loss of generality for direct revelation DSIC mechanisms, by the revelation principle \cite{M-81}. Notably, our results are stronger, as our mechanisms do not make use of all the information in the agents' reports, and so can be easily adapted into DSIC mechanisms for restricted strategy spaces in which agents need to report much less than the entire valuation functions (i.e., each only reports a strict ordinal preference over goods that is consistent with her valuation, and perhaps a small amount of additional cardinal information).

We start with deterministic integral mechanisms.

\begin{definition}
A \emph{deterministic (allocation) mechanism} for $n$ agents with valuations from class $\V$ is a function $\mech:\Vprofs\rightarrow \allocs$, that is, given a profile of $n$ valuation functions, mechanism $\mech$ outputs an allocation of $\M$ to the set of agents $\N$.
\end{definition}

We also consider mechanisms that output fractional allocations:
\begin{definition}
    \label{def:frac-alloc-mech} 
    A \emph{fractional allocation mechanism} for $n$ agents with valuations from class $\V$ is a function $\fracMech:\Vprofs\rightarrow \fracallocs$, that is, given a profile of $n$ valuation functions, mechanism $\fracMech$ outputs a fractional allocation in $\fracallocs$.
\end{definition}

For mechanisms that use randomization, we distinguish between mechanisms that only output a sample from an underlying distribution over allocations (for which we use the term \emph{randomized mechanisms}), and mechanisms that explicitly output a distribution over allocations (which we call \emph{distributional mechanisms}). Note that the latter can be used to create the former, by sampling from the output distribution, yet the reverse direction may be infeasible for polynomial time mechanisms, as the explicit distribution may be exponentially large.

\begin{definition} 
    A \emph{randomized} (allocation) mechanism for $n$ agents with valuations from class $\V$, is a function $\randMech:\Vprofs\rightarrow \Delta(\allocs)$, which maps every valuation profile to a distribution over allocations. 
    A randomized mechanism is typically implemented by mapping a valuation profile to a sample from the (implicit) underlying distribution. We use the term \emph{distributional} mechanism to refer to a mechanism which outputs an explicit representation of the output distributional allocation.
    
    Every randomized mechanism $\randMech$ \emph{induces} a fractional allocation mechanism $\fracMech$, which for every input profile $\valprof\in\Vprofs$ outputs the fractional allocation $\fracD$ induced by the distribution $\D=\randMech(\valprof)$. Inversely, the fractional allocation mechanism $\fracMech$ is \emph{implemented} by randomized mechanism $\randMech$ if $\randMech$ induces $\fracMech$.
\end{definition}

A distributional mechanism is fundamentally a deterministic function mapping inputs to distributional allocations; thus, if it operates in polynomial time, its output distribution must have a support of polynomial size. Furthermore, polynomial time distributional mechanisms induce fractional allocations which can also be computed in polynomial time. 
A distributional mechanism naturally induces a randomized mechanism, by sampling from its output distribution, so we sometimes refer to a distributional mechanism as a randomized one. That is, for claims regarding the properties of the distribution generated by a randomized or a distributional mechanism, we can treat both randomized and distributional mechanisms the same.

Distributional mechanisms have several advantages over randomized mechanisms. First, distributional mechanisms create transparency, as they allow each agent to verify her ex-ante fairness guarantees (e.g., proportionality) are satisfied. This is in contrast to randomized mechanisms, where each agent only views a single realized sample of the distribution. Second, distributional mechanisms give the agents freedom to decide how to sample the random allocation (e.g., they can toss coins together), which can increase trust that realized outcome was sampled truly randomly. Finally, the computation of distributional mechanisms is purely deterministic, so after the valuations are reported each agent can individually verify the outcome distribution is correct. See \cite{BF-26} for a more detailed discussion regarding this issue.

Two of our results  (\Cref{thm:logn-apx}, \Cref{thm:two-agents}) are mechanisms satisfying the more demanding requirement of being distributional mechanisms. Our third result (\Cref{thm:loglog-apx}) can be viewed as a ``hybrid'' mechanism, as it begins by sampling a random ordering over the agents (of which there are too many to explicitly list), but still outputs a distributional allocation corresponding to the sampled ordering.

Risk-neutral agents seek to strategically maximize their expected value from the mechanism. Thus, given a distributional allocation $D = \{ (\pk, \allock) \}_{k \in [K]}$ the value $\val_i(D)$ that a risk-neutral agent $i\in\N$ with an additive valuation $\val_i\in\Vadd$ assigns to $D$ is the expected value of $i$ according to $D$, that is, $\val_i(D) = \sum_{k=1}^K \pk\cdot \val_i (\allocki)$.  
Similarly, given a fractional allocation $\falloc$, the \emph{value} that agent $i\in\N$ with an additive valuation $\val_i\in\Vadd$ assigns to $\falloc$ is $\val_i(\falloc)=\sum_{\ix\in\M}\fix\cdot\val_i(\ix)$. 
Thus, the value function extends linearly from integer allocations to distributional and fractional ones.

\begin{definition}\label{def:truthful}
Consider any mechanism (deterministic, fractional, or distributional) $\mech$ for $n$ agents with valuations from class $\V$. We say that $\mech$ is \emph{truthful} (in dominant strategies) if for every valuation profile $\valprof\in\Vprofs$, every agent $i \in \N$, and every alternative valuation $\valp_i\in\V$, it holds that
\[
v_i\!\left(\mech\!\left(\val_i,\, \val_{-i}\right)\right)
\ge
v_i\!\left(\mech\!\left(\val'_i,\, \val_{-i}\right)\right).
\]
When $\mech$ is randomized or distributional, the values above are interpreted in expectation over the randomness of the mechanism, and the mechanism is said to be \emph{truthful-in-expectation (TIE)} instead.
\end{definition}

We note that truthful deterministic mechanisms are \emph{dominant strategy incentive compatible (DSIC)} as each agent's dominant strategy is simply to report their own valuation. On the other hand, TIE randomized or distributional mechanisms are DSIC only when agents are expectation-maximizers, that is, they seek to maximize their expected value, which is the standard assumption in the field of mechanism design. A randomized (or a distributional) mechanism that is created by randomizing over deterministic truthful mechanisms is called a \emph{universally-truthful} mechanism. Universally-truthful mechanisms are truthful even without the assumption of risk neutrality.

Finally, we remark that randomization here is only in the allocation mechanism, and all incentive notions we consider do not need (or use) a prior distribution over types of the other agents.

One of our mechanisms is not exactly truthful-in-expectation, but almost so. We formalize this next.

\begin{definition}\label{def:eps-truthful}
Fix a function $\varepsilon: \mathbb{N}\rightarrow [0,1]$.
Consider any randomized (or distributional) mechanism $\randMech$ for $n$ agents with valuations from class $\V$. 
We say that $\randMech$ is $(1-\varepsilon(n))$-TIE, if for every valuation profile $\valprof\in\Vprofs$, every agent $i \in \N$, and every alternative valuation $\valp_i\in\V$, it holds that
\[
v_i\!\left(\randMech\!\left(\val_i,\, \val_{-i}\right)\right)
\ge (1-\varepsilon(n))\cdot
v_i\!\left(\randMech\!\left(\val'_i,\, \val_{-i}\right)\right).
\]

Where the values above are interpreted in expectation over the randomness of the mechanism.
\end{definition}

We say that a function $\varepsilon: \mathbb{N}\rightarrow [0,1]$ is \emph{negligible} if $\lim_{n\rightarrow\infty}{\varepsilon(n)\cdot n^c}=0$ for every constant $c\ge 0$. Note that, in particular, this implies that $\lim_{n\rightarrow\infty}{\varepsilon(n)}=0$. {We say an event occurs with \emph{overwhelming} probability if it occurs with probability at least $1-\varepsilon(n)$ for some negligible $\varepsilon(n)$.}

\subsection{Fairness}

In this paper we focus on share-based fairness. A share function defines for each agent her deserved share of the total value, which is independent of the other agents' preferences. Each agent is then satisfied when she receives value which is at least her fair share, or a specified fraction of it. We next present the three share functions which are most relevant to our work.

\begin{definition}
    The \emph{proportional share} of an agent with valuation $\val:2^\M\rightarrow\Real$ when allocating goods among $n$ agents is $\PROP_n(\val,\M)=\frac{\val(\M)}{n}$.
\end{definition}

A central share for ex-post allocation of indivisible items is the maximin-share (MMS), introduced by Budish \cite{B-11}.

\begin{definition}
The \emph{maximin-share (MMS)} of an agent with valuation $\val:2^\M\rightarrow\Real$ when allocating goods among $n$ agents, denoted $\MMS_n(\val,\M)$, is the maximal value an agent with valuation $\val$ can guarantee herself when she partitions the items into $n$ bundles and gets the one with the lowest-value. Formally:
\[\MMS_n(\val,\M)=\max_{\alloc\in\allocs}\min_{i\in\N}\left\{ \val(\allocset_{i})\right\}\]
\end{definition}

The truncated-proportional share (TPS) is a share for additive agents that was introduced by \cite{BEF-22}.

\begin{definition}\label{def:TPS}    
The \emph{truncated-proportional share (TPS)} of an agent with valuation $\val:2^\M\rightarrow\Real$ when allocating the set $\M$ of goods among $n$ agents, {denoted $\TPS_n(\val,\M)$,} is the minimum between her proportional share and her $\TPS$ on the instance after removing her highest ranked item and a single agent. That is, letting $\best_v(S)$ denote for each set $S\subseteq M$ a highest-value good in $S$ for $\val$, the $\TPS$ is defined recursively as follows:
    \[\TPS_n(\val, \M)=\min\left(\frac{\val(\M)}{n},\TPS_{n-1}\left(\val,\M\setminus\best_v(\M)\right)\right)\]
    With stopping condition $\TPS_1(\val,S)=\val(S)$.
\end{definition}

The TPS is equal to the proportional share when no item has value above the proportional share, and it is zero when there are fewer items than agents.

{When the set of goods $\M$ is clear from context, we will drop it from the notation of the shares.} In \cite{BEF-22} it was shown that these three shares satisfy the following inequalities.
\begin{lemma}
    \label{lem:share-inequality}
    For every $\M,n$, and an additive valuation $\val\in \Vadd$ it holds that:
    \[\MMS_n(\val)\le \TPS_n(\val)\le \PROP_n(\val)\]
\end{lemma}

These three inequalities can sometimes be strict and sometimes hold with equality.

Notably, while $\TPS$ and $\PROP$ are computable in polynomial times, computing the $\MMS$ for a given valuation is an NP-hard problem \cite{W-97}. For this reason, throughout the paper we construct our mechanisms to approximate the $\TPS$ instead of $\MMS$, which, by \Cref{lem:share-inequality}, imply identical results for the $\MMS$, in addition to the computational advantage.

\subsubsection{Fair Randomized Mechanisms}

Our mechanisms use randomization for incentives, with the main goal of obtaining ex-post fairness, $\MMS$ (or $\TPS$) approximation which must hold for every realized integral allocation within the distribution's support. Yet, in alignment with the Best-of-Both-Worlds paradigm, our mechanisms also guarantee some ex-ante fairness, in the sense of ex-ante proportionality. 

Our ex-post fairness guarantees are multiplicative approximations of the $\MMS$ and $\TPS$ fair shares.

\begin{definition}
{For $\apx>0$ and valuation profile $\valprof\in \Vprofs$, we say allocation $\alloc\in\allocs$ is \emph{$\apx$-$\MMS$} if each agent $i\in \N$ receives at least $\apx$ fraction of her $\MMS$, that is $\val_i(\allocset_i)\geq \apx \cdot \MMS_n(\val_i)$. A randomized (or a distributional) mechanism $\randMech$ is \emph{$\apx$-$\MMS$ ex-post} if for every valuation profile $\valprof\in \Vprofs$, every allocation in the support $\alloc\in\supp{\randMech(\valprof)}$ is $\apx$-$\MMS$. We define $\apx$-$\TPS$ allocations and $\apx$-$\TPS$ ex-post similarly.}
\end{definition}

We will be using \emph{$\apx$-$\MMS$} and \emph{$\apx$-$\TPS$} only in the ex-post sense, so we might sometimes drop the words ``ex-post'' when it is clear from context.

From the perspective of ex-ante fairness, our mechanisms are ex-ante proportional. 

\begin{definition}
    A randomized (or a distributional) mechanism $\randMech$ is \emph{proportional ex-ante} if for every valuation profile $\valprof\in \Vprofs$, every agent $i\in\N$ receives at least her proportional share in expectation, that is $\expect[\alloc\sim \randMech(\valprof)]{\val_i(\allocset_i)}\ge\frac{\val_i(\M)}{n}$.
\end{definition}

\section{Designing the Truthful Fractional Allocation Mechanism}
\label{sec:design}

We seek to design distributional mechanisms for agents with additive valuations that are truthful-in-expectation and guarantee every agent a high fraction of their $\MMS$ ex-post. When agents have additive valuations, TIE distributional mechanisms induce truthful fractional allocations, and the same holds in reverse. Formally:

\begin{restatable}{lemma}{tiefrac}
    \label{lem:tie-fractional}
    When agents have additive valuations over goods, a randomized (or distributional) mechanism {$\randMech$} is TIE if and only if its induced fractional allocation mechanism $\fracMech$ is truthful.
\end{restatable}

The proof follows immediately from the definitions and is given for completeness in \Cref{app:design}. Following the approach of previous works \cite{BT-24, GNS-25}, we make use of \Cref{lem:tie-fractional} to go in reverse: we first design a truthful fractional mechanism, and only later implement it with a distributional mechanism. In this section, we focus on the first step, and discuss our approach for designing truthful fractional allocation mechanisms, which we use to design the \emph{cyclic-unit-quota} fractional mechanism. In the following sections (\Cref{sec:logn}, \Cref{sec:loglogn}, and \Cref{sec:two-agents}) we then perform the second step, each time implementing the cyclic-unit-quota fractional mechanism with distributional or randomized mechanism with different ex-post fairness guarantees.

A trivially truthful fractional allocation mechanism is one that ignores the input entirely, for example, by giving every agent $1/n$ of each item. In \cite{BT-24}, it was shown that this straightforward mechanism can be implemented by a distributional mechanism guaranteeing $1/n$-$\MMS$ ex-post. More sophisticated mechanisms allow the fractional allocation to depend on agents' reports. In the same work, \cite{BT-24} suggest a more intricate truthful mechanism for three agents that does depend on reports. Yet, even in that mechanism, one agent (chosen independently of the reports) still receives each item with identical probability $1/n$. The following impossibility result proves that any mechanism with this property fails to always guarantee every agent strictly better than $1/n$-$\MMS$ ex-post.

\begin{restatable}{proposition}{nimpossibility}    
\label{prop:n-approximation-impossibility}
    Consider a {randomized (or distributional)} mechanism $\randMech$ for $n$ agents with additive valuations over $m\geq 2n-1$ goods. If there exists an agent $i\in\N$ and set $X\subseteq\M$ of $|X|=n-1$ items such that for every valuation profile $\valprof\in \Vaddprofs$ there is a  positive probability that agent $i$ does not receive any item in $X$, that is $\prob[\alloc\sim\randMech(\valprof)]{ X\cap \allocset_i=\emptyset } >0$, then there exists a valuation profile for which the mechanism cannot guarantee all agents strictly more than $\frac{1}{n}$-$\MMS$ ex-post.
\end{restatable}
\begin{proof} 
    Let $Y\subseteq\M\setminus X$ be a set of $|Y|=n$ items disjoint from $X$ (such a set must exist because $m\geq 2n-1$). Consider the following valuation profile:     
    \begin{center}
        \centering
        \begin{tabular}{c|ccc}
             &  $x_1,\ldots,x_{n-1}\in X$&  $y_1,\ldots,y_n\in Y$ & $\M\setminus(X\cup Y)$\\\hline 
             $\val_i$&  $1$&  $1/n$ & $0$\\ 
             $\forall j\neq i,\space\val_j$&  $0$&  $1$ & $0$\\ 
        \end{tabular}
    \end{center}
    
    First note that the $\MMS$ of all agents is exactly $1$. {Their $\MMS$ is at most $1$ because their proportional share is $1$ (using \Cref{lem:share-inequality}).} For agent $i$ this is achieved by placing each of the $n-1$ items in $X$ in separate bundles and all of $Y$ in the last bundle. For agents $j\neq i$ this is achieved by placing each of the $n$ items in $Y$ in a separate bundle.
    
    Now, recall that we assumed $\prob[\alloc\sim\randMech(\valprof)]{X\cap \allocset_i=\emptyset } >0$, therefore there must exist some ex-post allocation $\alloc\in \supp{\randMech(\valprof)}$ for which $X\cap \allocset_i=\emptyset$. In this allocation $\alloc$, if any of the agents $j\neq i$ does not receive at least one item from $Y$ then they get value $0$, a $0$ fraction of their $\MMS$. Otherwise, agent $i$ can only receive at most one item from $Y$, which combined with $\allocset_i\cap X=0$ implies that agent $i$ gets value of at most $\frac{1}{n}$ (the value of a single item from $Y$). In either case, mechanism $\randMech$ cannot guarantee all agents strictly more than $\frac{1}{n}$-$\MMS$ ex-post.
\end{proof}

This impossibility result has the following significant corollary:

\begin{restatable}{corollary}{imposcor}
    \label{cor:impossibility}
    Consider a {randomized (or distributional)} mechanism $\randMech$ for $n$ agents with additive valuations over $m\geq 2n-1$ goods. If there exists an agent $i$ such that $\randMech$ allocates each item to $i$ with marginal probability of at most $\frac{1}{n}$, then there exists a valuation profile $\valprof\in\Vaddprofs$ for which the mechanism cannot guarantee all agents strictly more than $\frac{1}{n}$-$\MMS$ ex-post. 
\end{restatable}
\begin{proof}
{\Cref{cor:impossibility} follows immediately from \Cref{prop:n-approximation-impossibility} since for any set $X\subseteq\M$ of size $|X|=n-1$ satisfies the requirement for the impossibility result by the union-bound:
\[\prob[\alloc\sim\randMech(\valprof)]{ X\cap \allocset_i\neq \emptyset }=\prob[\alloc\sim\randMech(\valprof)]{\bigcup_{\ix\in X}x\in \allocset_i}\le \sum_{x\in X}\frac{1}{n}\le \frac{n-1}{n}<1\]}
as required.
\end{proof}

{Note that \Cref{cor:impossibility} holds, in particular, for the uniform fractional allocation.} 

To overcome this impossibility result, we design a new truthful fractional allocation mechanism which depends on the reports of all the agents. In order to design this mechanism we first take a step back, and start out from a universally-truthful mechanism.

\subsection{Universally-Truthful Mechanisms to the Rescue}
\label{sec:univers-truthful}

Recall that a universally-truthful mechanism is simply a distribution over deterministic truthful mechanisms. Our approach is the following:
start from a universally-truthful mechanism, find its induced truthful fractional mechanism, then re-implement the fractional mechanism with a distributional mechanism supported only on allocations that are ex-post fair. As the fractional-allocation mechanism does not change when re-implementing it on a different support, the resulting mechanism is truthful-in-expectation for additive agents. We formalize this approach in the following result.

\begin{restatable}{proposition}{unitotie}
    \label{prop:universal-to-tie}
    Consider the setting where agents have additive valuations.
    Let $\univMech$ be a \emph{universally-truthful} mechanism and let $\randMech$ be any randomized (or distributional) mechanism. If the fractional allocation mechanisms induced by $\univMech$ and $\randMech$ are identical, then $\randMech$ is \emph{truthful-in-expectation}.
\end{restatable}

The proof follows from applying \Cref{lem:tie-fractional} twice, once from the $\univMech$ to the fractional-allocation mechanism and then again from the fractional-allocation mechanism to $\distMech$. We defer the formal proof to \Cref{app:design}.

\subsection{Unit-Quota Mechanisms with a Random Starting Agent}
\label{subsec:unit-quota}

\Cref{prop:universal-to-tie} states that we can take any universally-truthful mechanism and use it to create TIE mechanisms for additive agents by re-implementing its fractional-allocation mechanism. For any given input, we want to re-implement its fractional-allocation on ex-post fair allocations. Which universally-truthful mechanism should we use? and how should we re-implement its fractional allocation?   

Our starting point is serial-quota mechanisms, a class of deterministic truthful ordinal mechanisms. These mechanisms were studied in multiple papers \cite{P-00, H-09, HL-15, ABM-16,ABCM-17, ALW-20,BGLM-25, BMM-25}. Essentially, a serial-quota mechanism first sets an order over the agents and sets a quota (quantity) for each, and then each agent in her turn (according to the order), picks items with number equal to her quota. We focus on one such mechanism: the one in which each agent, but the last, picks a single item from the ones remaining, and the last gets all the rest.

This mechanism, which we call the \emph{unit-quota mechanism} has received special attention in previous works for its fairness properties. Its $\MMS$ approximation was studied in \cite{ABM-16}, which showed that the unit-quota mechanism (notated $M^{(n,m)}_{\text{PICK SEQ}}$) obtains $1/\floor{\frac{m-n+2}{2}}$ fraction of the $\MMS$, which they then show to be the best possible among all serial-quota mechanisms. The specific case of two agents is studied in \cite{ABCM-17}, which showed that the unit-quota mechanism obtains the best possible $\MMS$ approximation among all truthful deterministic mechanisms for two agents (with the additional restriction that all items must be allocated). In \cite{BMM-25} this optimality result was generalized to deterministic mechanisms for an arbitrary number of agents with the additional assumptions of neutrality and non-bossiness. 

Despite these optimality results, the fairness guarantees of the unit-quota mechanism are extremely poor, guaranteeing only $O(1/m)$-$\MMS$. Yet, we show, maybe surprisingly, that a randomized version of this mechanism can be used to design a fractional allocation mechanism which is implementable ex-post on a support of integral allocations which obtain much higher approximations of the $\MMS$ ($\Omega(1/\ln n)$ in \Cref{thm:logn-apx} and $\Omega(1/\log\log n)$ in \Cref{thm:loglog-apx} with a relaxed incentive requirement). Notably, these bounds only depend on the number of agents $n$, so for a fixed number of agents the approximation stays constant even as the number of items grows large.

For an agent with an additive valuation, once a tie-breaking rule over goods is set (i.e. breaking ties by label), it induces a strict order over items. The ordinal mechanisms that we define use these strict orders.  
\begin{definition}
    \label{def:det-unit-quota}
    Consider a setting with a set $\goods$ of items and $n$ agents, each with a strict ordinal preference over $\goods$. For ordering $\per\in\PerN$ of the agents, the \emph{$\per$-unit-quota mechanism} is a deterministic ordinal mechanism in which for each index {$j$} from $1$ to $n-1$, agent {$\pinv_j$}
    gets her most preferred item (according to her strict ordinal preference) out of the remaining items. {Finally, the last agent $\pinv_n$ gets all remaining items}. 
\end{definition}

Our mechanisms require each agent to have a strict ordering of goods which remains consistent across all cyclic orderings. Under the direct revelation model, the mechanisms can enforce this by systematically breaking ties in the reported valuations by some arbitrary rule (e.g., lexicographically by item label). Alternatively, the mechanism can be formulated with a restricted strategy space where agents directly report strict ordinal preferences. In this mechanism, reporting any strict ordering consistent with the agent's underlying cardinal valuation constitutes a dominant strategy.

Note that the unit-quota mechanism is truthful for any ordering $\per$, as each agent obtains their optimal item by reporting truthfully and no misreport can get them anything better. In this paper we focus on the universally-truthful mechanism created by picking $j\in [n]$ uniformly at random, and running the $\per$-unit-quota mechanism after shifting the ordering by $j$ positions \emph{cyclically}. 

\begin{definition}
    \label{def:cyclic-shifts}
    For some initial ordering of the agents $\per\in\PerN$ and index $j\in[n]$ we use $\percyc{j}$ to denote the ordering resulting from cyclically shifting $\per$ to start at the $j$-th index. Formally: \[\percycij=\begin{cases}
        \per_i-j+1 & \per_i\geq j \\
        \per_i-j+1+n&\text{else}
    \end{cases}\]
\end{definition}

Using this notation, we can define the following universally-truthful distributional mechanism.

\begin{definition}
    \label{def:rand-unit-quota}
    
    The \emph{cyclic-$\per$-unit-quota} distributional mechanism for an initial ordering $\per$, is the mechanism that samples $j\in[n]$ uniformly at random and runs the $\percyc{j}$-unit-quota mechanism.
    
    The \emph{cyclic-$\per$-unit-quota fractional allocation mechanism} is the fractional allocation mechanism induced by the cyclic-$\per$-unit-quota distributional mechanism (with ties broken lexicographically by item names). For a given valuation profile $\valprof\in\Vaddprofs$, we denote this fractional allocation by $\cuqalloc(\valprof,\per)$.
\end{definition}

We note that both the cyclic-$\per$-unit-quota distributional and fractional mechanisms depend only on the strict ordinal preferences over the goods, and thus can be easily adapted to the restricted strategy space where each agent 
may only report a strict ordinal preference over the goods. Additionally, both can be computed explicitly in polynomial time by iterating over all values of $j\in[n]$.

{
The following claim shows that when valuations are additive, implementing the cyclic-$\per$-unit-quota fractional mechanism results in a randomized mechanism which is TIE and proportional ex-ante. TIE follows from applying \Cref{prop:universal-to-tie} from the universally-truthful cyclic-$\per$-unit-quota distributional mechanism, and proportionality ex-ante follows from the proportionality of the fractional mechanism. 
}

\begin{restatable}{lemma}{cuqTieProp}
    \label{lem:cuq-tie}
    Fix any ordering $\per\in\PerN$. When agents have additive valuations, every randomized (or distributional) mechanism which induces the cyclic-$\per$-unit-quota fractional mechanism is \emph{TIE} and \emph{proportional ex-ante}. 
\end{restatable}

We defer the proof to \Cref{app:design}.

In the following section, we will design a distributional allocation mechanism which implements the cyclic-$\per$-unit-quota fractional allocation mechanism. By \Cref{lem:cuq-tie} any such mechanism is both TIE and proportional, and so the challenging step is to design the mechanism to only output distributional allocations whose support consists of allocations with a good approximation to the $\MMS$.

\section{Implementing the Cyclic-Unit-Quota Fractional Mechanism}
\label{sec:gen-imp}

In the previous section we have established that the Cyclic-Unit-Quota fractional mechanism is truthful, and so any distributional mechanism which implements it is TIE. In this section, we present a {generic result} for implementing this fractional allocation on a distribution supported only on fair allocations (\Cref{lem:implement-export}). These results are then used to prove our main results (\Cref{thm:logn-apx} and \Cref{lem:general-weight}, the main step in proving \Cref{thm:loglog-apx}).

Fix an additive valuation profile $\valprof\in \Vaddprofs$ and an ordering of the agents $\per\in\PerN$. These will remain constant throughout the entire section, and are a prerequisite to every definition.

For a given index $j\in[n]$, we consider the cyclic shift (see \Cref{def:cyclic-shifts}) $\percyc{j}$, which is derived from the initial ordering $\per$ by shifting the initial agent ``forward'' cyclically $j-1$ times. For every cyclic shift $\percyc{j}$ we can run the respective picking order, in which each agent picks a single item in turn according to $\percyc{j}$. Note that unlike the deterministic $\percyc{j}$-unit-quota mechanism, here we do not allocate all remaining goods to the last agent. Instead, we will ``split'' this suffix of goods fractionally among the agents.

Our generic method for splitting the suffix of goods  allocates each agent $i\in \N$ a ``portion'' $0\le \portij\le 1$ of every item in the suffix in the cyclic shift $\percyc{j}$. These portions form a matrix:

\begin{definition}
    A \emph{portion matrix} is an $n\times n$ matrix $\portmat\in {[0,1]}^{n\times n}$ such that the sum of each row and each column is at most $1$. Let $\portmats$ denote the set of all such matrices. 
\end{definition}

We will show that every portion matrix can be used to construct a distributional allocation which induces the {fractional allocation} $\cuqalloc(\valprof,\per)$. Importantly, the entries of the  portion matrix may depend on the valuations of the agents (and not only the preference orders). In fact, the portion matrix is the \emph{only} element in this section which can incorporate cardinal information: all other definitions only use ordinal information.

We now introduce the central result of this section, which is used in the proofs of \Cref{thm:logn-apx} and \Cref{thm:loglog-apx}.

\begin{restatable}{lemma}{implementExport}
    \label{lem:implement-export}
    There exists a polynomial time ordinal algorithm for the following problem. The algorithm is given an ordering of $n$ agents $\per\in\PerN$, a profile of $n$ strict ordinal preferences  $(\en{1},\ldots,\en{n})$ that is consistent with an additive profile $\valprof \in \Vaddprofs$, and a portion matrix $\portmat\in\portmats$  and computes a distributional allocation $\D$, and $n$ additional distributional allocations $\Dnum{1},\ldots,\Dnum{n}$, such that:
    \begin{enumerate}
        \item $\D$ induces $\cuqalloc(\valprof,\per)$, and has support of size at most $n\cdot m$ that satisfies $\supp{\D}\subseteq\bigcup_{j=1}^n\supp{\Dj}$.
        \item For every agent $i\in\N$ and index $j\in [n]$, for every ex-post allocation $\alloc\in\supp{\Dnum{j}}$ there exists $1\le r\le \percycij$ such that:
        \[\val_i(\allocset_i)\ge \max\left(\val_i(\eiat{r})\;,\;\portij\cdot \val_i(\M\setminus\{\eiat{1},\ldots,\eiat{2n-r}\})\right)\]
    \end{enumerate}
\end{restatable}

The algorithm of \Cref{lem:implement-export} is entirely ordinal, that is, it {never makes use of the} agents' cardinal valuations, only a strict ordinal preference profile that is consistent with the additive profile. This is possible even though $\cuqalloc$ depends on $\valprof$, since we show that $\cuqalloc$ can also be computed with ordinal information alone (\Cref{lem:cuq-fractional}).
The valuations are mentioned in the lemma, as the value guarantee holds for them (although they are unknown to the algorithm).

The remainder of section is dedicated to proving \Cref{lem:implement-export}. In \Cref{subsec:cyclic-pick-notation} we present some useful notation, and use it to obtain an explicit characterization of $\cuqalloc$. In \Cref{subsec:partial-frac} we define for portion matrix $\portmat$ its corresponding partial fractional allocations. In \Cref{subsec:completing-partial} we then complete these into full fractional allocations. In \Cref{subsec:ex-post-fairness} we use the faithful implementation technique to obtain some general ex-post fairness results. Finally, in \Cref{subsec:proving-export} we complete the proof of \Cref{lem:implement-export}. 

\subsection{Cyclic Picking Order Notation}
\label{subsec:cyclic-pick-notation}

We now introduce some useful notation for the cyclic picking orders corresponding to the $j\in [n]$ cyclic shift.

\begin{definition}  
    Fix an ordering of the agents $\per\in\PerN$ and a  profile of $n$ strict ordinal preferences $(\en{1},\ldots,\en{n})$. Let $i\in\N$ and $j\in[n]$. Then:
    \begin{itemize}
        \item Let $\ip\in[n]$ denote the index such that agent $i$ picks last in the cyclic shift $\percyc{\ip}$, that is, $\percyc{\ip}_i=n$. Define $i_{-1}$ cyclically similarly.
        \item Let $\pitemij$ denote the good chosen by agent $i$ in the cyclic picking order $\percyc{j}$.
        \item Let $\Yj$ denote the set of $n$ goods chosen by all agents in the cyclic picking order starting at index $j$, that is $\Yj=\left\{\pitem{1}{j},\ldots,\pitem{n}{j} \right\}$.
    \end{itemize}
\end{definition}

Using this notation, we can now give an explicit characterization of the cyclic-$\per$-unit-quota fractional allocation. It uses $1[\cdot]$ to denote value of $1$ when the expression within the brackets holds, and $0$ otherwise.

\begin{lemma}
\label{lem:cuq-fractional}   
    Fix an ordering of the agents $\per\in\PerN$ and  a profile of $n$ strict ordinal preferences  $(\en{1},\ldots,\en{n})$, 
    that is consistent with an additive profile $\valprof \in \Vaddprofs$. Then for every agent $i\in\N$ and good $\ix\in\M$:
    \[\cuqix(\valprof,\per)=\frac{1}{n}\cdot \left({1}\left[\ix\notin\Yip\right]+\sum_{j=1}^n1\left[\ix=\pitemij\right]\right)\]
\end{lemma}
\begin{proof}
    Recall that the fractional allocation  $\cuqalloc$ was constructed by iterating over all $n$ cyclic picking orders, allocating all remaining items past the $n$ chosen items to the agent who picked last. Thus the fraction of item $\ix$ received by agent $i$ is equal to $\frac{1}{n}$ multiplied by the number of sequences where she receives item $\ix$. This can happen in two cases: either $\ix$ is picked explicitly by agent $i$, in which case $\ix=\pitemij$, or agent $i$ is last (which occurs in ordering $\percyc{\ip}$) and none of the agents picked $\ix$, thus $\ix\notin \Yip$.
\end{proof}

Note that \Cref{lem:cuq-fractional} implies that the cyclic-unit-quota fractional allocation can be computed in polynomial time given only strict ordinal preferences that are consistent with the additive valuations.

The following additional notation is useful when describing the partial fractional decomposition in \Cref{subsec:partial-frac}.

\begin{definition}
    Fix an ordering of the agents $\per\in\PerN$ and a  profile of $n$ strict ordinal preferences $(\en{1},\ldots,\en{n})$. Let $i\in\N$ and $j\in[n]$. Then:
    \begin{itemize}
        \item Let $\rij$ denote the rank of $\pitemij$ to agent $i$, so that $\eiat{\rij}=\pitemij$.
        \item Let $\zsetij=\left(\Yj\cup \Yip\right)\setminus\left\{\eiat{\rij}\right\}$ denote the set of ``zeroed'' goods.
        \item Let $\fsetij=\M\setminus\left(\Yj\cup \Yip\right)$ denote the set of ``fractional'' goods.
    \end{itemize}
\end{definition}

Note that for every $i\in\N,j\in[n]$ we have $\M=\{\eiat\rij\}\cup \zsetij\cup\fsetij$.

\subsection{Decomposition to Partial Fractional Allocations}
\label{subsec:partial-frac}

For a portion matrix $\portmat\in\portmats$, we will construct $n$ partial fractional allocations, one for each cyclic shift $\percyc{j}$, which fully allocate the first $n$ goods according to the picking order $\percyc{j}$, and splits the remaining suffix between the agents according to the $j$-th column of $\portmat$.

\begin{definition}
    \label{def:partial-frac}
    Fix an ordering of the agents $\per\in\PerN$ and a profile of $n$ strict ordinal preferences  $(\en{1},\ldots,\en{n})$.

    Given portion matrix $\portmat\in \portmats$, its corresponding partial fractional allocations $\pga{1},\ldots,\pga{n}$ are defined as follows.
    For every $i\in\N$ and $j\in[n]$:
    \begin{enumerate}
        \item Fully allocate her chosen item 
        $x=\pitemij$ in the  picking order, that is, set $\pgixj=1$.
        \item For every item $\ix\in\Yj\setminus\left\{\pitemij\right\}$ picked by the other agents, set $\pgixj=0$ (a ``zeroed'' good $\ix\in\zsetij$).
        \item For every item $\ix\in \Yip\setminus \left\{\pitemij\right\}$, set $\pgixj=0$ (also a ``zeroed'' good $\ix\in\zsetij$).
        \item Finally, for every remaining ``fractional'' good $\ix\in\fsetij=\M\setminus\left(\zsetij\cup\{\pitemij\}\right)$, set $\pgixj=\portij$.
    \end{enumerate}
\end{definition}

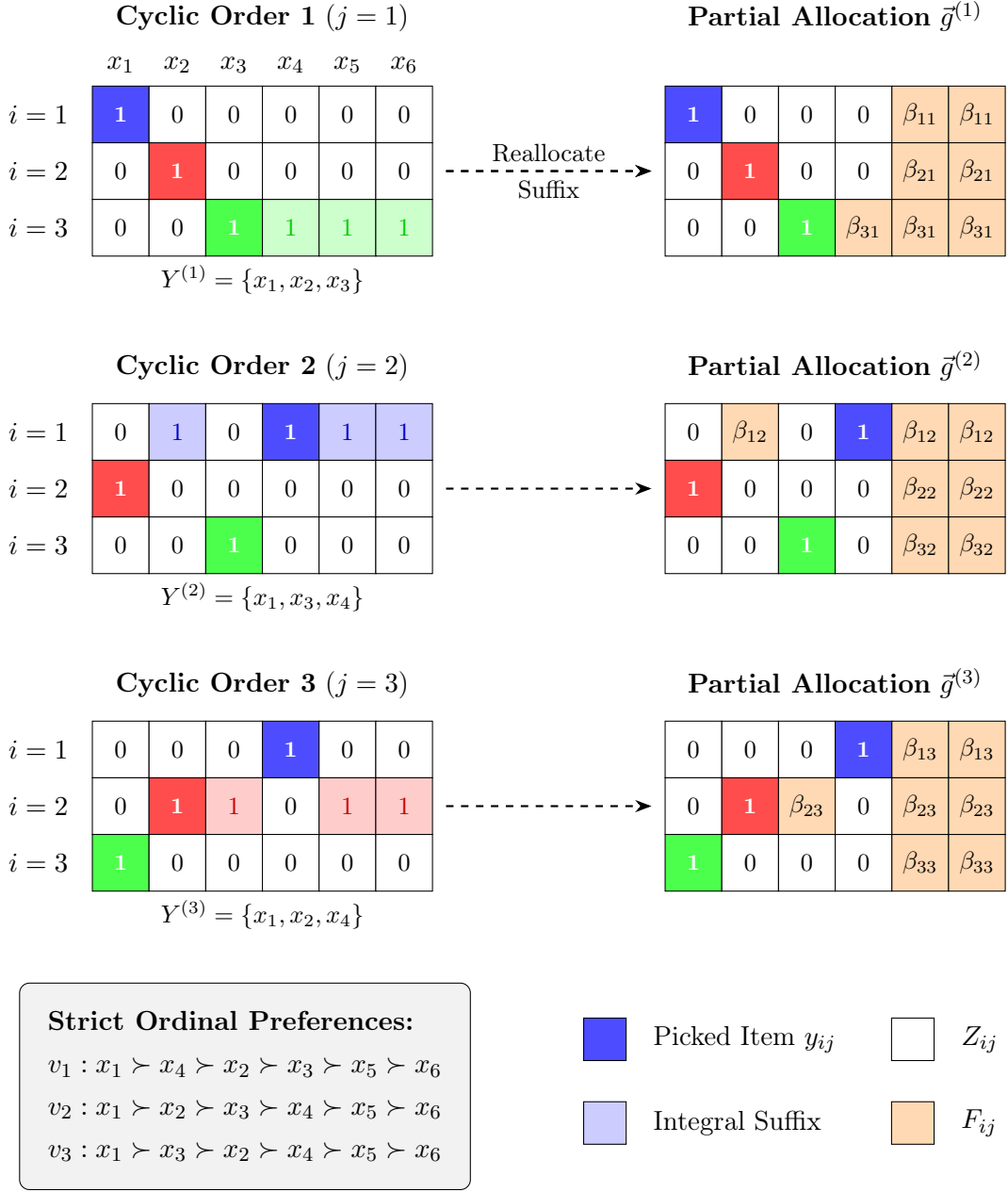
\begin{figure}[htbp]
    \centering

        \begin{tikzpicture}[
        scale=0.5,
        >=Stealth,
        cell/.style={rectangle, draw=black, minimum width=8mm, minimum height=8mm, align=center, anchor=center, font=\small},
        pick1/.style={cell, fill=blue!70, text=white, font=\footnotesize\bfseries},
        pick2/.style={cell, fill=red!70, text=white, font=\footnotesize\bfseries},
        pick3/.style={cell, fill=green!70, text=white, font=\footnotesize\bfseries},
        suff1/.style={cell, fill=blue!20, text=blue!80!black},
        suff2/.style={cell, fill=red!20, text=red!80!black},
        suff3/.style={cell, fill=green!20, text=green!80!black},
        frac/.style={cell, fill=orange!30, text=black}, 
        zero/.style={cell, fill=white, text=black},
        unavail/.style={cell, fill=gray!30, text=black},
        matrix style/.style={matrix of nodes, nodes in empty cells, row sep=-\pgflinewidth, column sep=-\pgflinewidth}
    ]


    \matrix (M1) [matrix style] {
        |[pick1]| 1& |[zero]| 0& |[zero]| 0& |[zero]| 0& |[zero]| 0& |[zero]| 0\\
        |[zero]| 0& |[pick2]| 1& |[zero]| 0& |[zero]| 0& |[zero]| 0& |[zero]| 0\\
        |[zero]| 0& |[zero]| 0& |[pick3]| 1& |[suff3]| 1& |[suff3]| 1& |[suff3]| 1\\
    };
    \node[above=0.5cm of M1] {\textbf{Cyclic Order 1} ($j=1$)};
    \node[left=0.2cm of M1-1-1] {$i=1$};
    \node[left=0.2cm of M1-2-1] {$i=2$};
    \node[left=0.2cm of M1-3-1] {$i=3$};

    \foreach \c/\item in {1/x_1, 2/x_2, 3/x_3, 4/x_4, 5/x_5, 6/x_6} {
        \node[above=0.05cm of M1-1-\c, font=\normalsize] {$\item$};
    }

    \matrix (M2) [matrix style, below=1.8cm of M1] {
        |[zero]| 0& |[suff1]| 1& |[zero]| 0& |[pick1]| 1& |[suff1]| 1& |[suff1]| 1\\
        |[pick2]| 1& |[zero]| 0& |[zero]| 0& |[zero]| 0& |[zero]| 0& |[zero]| 0\\
        |[zero]| 0& |[zero]| 0& |[pick3]| 1& |[zero]| 0& |[zero]| 0& |[zero]| 0\\
    };
    \node[above=0.05cm of M2] {\textbf{Cyclic Order 2} ($j=2$)};
    \node[left=0.2cm of M2-1-1] {$i=1$};
    \node[left=0.2cm of M2-2-1] {$i=2$};
    \node[left=0.2cm of M2-3-1] {$i=3$};
    
    \matrix (M3) [matrix style, below=1.8cm of M2] {
        |[zero]| 0& |[zero]| 0& |[zero]| 0& |[pick1]| 1& |[zero]| 0& |[zero]| 0\\
        |[zero]| 0& |[pick2]| 1& |[suff2]| 1& |[zero]| 0& |[suff2]| 1& |[suff2]| 1\\
        |[pick3]| 1& |[zero]| 0& |[zero]| 0& |[zero]| 0& |[zero]| 0& |[zero]| 0\\
    };
    \node[above=0.05cm of M3] {\textbf{Cyclic Order 3} ($j=3$)};
    \node[left=0.2cm of M3-1-1] {$i=1$};
    \node[left=0.2cm of M3-2-1] {$i=2$};
    \node[left=0.2cm of M3-3-1] {$i=3$};


        

    \node[below=-0.15cm of M1, font=\small] {$Y^{(1)} = \{x_1, x_2, x_3\}$};
    \node[below=-0.15cm of M2, font=\small] {$Y^{(2)} = \{x_1, x_3, x_4\}$};
    \node[below=-0.15cm of M3, font=\small] {$Y^{(3)} = \{x_1, x_2, x_4\}$};


    \matrix (M4) [matrix style, right=3cm of M1] {
        |[pick1]| 1& |[zero]| 0& |[zero]| 0& |[zero]| 0& |[frac]| ${\port_{11}}$& |[frac]| $\port_{11}$\\
        |[zero]| 0& |[pick2]| 1& |[zero]| 0& |[zero]| 0& |[frac]| ${\port_{21}}$& |[frac]| ${\port_{21}}$\\
        |[zero]| 0& |[zero]| 0& |[pick3]| 1& |[frac]| ${\port_{31}}$& |[frac]| ${\port_{31}}$& |[frac]| $\port_{31}$\\
    };
    \node[above=0.5cm of M4] {\textbf{Partial Allocation} $\pga{1}$};
    
    \matrix (M5) [matrix style, right=3cm of M2] {
        |[zero]| 0& |[frac]| ${\port_{12}}$& |[zero]| 0& |[pick1]| 1& |[frac]| ${\port_{12}}$& |[frac]| ${\port_{12}}$\\
        |[pick2]| 1& |[zero]| 0& |[zero]| 0& |[zero]| 0& |[frac]| ${\port_{22}}$& |[frac]| ${\port_{22}}$\\
        |[zero]| 0& |[zero]| 0& |[pick3]| 1& |[zero]| 0& |[frac]| ${\port_{32}}$& |[frac]| ${\port_{32}}$\\
    };
    \node[above=0.05cm of M5] {\textbf{Partial Allocation} $\pga{2}$};

    \matrix (M6) [matrix style, right=3cm of M3] {
        |[zero]| 0& |[zero]| 0& |[zero]| 0& |[pick1]| 1& |[frac]| ${\port_{13}}$& |[frac]| ${\port_{13}}$\\
        |[zero]| 0& |[pick2]| 1& |[frac]|${\port_{23}}$& |[zero]| 0& |[frac]| ${\port_{23}}$& |[frac]| ${\port_{23}}$\\
        |[pick3]| 1& |[zero]| 0& |[zero]| 0& |[zero]| 0& |[frac]| ${\port_{33}}$& |[frac]| ${\port_{33}}$\\
    };
    \node[above=0.05cm of M6] {\textbf{Partial Allocation} $\pga{3}$};

    \draw[->, thick, dashed] ($(M1.east)+(0.1,0)$) -- ($(M4.west)-(0.1,0)$) node[midway, above, font=\small, align=center] {Reallocate};
    \draw[->, thick, dashed] ($(M1.east)+(0.1,0)$) -- ($(M4.west)-(0.1,0)$) 
    node[midway, below, font=\small, align=center] {Suffix};
    \draw[->, thick, dashed] ($(M2.east)+(0.1,0)$) -- ($(M5.west)-(0.1,0)$);
    \draw[->, thick, dashed] ($(M3.east)+(0.1,0)$) -- ($(M6.west)-(0.1,0)$);

    \path (M3.south west) -- (M6.south east) coordinate[midway] (midbottom);

    
    \node (Vals) [draw, fill=gray!10, rounded corners, align=left, font=\normalsize, inner sep=4mm, anchor=north east, execute at begin node={\linespread{1.25}\selectfont}] at ([xshift=-2.2cm, yshift=-2.3cm]midbottom) {
        \textbf{Strict Ordinal Preferences:}\\
        $v_1: x_1 \succ x_4 \succ x_2 \succ x_3 \succ x_5 \succ x_6$\\
        $v_2: x_1 \succ x_2 \succ x_3 \succ x_4 \succ x_5 \succ x_6$\\
        $v_3: x_1 \succ x_3 \succ x_2 \succ x_4 \succ x_5 \succ x_6$
    };

    \matrix (Legend) [matrix, row sep=5.5mm, column sep=2.5mm, anchor=north west] at ([xshift=0.7cm, yshift=-3cm]midbottom) {
        \node[pick1, minimum width=6mm, minimum height=6mm, anchor=center] {}; & 
        \node[font=\normalsize, align=left, anchor=west]{Picked Item $\pitemij$}; &[0.4cm]
        
        \node[zero, minimum width=6mm, minimum height=6mm, anchor=center] {}; & 
        \node[font=\normalsize, align=left, anchor=west]{$\zsetij$}; &
        &\\
        
        \node[suff1, minimum width=6mm, minimum height=6mm, anchor=center] {}; & 
        \node[font=\normalsize, align=left, anchor=west]{Integral Suffix}; &

        \node[frac, minimum width=6mm, minimum height=6mm, anchor=center] {}; & 
        \node[font=\normalsize, align=left, anchor=west]{$\fsetij$};\\
        
        \\
    };

    \end{tikzpicture}
    
    \caption{ Visualization of the construction of the partial fractional allocations $\pga{1}, \pga{2}, \pga{3}$ for $n=3$ agents and $m=6$ goods, with the identity ordering and the strict ordinal preferences as presented above. The left column displays the three integral allocations created by the original cyclic picking sequences, where the final agent receives the entire integral suffix. The sets of picked goods are $\Yitems{1},\Yitems{2},\Yitems{3}$ written below 
    their corresponding allocations. The right column illustrates the result of the transformation into three partial allocations $\pga{1}, \pga{2}, \pga{3}$. For each $j\in \{1,2,3\}$, the explicitly picked item $\pitemij$ is fully allocated ({``fraction''} of $1$). {The white cells denote the ``zero'' goods $\zsetij$ that are allocated with fraction $0$ to agent $i$ in the cyclic shift $j$. The integral suffix is reallocated fractionally between the agents across the three partial allocations, where in cyclic shift $j$ agent $i$ receives fraction $\portij$ of every fractional good in $\fsetij$.}}
    \label{fig:partial_allocations_construction}
\end{figure}

For a visualization of the constructed partial fractional allocations $\pga{1},\ldots,\pga{n}$ for a specific example, see \Cref{fig:partial_allocations_construction}. Note that for any given portion matrix $\portmat\in \portmats$, the partial fractional allocations $\pga{1},\ldots,\pga{n}$ can be constructed in polynomial time.

\begin{lemma}
    \label{lem:partial-valid}
    Fix an ordering of the agents $\per\in\PerN$ and a profile of $n$ strict ordinal preferences  $(\en{1},\ldots,\en{n})$.

    For every portion matrix $\portmat\in\portmats$, the corresponding $\pga{1},\ldots,\pga{n}$ are valid partial fractional allocations, that is, for every $j\in[n]$ and good $\ix\in\M$, it holds that $\sum_{i\in\N}\pgixj\le 1$.
\end{lemma}
\begin{proof}
    Fix an index $j\in[n]$ and good $\ix\in\M$. We must prove that the total assigned fraction of item $\ix$ among all agents (every $i\in\N$) 
    does not exceed $1$, or formally, $\sum_{i\in\N}\pgixj\le 1$.
    
    If there exists some agent $i\in\N$ such that $\ix=\pitemij$, meaning agent $i$ picks $\ix$ in the picking order $\percyc{j}$, then item $\ix$ is fully allocated to agent $i$, that is $\pgixj=1$ and $\pg{j}_{i',\ix}=0$ for every agent $i'\in\N\setminus\{i\}$, so  $\sum_{i'\in\N}\pg{j}_{i',\ix}=1$, proving the requirement with equality.
        
    Otherwise, the good $\ix$ is not picked by any agent, so $\ix\in\M\setminus \Yj$. This implies by the construction that for every agent $i\in\N$ either $\pgixj=0$ or $\pgixj=\portij$. Therefore:
    \[\sum_{i\in\N}\pgixj\le \sum_{i\in\N}\portij\le 1\]
     as required.
\end{proof}

\subsection{Completing the Decomposition}
\label{subsec:completing-partial}

{To implement the CUQ fraction mechanism, for every valuation profile we must match its fractional allocation exactly. Thus, we need to ``complete'' the partial fractional allocations to full allocations, such that their average match the fractional allocation exactly.}

For a full fractional allocation $\falloc$ and partial fractional allocation $\galloc$, we say that $\falloc$ \emph{completes} $\galloc$ if $\falloc$ is pointwise larger than $\galloc$, that is $\fix\ge \gix$ for every $i\in\N,x\in\M$.

In order to complete the partial fractional allocations $\pga{1},\ldots,\pga{n}$, we make use of the following general results for ``completing'' matrices to target sums.

\begin{restatable}{claim}{greedyCompletion}
    \label{claim:greedy-completion} There exists a polynomial-time algorithm for the following ``matrix completion'' problem. 
    Given a matrix $\matr\in {[0,1]}^{r\times c}$, row targets $\vec{s}\in {[0,c]}^r$ and column targets $1$ for every column, such that \;\emph{(1)} $\sum_{v=1}^c\matr_{uv}\le s_u$ for every row $u\in[r]$, \;\emph{(2)} $\sum_{u=1}^r\matr_{uv}\le 1$ for every column $v\in[c]$, \;\emph{(3)} $\sum_{u=1}^r s_u=c$, the algorithm outputs $\matr'\in{[0,1]}^{r\times c}$ such that $\matr'_{uv}\ge \matr_{uv}$ for every $u\in[r],v\in[c]$, the row sums equal $\vec{s}$ exactly, and the column sums equal $1$ exactly.
\end{restatable}

We defer the proof of \Cref{claim:greedy-completion} to \Cref{app:logn}. The proof relies on a straightforward greedy algorithm that increases the fraction in cells one-by-one until all rows and columns meet their targets, which will happen concurrently due to the third input constraint.

{We next show that there is a polynomial time algorithm for the problem of ``completing'' the partial fractional allocations to match the CUQ fractional allocation.}

\begin{lemma}
    \label{lem:completion}
    There exists a polynomial time algorithm for the following problem. The algorithm is given an ordering of the agents $\per\in\PerN$, a profile of $n$ strict ordinal preferences  $(\en{1},\ldots,\en{n})$ that is consistent with an additive profile $\valprof \in \Vaddprofs$, and a portion matrix $\portmat$, computes full fractional allocations $\dfa{1}, \ldots, \dfa{n}$ such that:
    \begin{enumerate}
        \item For portion matrix $\portmat$ it computes the corresponding partial fractional allocations $\pga{1},\ldots,\pga{n}$.
        \item $\dfa{j}$ completes $\pga{j}$ for every $j\in[n]$.
        \item $\sum_{j=1}^n\dfa{j}=n\cdot \cuqalloc(\valprof,\per)$.
        \item For every $j\in[n],i\in[n],x\in\M$, $\dfa{j}$ is strictly fractional $0<\dfixj<1$ only when $x\in \fsetij$.
    \end{enumerate}
\end{lemma}
\begin{proof}
     We first compute (in polynomial time) the partial fractional allocations $\pga{1},\ldots,\pga{n}$ corresponding to the given portion matrix $\portmat$, as well as the cyclic-unit-quota fractional allocation $\cuqalloc=\cuqalloc(\valprof,\per)$. Note the algorithm can compute $\cuqalloc$ using only the given strict ordinal preference profile, as by \Cref{lem:cuq-fractional} $\cuqalloc$ only depends on the agents' strict ordinal preferences.

    Fix an item $\ix\in\M$, we construct the full fractional allocations for good $\ix$. Construct matrix $\matr\in {[0,1]}^{n\times n}$ as follows. For every $i\in\N,j\in[n]$, set $\matr_{ij}=\pgixj$.  Set row targets $\vec{s}\in\Real^{n}$ defined to be $s_i=n\cdot \cuqix$ for every $i\in \N$. We now prove this construction satisfies all the requirements of \Cref{claim:greedy-completion}.
    
    \begin{claim}
        \label{claim:row-column-targets}
        For item $\ix$ fixed above, for $r=c=n$ the defined matrix $\matr$ and row target $\vec{s}$ satisfy all three requirements of \Cref{claim:greedy-completion}.  
    \end{claim}
    \begin{proof} We prove each requirement separately.
    
        \textit{Requirement 1:} Let $i\in \N$ be some agent (corresponding to a row in $\matr$). Recall from \Cref{lem:cuq-fractional} that  $n\cdot\cuqix=\left({1}\left[\ix\notin\Yip\right]+\sum_{j=1}^n1\left[\ix=\pitemij\right]\right)$. Separate into cases. If $\ix\in\Yip$, then by our construction of $\pga{1},\ldots,\pga{n}$, for every $j\in[n]$, either $\pitemij=\ix$ and then $\pgixj=1$ (good $\ix$ is fully allocated to agent $i$), or otherwise $\pgixj=0$. Thus in this case:
        \[\sum_{j=1}^n \matr_{ij}=\sum_{j=1}^n \pgixj=\sum_{j=1}^n1\left[\pitemij=\ix\right]=n\cdot \cuqix=s_i\]
        So the requirement holds with equality. Otherwise, assume $\ix\notin\Yip$. In this case, again if $\pitemij=\ix$ then $\pgixj=1$, but now otherwise either $\pgixj=0$ or $\pgixj=\portij$. Note that in all cases, $\pgixj\le \portij+1\left[\ix=\pitemij\right]$. Therefore: 
        \[\sum_{j=1}^n \matr_{ij}=\sum_{j=1}^n \pgixj\le \sum_{j=1}^n\left(\portij+1\left[\ix=\pitemij\right]\right)\le 1+\sum_{j=1}^n1\left[\ix=\pitemij\right]=n\cdot \cuqix=s_i\]   
        where the last inequality follows from the definition of a portion matrix.

        \textit{Requirement 2:} Let $j\in[n]$ be some index (corresponding to a column in $\matr$). Then the validity of the construction of the partial allocation $\pga{j}$ shown in \Cref{lem:partial-valid} implies that:
        \[\sum_{i\in\N}\matr_{ij}=\sum_{i\in\N}\pgixj\le 1\]
        as required.
    
        \textit{Requirement 3:} The sum of row targets is $\sum_{i\in\N}s_i=\sum_{i\in\N}n\cdot \cuqix$. Since $\cuqalloc$ is a full fractional allocation, it must fully allocate $\ix$ across all agents, that is $\sum_{i\in\N} \cuqix=1$, proving that $\sum_{i\in\N}s_i=n$, as required.
    \end{proof}
    
    \Cref{claim:row-column-targets} implies that we can apply \Cref{claim:greedy-completion} to obtain a new matrix $\matr'\in {[0,1]}^{n\times n}$ such that the target row and column sums are reached exactly. Now, simply set $\dfixj=\matr'_{ij}$ for every $i\in\N,j\in[n]$. Note that for every  $i\in\N,j\in[n]$ this implies $\dfixj=\matr'_{ij}\ge\matr_{ij}=\pgixj$. Because the columns of $\matr'$ sum to $1$, for every $j\in[n]$ item $\ix$ is fully allocated in $\dfixj$ among all agents. For every agent $i\in\N$, because the $i$-th row of $\matr'$ sums to $n\cdot \cuqix$, we have $\frac{1}{n}\sum_{j=1}^n\dfixj=\cuqix$.
    
    Altogether, by applying this construction for each item $\ix\in\M$ in turn, we obtain full fractional allocations $\dfa{1}, \ldots, \dfa{n}$ such that for every $j\in[n]$ fractional allocation $\dfa{j}$ completes $\pga{j}$ and the mean $\sum_{j=1}^n\dfa{j}$ equals $n\cdot \cuqix$.
    
    We now show the final requirement. Let $j\in[n]$, $i\in\N$, $\ix\in\M$ so that $\ix\notin \fsetij$ ($\ix\in \Yj\cup \Yip$), we show that $\dfixj\in \{0,1\}$. If $\ix=\pitemij$ then $\pgixj=1$, and so since $\dfixj\ge\pgixj$, also $\dfixj=1$. Otherwise, $\ix \in \left(\Yj\cup \Yip\right)\setminus\{\pitemij\}$, then $\pgixj=0$. If $\ix\in \Yj\setminus\{\pitemij\}$, then $\sum_{i'\in\N}\pga{j}_{i',\ix}=1$, and so $\dfixj$ cannot be strictly greater than $\pgixj=0$, otherwise $\sum_{i'\in\N}\dfa{j}_{i',\ix}> 1$. Otherwise $\ix\in \Yip\setminus\{\pitemij\}$, then $\sum_{j'\in[n]}\pga{j'}_{i,\ix}=\sum_{j'\in[n]}1\left[x=\pitem{i}{j'}\right]=n\cdot \cuqix$, and so $\dfixj$ cannot be strictly greater than $\pgixj=0$, otherwise $\sum_{j'\in[n]}\dfa{j'}_{i,\ix}> n\cdot \cuqix$. In both cases we get that $\dfixj=\pgixj=0$.
\end{proof}

\subsection{Fairness Guarantees and Support Size Reduction}
\label{subsec:ex-post-fairness}

We first present the faithful implementation technique, as presented in \cite{BEF-22}. Restated in our notation:

\begin{lemma}
    \label{lem:faithful}
    There exists a polynomial time \emph{ordinal} algorithm that, 
    given a profile of $n$ strict ordinal preferences $(\en{1},\ldots,\en{n})$ that is consistent with an additive profile $\valprof \in \Vaddprofs$, and a fractional allocation $\falloc\in\fracallocs$, outputs a distributional allocation $\D$ that induces $\falloc$, with support of polynomial size in $n$ and $m$, such that for every agent $i\in\N$, the difference in value between any two ex-post allocations is at most the value of one item $\ix\in\M$ which is allocated to agent $i$ strictly fractionally in $\falloc$ ($0<\fix<1$).

    Moreover, if every strictly fractional value in $\falloc$ is equal to $\frac{1}{k}$ for some integer $k$ then the support has size exactly $k$.
\end{lemma}

{Note the above implies in particular that agent $i$ is guaranteed at least her ex-ante value $\val_i(\falloc_i)$ up to the value of one strictly fractional good.

In the above presentation we add two additions to this result which were not stated explicitly \cite{BEF-22}, but follow immediately from the proof. The first is that the algorithm is purely ordinal, which follows immediately from the construction. The second is the $k$ bound on the size of the support, which also follows directly from the construction, and is used exclusively in \Cref{sec:two-agents} in the proof of \Cref{thm:two-agents}.}

The following lemma uses the faithful implementation technique to generate distributional allocations which implement the fractional allocations $\dfa{1},\ldots,\dfa{n}$ as defined in the previous section,  with an ex-post guarantee on the value each agent receives.

\begin{lemma}
    \label{lem:gen-ex-post}
    There exists a polynomial time {ordinal} algorithm that receives as input:
    \begin{enumerate}
        \item An ordering $\per\in\PerN$ over  $n$ agents.
        \item A profile of $n$ strict ordinal preferences $(\en{1},\ldots,\en{n})$ that is consistent with {$\valprof\in\Vaddprofs$.} 
        \item A portion matrix $\portmat\in\portmats$, with corresponding partial fractional allocations $\pga{1},\ldots,\pga{n}$. 
        \item The fractional allocations $\dfa{1},\ldots,\dfa{n}$ generated by the algorithm in \Cref{lem:completion} for input $\per,(\en{1},\ldots,\en{n}),\portmat$. 
    \end{enumerate}
    Outputs distributional allocations $\Dnum{1},\ldots,\Dnum{n}$, so that $\Dj$ has support of polynomial size in $n$ and $m$ and induces $\dfa{j}$ for every $j\in[n]$, and additionally, for every allocation $\alloc\in \supp{\Dj}$ and agent $i\in\N$:
    \[\val_i(\allocset_i)\ge \max\left(\val_i(\eiat{\rij})\;,\;\portij\cdot \val_i(\fsetij)\right)\]
\end{lemma}
\begin{proof}
    The algorithm is straightforward, and simply applies the faithful implementation technique (\Cref{lem:faithful}) on $\dfa{j}$ to obtain (in polynomial time) a distributional allocation $\Dj$ that induces $\dfa{j}$, with the ex-post guarantees written in the lemma. The algorithm then outputs $\Dj$.
    
    It remains to show that this construction achieves the required ex-post guarantees. Fix an agent $i\in\N$, we will prove the claim for agent $i$. For brevity we denote $y=\eiat{\rij}$ and $\gamma=\max\left(\val_i(y)\;,\;\portij\cdot \val_i(\fsetij)\right)$. Let $\alloc\in\supp{\Dj}$, we must prove that $\val_i(\allocset_i)\ge \gamma$.

    In the partial allocation $\pga{j}$, agent $i$ receives her $\rij$-best item $y=\eiat{\rij}=\pitemij$ with fraction $1$, her ``fractional'' items $\ix\in \fsetij$ with fraction $\portij$, and all remaining goods $\ix\in \M\setminus(\fsetij\cup \{y\})$ with fraction $0$. Therefore, her ex-ante value for $\pga{j}$ is $\val_i(\pga{j})=\val_i(y)+\portij\cdot \val_i(\fsetij)$. Therefore, because $\dfa{j}$ completes $\pga{j}$, we have $\val_i(\dfa{j})\ge\val_i(y)+\portij\cdot \val_i(\fsetij)$. We now separate into cases:

    If $\val_i(y)\ge \gamma$, then because $\Dj$ induces $\dfa{j}$ and $\dfa{j}_{i,y}=\pga{j}_{i,y}=1$, every allocation in the support of $\Dj$ must allocate $y$ to agent $i$, meaning $y\in \allocset_i$, and so $\val_i(\allocset_i)\ge \val_i(y)\ge \gamma$ as required.

    Otherwise, assume $\val_i(y)< \gamma$, then $\gamma=\portij\cdot\val_i(\fsetij)$. By \Cref{lem:completion}, a good $\ix\in\M$ can only be allocated to agent $i$ with strict fraction in $\dfa{j}$ ($0<\dfixj<1$)  if $\ix\in\fsetij$. Note that such good must be ordered in $\ei$ later (worse) than $y=\eiat{\rij}$, otherwise agent $i$ would have picked the good $\ix$ over $y$. We now make use of the ex-post guarantees of \Cref{lem:faithful}, that $\allocset_i$ is guaranteed to be worth to agent $i$ at least her ``ex-ante'' value $\val_i(\dfa{j}_i)$ up to one strictly fractional allocated good. Therefore:
    \[\val_i(\allocset_i)\ge \val_i(y)+\portij\cdot \val_i(\fsetij)-\max_{x\in \fsetij}\val_i(\ix)\ge \portij\cdot \val_i(\fsetij)=\gamma\]
    as required.
\end{proof}

{
In order to get \Cref{lem:implement-export} from \Cref{lem:gen-ex-post}, we need the intermediary result \Cref{lem:union-picking-size}. Before showing this result, we first prove the following general result for picking orders, showing that in a subsequence of a picking order, agents will pick a subset of the items picked in the full ordering.
}

\begin{claim}
    \label{claim:subsequence-picking}    
    {Let $\per\in\PerN$ be an ordering of the agents $\omega_1,\ldots,\omega_n$, and let $A=\{a_1,\ldots,a_n\}$ denote the items chosen in the corresponding picking order. Then for any subsequence $\omega_{i_1}, \ldots, \omega_{i_k}$ for indices $1\leq i_1<\cdots<i_k\leq n$, letting $B=\{b_1,\ldots,b_k\}$ denote the set of items chosen by the picking order $\omega_{i_1},\ldots,\omega_{i_k}$, it holds that $B\subseteq A$.}
\end{claim}
\begin{proof}
    We prove by induction on index $j$ from $1$ to $k$ that $\{b_1,\ldots,b_j\}\subseteq\{a_1,\ldots,a_{i_j}\}$.
    
    \emph{Base Case:} $j=1$. Agent $\omega_{i_1}$ can pick any item in $\M$, and so her chosen item $b_1$ will be her favorite item. Assume for contradiction that $b_{1}\notin\{a_1,\ldots,a_{i_1}\}$. Then in the original picking order agent $\omega_{i_1}$ would have also picked $b_1$ over her chosen item $a_{i_1}$, a contradiction.

    \emph{Induction Case:} Let $j>1$, assume that $\{b_1,\ldots,b_{j-1}\}\subseteq\{a_1,\ldots,a_{i_{j-1}}\}$. Then the item $a_{i_j}$ must be available for selection by agent $\omega_{i_j}$ in the subsequence. Assume for contradiction that $b_j\notin \{a_1,\ldots,a_{i_{j-1}},a_{i_j}\}$. Then in the original picking order agent $\omega_{i_j}$ would have also been able to pick $b_j$, and she would have also preferred $b_j$ over $a_{i_j}$, a contradiction.
\end{proof}

We remark that \Cref{claim:subsequence-picking} holds more generally, it holds for picking ``sequences'' (and not only picking orders) in which an agent can appear in the sequence multiple times.

The following result, which follows from \Cref{claim:subsequence-picking}, allows us to bound the size of the set of ``zeroed goods''.

\begin{lemma}
    \label{lem:union-picking-size}
    Fix an ordering of the agents $\per\in\PerN$ and strict ordinal preference profile $(\en{1},\ldots,\en{n})$.
    For every $i\in\N$ and $j\in[n]$ it holds that $\left|\zsetij\right|\leq 2n-\rij-1$.
\end{lemma}

\begin{proof}
    Let $a_1,\ldots,a_n$ denote the agents ordered by $\per$. Then for every $j\in[n]$, in $\percyc{j}$ agents are ordered $a_{j},\ldots,a_n,a_1,\ldots,a_{j_{-1}}$.  By definition, $\zsetij=\left(\Yj\cup \Yip\right)\setminus \left\{\eiat{\rij}\right\}$, where $\Yj$ and $\Yip$ are the sets of items selected by the agents in the cyclic picking orders $\percyc{j}$ and $\percyc{\ip}$ respectively. Agent $i$ picks in position $\percycij$ in the ordering $\percyc{j}$, so she receives good $\eiat{\rij}$ of rank at most $\rij\le \percycij$.  Consider the sequence starting at index $\ip$: the first $\percyc{\ip}_{a_j}-1$ agents will be $a_{\ip},\ldots,a_{j_{-1}}$, followed by the $\percycij$ agents $a_j,\ldots,a_{\per_i}=i$. Therefore, the first $\percycij$ agents in the ordering $\percyc{j}$ are a subsequence (more precisely a suffix) of the full ordering of agents starting at $\ip$. Therefore, by \Cref{claim:subsequence-picking} the first $\percycij$ items selected to $\Yj$ are also within $\Yip$. Thus $\left|\Yj\cap \Yip\right|\geq \percycij$, and so:    
    \[\left|\zsetij\right|=\left|\left(\Yj\cup \Yip\right)\setminus \{\eiat{\rij}\}\right|=\left|\Yj\cup\Yip\right|-1\]
    \[=\left|\Yj\right|+ \left|\Yip\right|-\left|\Yj\cap \Yip\right|-1\leq 2 n-\percycij-1\le 2n-\rij-1
    \]
    as required. A visualization of the proof is shown in \Cref{fig:union-picking-size}. 
\end{proof}

\begin{figure}[htbp]
    \centering

        \begin{tikzpicture}[
        scale=1, transform shape,
        font=\normalsize,
        block/.style={
            draw=black!50,
            semithick,
            rectangle,
            minimum height=1cm,
            align=center,
            rounded corners=2pt
        },
        disjoint/.style={
            block,
            fill=gray!5
        },
        common/.style={
            block,
            fill=blue!5,
            draw=blue!40!black, 
            thick
        },
        brace/.style={
            draw=black!70,
            thick,
            decorate,
            decoration={brace, amplitude=6pt, raise=3pt}
        },
        brace mirror/.style={
            draw=black!70,
            thick,
            decorate,
            decoration={brace, amplitude=6pt, mirror, raise=3pt}
        }
    ]

        \def\blockAwidth{3cm} 
        \def\blockBwidth{3cm} 
        \def\rowsep{1.3cm}

        \node[common, minimum width=\blockAwidth] (r1a) at (0, 0) {$a_{\ip} \; , \; \dots \; , \; a_{j_{-1}}$};
        \node[common, minimum width=\blockBwidth, anchor=west] (r1b) at (r1a.east) {$a_j \; , \; \dots \; , \; a_{\per(i)}=i$};
        
        \node[common, minimum width=\blockBwidth, anchor=north west] (r2a) at ([yshift=-\rowsep]r1b.south west) {$a_{j} \; , \; \dots \; , \; a_{\per(i)}=i$};
        \node[disjoint, minimum width=\blockAwidth, anchor=west] (r2b) at (r2a.east) {$a_{\ip} \; , \; \dots \; , \; a_{j_{-1}}$};

        \draw[densely dotted, thick, blue!50!black] (r1a.south west) -- (r2a.north west);
        \draw[densely dotted, thick, blue!50!black] (r1b.south east) -- (r2a.north east);

        \coordinate (LabelX) at ([xshift=-0.3cm]r1a.west);
        \node[anchor=east, font=\bfseries] at (LabelX |- r1a.west) {$\Yitems{\ip}$};
        \node[anchor=east, font=\bfseries] at ([xshift=-0.3cm]r2a.west) {$\Yitems{j}$};

        \draw[brace] (r1a.north west) -- (r1a.north east)
            node[midway, above=10pt] {$n - \percycij$ agents};

        \draw[brace] (r1b.north west) -- (r1b.north east)
            node[midway, above=10pt] {$\percycij$ agents};

        \draw[brace mirror] (r2a.south west) -- (r2a.south east)
            node[midway, below=10pt] {$\percycij$ agents};

        \draw[brace mirror] (r2b.south west) -- (r2b.south east)
            node[midway, below=10pt] {$n - \percycij$ agents};

    \end{tikzpicture}
        
    \caption{Alignment of the cyclic picking orders. The first $\percycij$ agents in the ordering $\percyc{j}$ perfectly match the last $\percycij$ agents in the ordering $\percyc{\ip}$. As indicated by the dotted lines mapping the boundaries, by \Cref{claim:subsequence-picking} the items picked by this common subsequence in $\Yj$ are a subset of the items picked by the entire picking order $\Yip$, thus proving that $\left|\Yj\cap \Yip\right|\geq \percycij$.}
    \label{fig:union-picking-size}
\end{figure}
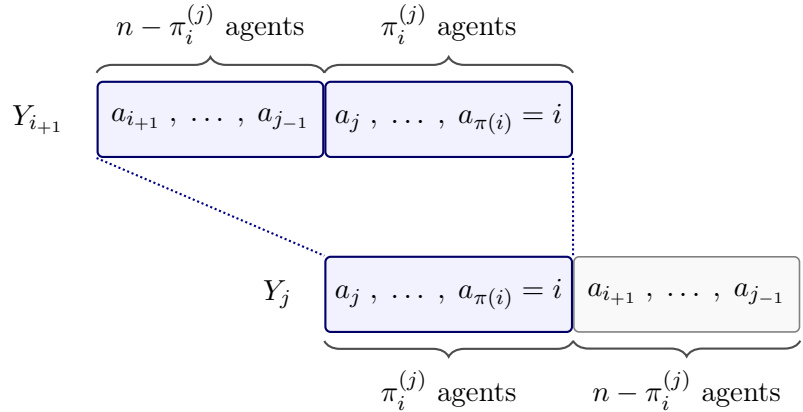

Additionally, we use the following support-size reduction, based on standard techniques. For completeness, we sketch its proof.

\begin{lemma}
    \label{lem:support-reduction}
    There exists a polynomial-time algorithm that for an $n$ agents, $m$ goods setting, given a distributional allocation $\D=\{(\pk,\allock)\}_{k\in[K]}$ with induced fractional allocation $\fracD$, computes a new distributional allocation $\Dprime$ which implements $\fracD$ on a support of size at most $n\cdot m$, satisfying $\supp{\Dprime}\subseteq \supp{\D}$.
\end{lemma}

\begin{proof}
Set up a linear program (LP) with variables $q^{(k)}$ that are intended to replace the $\pk$ in the distributional allocation $\D=\{(\pk,\allock)\}_{k\in[K]}$, without changing the induced fractional allocation $\fracD$. As such, the $q^{(k)}$ need to satisfy non-negativity constraints, the constraint $\sum_k q^{(k)}=1$, and at most $n\cdot m$ additional constraints, preserving each $\fracD_{i,x}$. However, given that each $\allock$ is an allocation, the number of independent additional constraints is only $(n-1)m$. (We can drop the constraints for one agent, and they will hold automatically if all other constraints hold). Thus, any basic feasible solution of the LP has at most $(n-1)m + 1 \le n\cdot m$ non-zero variables. Finally, basic feasible solutions of LPs can be computed in polynomial time.
\end{proof}

\subsection{\texorpdfstring{Proof of \Cref{lem:implement-export}}{Proof of~\ref{lem:implement-export}}}
\label{subsec:proving-export}

We are now ready to prove \Cref{lem:implement-export}.

\implementExport*
\begin{proof}
    Apply the algorithm in \Cref{lem:completion}, using the ordering $\per$, the agents' strict ordinal preferences, and the portion matrix $\portmat$, to obtain the partial allocations $\pga{1},\ldots,\pga{n}$ corresponding to $\portmat$, and full fractional allocations $\dfa{1},\ldots,\dfa{n}$ satisfying the properties in the lemma, namely that $\sum_{j=1}^n\dfa{j}=n\cdot \cuqalloc(\valprof,\per)$, and for every $j\in[n]$, it holds that $\dfa{j}$ completes $\pga{j}$.

    We now run the algorithm of \Cref{lem:gen-ex-post} using the ordering $\per$, the agents' strict ordinal preferences, the portion matrix $\portmat$ with corresponding partial allocations $\pga{1},\ldots,\pga{n}$, and the fractional allocations $\dfa{1},\ldots,\dfa{n}$ constructed by \Cref{lem:completion}. The algorithm outputs $n$ distributional allocations $\Dnum{1},\ldots,\Dnum{n}$, satisfying the properties of the lemma.

    We first construct a distributional allocation $\Dprime$ which is the ``mean'' of $\Dnum{1},\ldots,\Dnum{n}$. That is, $\Dprime$ is the distribution which corresponds to first uniformly sampling an index $j\in[n]$, then sampling an ex-post integral allocation from $\Dj$. To compute this distributional allocation explicitly, we simply take each of the $n$ distributions, multiply the probability of each ex-post allocation by $\frac{1}{n}$, and then add together the probabilities for every allocation. Since the support of each is of polynomial size in $n,m$, the combined support is also of polynomial size in $n,m$. Because $\Dj$ induces $\dfa{j}$ for every $j\in[n]$, the fractional allocation induced by $\Dprime$ is $\frac{1}{n} \cdot \sum_{j=1}^n \dfa{j}=\cuqalloc(\valprof,\per)$.

    Finally, using \Cref{lem:support-reduction}, we compute a new distributional allocation $\D$ that also induces $\cuqalloc(\valprof,\per)$, has support of size at most $n\cdot m$, and satisfies $\supp{\D}\subseteq\supp{\Dprime} =\bigcup_{j=1}^n\supp{\Dj}$. The algorithm then outputs $\D$ {and the $n$ distributional allocations $\Dnum{1},\ldots,\Dnum{n}$}. 
    
    All that remains to be shown is the ex-post value guarantee holds for these $\Dnum{1},\ldots,\Dnum{n}$. Fix $j\in[n]$ and let $\alloc\in\supp{\Dj}$. Fix an agent $i\in \N$. The construction of \Cref{lem:gen-ex-post} guarantees that:
    \[\val_i(\allocset_i)\ge \max\left(\val_i(\eiat{\rij})\;,\;\portij\cdot \val_i(\fsetij)\right)\]

    Letting $r=\rij\le \percycij$, we use \Cref{lem:union-picking-size} to get that $|\zsetij|\le 2n-r-1$. Then $|\fsetij|=m-|\zsetij|-1\ge m-2n+r$. Then the lowest possible value the set $\fsetij$ could have for agent $i$ is achieved when she receives the $m-2n+r$ worst goods according to $i$'s strict ordering:

    \[\val_i(\fsetij)\ge  \val_i(\M\setminus\{\eiat{1},\ldots,\eiat{2n-r}\})\]
    as required.
\end{proof}

\section{\texorpdfstring{An Ordinal TIE Mechanism Obtaining $\Omega(1/\ln n)\hyp\MMS$}{An Ordinal TIE Mechanism Obtaining Omega(1/ln n)-MMS}}
\label{sec:logn}

In this section we present our first main result (\Cref{thm:logn-apx}), a TIE mechanism which guarantees
$\frac{1}{H_{n-1}+2}> \frac{1}{(\ln n) +3 }$ approximation to the $\MMS$ ex-post. Moreover, the mechanism is ex-ante proportional, and thus the mechanism is also a ``Best-of-Both-Worlds'' mechanism.

Recall that the previous best ex-post $\MMS$ approximation by a TIE mechanism is the $1/n$ guarantee of \cite{BT-24}. Since $H_{n-1} + 2 = \Theta(\ln n)$, our result provides an exponential improvement over the previous bound as the number of agents grows large. Furthermore, this improvement takes effect even for small numbers of agents, as $1/(\harmonic{n-1}+2)$ strictly surpasses the $1/n$ bound for all $n \geq 4$. For the case of $n=2$, we later provide a separate, refined implementation, that achieves an even better fraction of $2/3$ (see \Cref{sec:two-agents}). Consequently, our work strictly improves upon the prior $1/n$ benchmark for all values of $n$, except $n=3$.

\lognapx*
We prove the ex-post guarantee with respect to the $\TPS$ (\Cref{def:TPS}), which is at least as large as the $\MMS$ (\Cref{{lem:share-inequality}}).
We use the following basic property of the $\TPS$ to derive \Cref{lem:tps-r-monotone}, which is used in the proof of the theorem.

\begin{claim}
    \label{claim:tps-monotonicity}
    For any $n,\M$, additive valuation $\val\in\Vadd$, and subset $S\subseteq \M$, for any good $\ix\in S$ it holds that $\TPS_n(\val, S)\leq \TPS_{n-1}(\val, S\setminus\{\ix\})$.
\end{claim}

\Cref{claim:tps-monotonicity} follows directly from \Cref{def:TPS}, and implies the following useful property.

\begin{lemma}
    \label{lem:tps-r-monotone}
    For any $n,\M$, and an additive valuation $\val\in\Vadd$, the remaining value to $\val$ after removing any $r$ goods is at least $n-r$ times her $\TPS$. That is, for any $S\subseteq \M$ of size $|S|=r\le n$, we have $\val(\M\setminus S)\ge (n-r)\cdot \TPS_n(\val,\M)$.
\end{lemma}
\begin{proof}
Recall that by \Cref{lem:share-inequality} the $\TPS$ is at most the proportional share, therefore
\[\val(\M\setminus S) = (n-r)\cdot \PROP_{n-r}(\val, \M\setminus S) \geq (n-r)\cdot \TPS_{n-r}(\val, \M\setminus S)\]
Now, by applying \Cref{claim:tps-monotonicity} $r$ times we get $\TPS_{n-r}(\val, \M\setminus S)\geq \TPS_n(\val,\M)$, proving that  $ \val(\M\setminus S) \ge (n-r)\cdot \TPS_n(\val,\M)$, as required.
\end{proof}

We now prove \Cref{thm:logn-apx}. 

\begin{proof} [Proof of \Cref{thm:logn-apx}]
    Fix an arbitrary ordering $\per\in\PerN$. We design a mechanism which implements the cyclic-$\per$-unit-quota fractional mechanism. Let $\portmat\in\portmats$ be the portion matrix defined $\portij=\frac{1}{(n-\percycij+1)\cdot \harmonic{n}}$ for every $i\in\N,j\in[n]$. We note that this construction is valid, that is, columns and rows sum to at most $1$, as for every fixed $i\in\N$:
    \[\sum_{j=1}^n\portij
    =\frac{1}{\harmonic{n}}\cdot\sum_{j=1}^n\frac{1}{n-\percycij+1}
    =\frac{1}{\harmonic{n}}\cdot\sum_{k=1}^n\frac{1}{k}=1\]
    similarly, it holds that $\sum_{i=1}^n\portij =1$ for every $j\in [n]$.

    The statement of the theorem presents it as a direct-revelation mechanism, that is, it receives as input a valuation profile $\valprof\in \Vaddprofs$. Yet, the mechanism itself is purely ordinal, so its input can be any profile of $n$ strict ordinal preferences $(\en{1},\ldots,\en{n})$ that is consistent with $\valprof$. A direct-revelation mechanism can generate such strict ordering by breaking ties in each valuation lexicographically.

    We now apply the algorithm in \Cref{lem:implement-export} on $\per$, the strict ordinal preferences, and $\portmat$, to obtain a distributional allocation $\D$ (and $n$ additional distributional allocations $\Dnum{1} ,\ldots, \Dnum{n}$). The mechanism simply outputs $\D$.

    First, note all steps above are done in polynomial time. \Cref{lem:implement-export} implies that the generated $\D$ has support of size at most $n\cdot m$, and that $\D$ induces $\cuqalloc(\valprof,\per)$. Thus overall our mechanism induces the cyclic-$\per$-unit-quota fractional mechanism, and so by \Cref{lem:cuq-tie} is TIE and proportional ex-ante. It remains to show the ex-post fairness guarantee. Fix $\apx=\frac{1}{\harmonic{n-1}+2}$.

    Let $\alloc\in\supp{\D}$, then since $\supp{\D}\subseteq\bigcup_{j=1}^n\supp{\Dj}$ there must exist $j\in[n]$ such that $\alloc\in \supp{\Dj}$. Fixing an agent $i\in \N$, \Cref{lem:implement-export} guarantees that there exists $r\le \percycij$ such that:
    \[\val_i(\allocset_i)\ge \max\left(\val_i(\eiat{r})\;,\;\portij\cdot \val_i(\M\setminus\{\eiat{1},\ldots,\eiat{2n-r}\})\right)\]

    If $\val_i(\eiat{r})\ge\apx\cdot\TPS_n(\val_i)$ then we are done. Otherwise, we get $\val_i(\eiat{k})<\apx\cdot\TPS_n(\val_i)$ for every $k\ge r$. Using this in conjunction with \Cref{lem:tps-r-monotone}:
    \[\val_i(\M\setminus\{\eiat{1},\ldots,\eiat{2n-r}\})\ge \val_i(\M\setminus\{\eiat{1},\ldots,\eiat{r-1}\})-\sum_{k=r}^{2n-r}\val_i(e_k)\]
    \[\ge (n-r+1)\cdot\TPS_n(\val_i)-(2n-2r+1)\cdot \apx\cdot \TPS_n(\val_i)\]

    Therefore:
    \[\frac{\val_i(\allocset_i)}{\TPS_n(\val_i)}\ge \portij\cdot \frac{\val_i(\M\setminus\{\eiat{1},\ldots,\eiat{2n-r}\})}{\TPS_n(\val_i)}\ge \frac{(n-r+1)-(2n-2r+1)\cdot \apx}{(n-\percycij+1)\cdot\harmonic{n}}\]

    Using $1\leq r\le \percycij$ and $\apx=\frac{1}{\harmonic{n-1}+2}=\frac{1}{\harmonic{n}-\frac{1}{n}+2}$:

    \[\frac{\val_i(\allocset_i)}{\TPS_n(\val_i)}\ge \frac{1-(2-\frac{1}{n-r+1})\cdot \apx}{\harmonic{n}}\ge\frac{1-(2-\frac{1}{n})\cdot \frac{1}{\harmonic{n}-\frac{1}{n}+2}}{\harmonic{n}} =\frac{1}{\harmonic{n}-\frac{1}{n}+2}=\apx\]
    as required.
\end{proof}

\section{\texorpdfstring{A $(1-\varepsilon(n))$-TIE Mechanism Obtaining $\Omega(1/\log\log n)\hyp\MMS$}{A (1-eps(n))-TIE Mechanism Obtaining Omega(1/log log n)-MMS}}
\label{sec:loglogn}

In this section, we present our second main result, which is a mechanism that uses a small amount of cardinal information on top of the ordinal information, and achieves an exponentially better $\MMS$ approximation (with a slight relaxation of incentives). 
Throughout the section we use $\log$ to denote the base 2 logarithm $\log_2$. Specifically, we present a $(1-\varepsilon(n))$-TIE mechanism which guarantees $\Omega(1/\log\log n)$-$\MMS$ ex-post, where $\varepsilon(n)$ is a negligible function in the number of agents $n$. 

\loglognapx*

We now provide a brief overview of our approach for this proof. In \Cref{subsec:deficiencies}, we introduce the notion of an agent's ``$\defapx$-deficiency'', defined as the number of items among her top $n$ choices that do not satisfy the desired $\defapx$ fraction of her $\TPS$. On top of the ordinal information, our mechanism makes use of only this minimal cardinal information: just a single integer between $0$ and $n$ for each agent. (We remark that an agent can instead report this deficiency rounded up to the nearest power of 2, requiring only $O(\log \log n)$ bits, at the cost of a small constant factor in the fairness guarantee). We also assign a ``weight'' to an ordering, which bounds the fractional ``demand'' agents place on the suffix of the allocation based on their positions and reported deficiencies. We use these definitions to prove \Cref{lem:general-weight}, which provides a generic method to implement the cyclic-unit-quota fractional mechanism, yielding an ex-post approximation that depends on the total weight of the ordering of the agents.

In \Cref{subsec:public-defi} we consider an idealized setting where these deficiencies are public knowledge. In \Cref{prop:constant-with-public-defi}, we show that if the mechanism knows the true deficiencies in advance (with no other cardinal information), we can specifically order the agents by non-decreasing deficiency. For this fixed ordering, applying our generic method yields a distributional mechanism that is TIE (with respect to the private ordinal preferences) and guarantees every agent at least a $\frac{1}{4}$-$\MMS$ approximation ex-post.

Building on this insight, in \Cref{subsec:epsilon-tie} we address the general setting where the mechanism has no prior information about the valuations, and the deficiencies are private. We utilize the statistical properties of random orderings: it is highly unlikely that agents will be positioned unfavorably (e.g., agents with high deficiencies appearing late in the order) often enough for the weight to be too large (\Cref{lem:random-ordering}). Our mechanism functions by sampling a random ordering. With \emph{overwhelming probability}, this ordering is ``good'', meaning each of its $n$ cyclic shifts has low weight. By bucketing the reported deficiencies into powers of $2$, we can bound this weight to $O(\log \log n)$. This allows us to implement the cyclic-unit-quota mechanism with a target approximation of $\Omega(1/\log \log n)$-$\MMS$.

In the rare event (occurring with negligible probability) that the ordering is ``bad'', we fall back to a non-truthful Best-of-Both-Worlds polynomial time algorithm, which guarantees ex-ante proportionality and a good enough ex-post MMS approximation. We specifically choose \Cref{prop:constant-with-public-defi} as our ``fallback'' mechanism, as the mechanism has the additional advantage of being purely ordinal once the deficiencies are known. Although this fallback step is not exactly truthful (agents might manipulate their reported deficiency to trigger it), it is triggered so rarely that the expected gain from lying is strictly bounded. This dilution of incentive compatibility is negligible, resulting in our overall $(1-\varepsilon(n))$-TIE guarantee.

\subsection{Incorporating Cardinal Information through Deficiencies}
\label{subsec:deficiencies}

We now present two central definitions that we  use to achieve the improved approximation presented in \Cref{thm:loglog-apx}. For $\defapx>0$, the \emph{deficiency} of an agent with an additive valuation is the number of items out of her top $n$ highest value items which have value less than $\alpha$ fraction of her $\TPS$. Notably, this is our only use of cardinal information throughout this entire section. Besides the deficiencies, the rest of the logic relies purely on the agents' strict ordinal preferences.

The use of the $\TPS$ in the definition (instead of the $\MMS$), is mainly for computational purposes, as computing the $\MMS$ is NP-hard (while the $\TPS$  is poly-time computable), and we want the mechanism (or the agents themselves) to be able to compute the deficiencies in polynomial time. Since the $\TPS$ is always at least as large as the $\MMS$ (see \Cref{lem:share-inequality}), any ex-post fairness guarantees we show for the $\TPS$ hold for the $\MMS$ as well.

\begin{definition}
    Given factor $\defapx>0$, item $\ix\in\M$ is \emph{$\defapx$-satisfactory} to agent $i\in\N$ with {additive} valuation $\val_i\in{\Vadd}$ if $\val_i(\ix)\geq \defapx\cdot \TPS_n(\val_i)$.
    
    The \emph{$\defapx$-deficiency} of agent $i\in\N$ with valuation $\val_i\in\Vadd$ is $\defi_i\in [n]\cup\{0\}$ that is $n$ minus the number of $\defapx$-satisfactory items for $i$ (but not less than 0). That is, letting $ S$ denote the set of $\defapx$-satisfactory items to agent $i$ we have $\defi_i=\max(n-| S|, 0)$. 
\end{definition}

Consider the picking order corresponding to some ordering  $\per\in\PerN$ of the agents. We can use the agents' deficiencies to more accurately define the portion of remaining goods agent $i$ ``demands'' (in order to receive her full $\TPS$) based on her position and deficiency, for her to get at least $\defapx\cdot \TPS_n(\val_i)$. If $\per_i \le n - \defi_i$ then agent $i$ is guaranteed to get an $\defapx$-satisfactory item and so she requires $0$. Otherwise, agent $i$ may not get a $\defapx$-satisfactory item, but because she has $n-\defi_i$ such items, the set of all non-$\defapx$-satisfactory items is worth at least $\defi_i$ times her $\TPS$ (and $\MMS$) by \Cref{lem:tps-r-monotone}. We formalize this intuition in the following definition, which assigns to each agent a \emph{weight} for every ordering $\per\in\PerN$.

\begin{definition}
    For a given approximation factor $\defapx>0$ and ordering of the agents $\per\in\PerN$, the \emph{demand}     $\walpha_i(\per)$ imposed on $\per$ by agent $i$ with $\defapx$-deficiency $\defi_i$, is $0$ if her position $\pi(i)$ is less than or equal to $n-\defi_i$, and $\frac{1}{\defi_i}$ otherwise. That is:
    \[\walpha_i(\per)=\begin{cases}
    0 & \per_i \le n - \defi_i \\
    \frac{1}{\defi_i} & \per_i > n - \defi_i 
    \end{cases}\]
    The weight {$\walpha(\per)$} of an ordering $\per$, is {the total demand over all agents:} $\walpha(\per) = \sum_{i\in\N} \walpha_i(\per)$.
\end{definition}

Note that when $\defi_i=0$, we have $\pi(i)\leq n-\defi_i$ for all $\per\in\PerN$, so $\walpha_i(\per)=0$ and $\walpha_i(\per)$ is always well defined, as no division by zero occurs.

In the worst case, the weight of an ordering may grow up to $\harmonic{n}$. For every position $k\in[n]$, the agent $i$ who picks at position $\pi(i)=k$ can {impose a demand} of at most $\frac{1}{n-k+1}$. Therefore for any ordering $\per$, the total weight is bounded by $\walpha(\per)\le \sum_{k=1}^n \frac{1}{n-k+1}=\harmonic{n}$, and this bound can be tight. Yet, maybe surprisingly, regardless of the $\alpha$-deficiencies, the expected weight of a random ordering is~1. This is by linearity of expectations, using the fact that for every agent $i$, the expectation of {$\walpha_i(\per)$} over a random choice of $\per$ is $\frac{\defi_i}{n} \cdot \frac{1}{\defi_i} = \frac{1}{n}$. This motivates a randomized approach. While our $\Omega({1}/{\harmonic{n}})$ approximation from \Cref{thm:logn-apx} automatically assumed the worst-case, by incorporating agents' deficiencies into the construction of the distributional allocation we can make use of the fact that the weight is low in the majority of cases. As we will show in \Cref{lem:random-ordering}, in the overwhelming majority of orderings the weight can be bound much more tightly at $O(\log\log n)$.

Because we implement the cyclic-{$\per$}-unit-quota fractional mechanism, we are interested in the maximum weight of all the $n$ cyclic shifts (see \Cref{def:cyclic-shifts}) of the ordering $\per$.

\begin{definition}
    Fix an additive valuation profile $\valprof\in\Vaddprofs$. For factor $\defapx>0$, the \emph{cyclic weight} $\Wcycapx(\per)$ of ordering $\per\in\PerN$ is the maximum weight of its $n$ cyclic shifts (capped below by $1$), that is $\Wcycapx(\per)=\max(\max_{j\in[n]}(\walpha(\percyc{j}), 1))$. 
\end{definition}

The outer $\max$ with $1$ is used to guarantee that $\Wcycapx$ is never below $1$, as this interferes with the proof of \Cref{lem:general-weight}. Note that a higher weight is worse, so increasing the weight only makes showing results harder, and is therefore without loss.

We now prove the following central lemma, which is used in both the proofs of \Cref{prop:constant-with-public-defi} and \Cref{thm:loglog-apx}. The lemma introduces a ``black-box'' method for designing distributional mechanisms: if for some factor $\defapx>0$ we can guarantee the ordering of the agents will have low enough cyclic weight $\Wcycapx(\per)$ such that $\defapx\le \frac{1}{2+\Wcycapx(\per)-\frac{1}{n}}$ for any input valuation profile, this will lead to a poly-time distributional mechanism which is ex-ante proportional and TIE (by \Cref{lem:cuq-tie}) and $\defapx$-$\TPS$ ex-post. We note that the factor $\defapx>0$ must be fixed ahead of time, and cannot depend on the agents' reported valuations, or in particular, on their reported deficiencies, since the deficiencies are defined using $\defapx$. The proof of the lemma closely follows the proof of \Cref{thm:logn-apx}, using the decomposition technique presented in \Cref{sec:gen-imp}.

\begin{restatable}{lemma}{generalweight}
    \label{lem:general-weight}
    Fix a factor $\defapx>0$. For every additive valuation profile $\valprof\in\Vaddprofs$ and ordering $\per\in\PerN$, if $\defapx\le \frac{1}{2+\Wcycapx(\per)-\frac{1}{n}}$ then the cyclic-$\per$-unit-quota fractional allocation $\cuqalloc(\valprof,\per)$ can be implemented by a distributional allocation which is $\defapx$-$\TPS$. Moreover, the distribution can be computed explicitly in polynomial time and has support of size at most $n\cdot m$.
\end{restatable}

\begin{proof}
    The proof follows a similar structure to the proof of \Cref{thm:logn-apx}, but with a different portion matrix. Let $\valprof\in\Vaddprofs$ be an additive valuation profile and let $\per\in\PerN$ be an ordering of the agents. Let $\defi_1,\ldots,\defi_i$ denote the agents' $\defapx$-deficiencies. Let $\portmat\in\portmats$ be the portion matrix defined $\portij=\frac{\walpha_i(\percyc{j})}{\Wcycapx(\per)}$ for every $i\in\N,j\in[n]$. Note that both $\walpha_i,\Wcycapx$ depend on the deficiencies, which are cardinal information. This is the only cardinal information we will use in the construction.

    We show this construction is valid, that is, columns and rows sum to at most $1$. Fixing an index $j\in[n]$, by definition of the cyclic weight $\walpha(\percyc{j})\le\Wcycapx(\per)$. Therefore:
    \[
        \sum_{i\in\N}\portij
        =\frac{\sum_{i\in\N}{\walpha_i(\percyc{j})}}{\Wcycapx(\per)}
        =\frac{\walpha(\percyc{j})}{\Wcycapx(\per)}
        \le \frac{\Wcycapx(\per)}{\Wcycapx(\per)}=1
    \]
    as required. Fixing an agent $i\in\N$, if $\defi_i=0$ then $\sum_{j=1}^n\portij=0\le 1$, as required. Otherwise $\defi_i>1$, then among the $n$ cyclic shifts $\percyc{1},\ldots,\percyc{n}$, agent $i$'s demand is $\frac{1}{\defi_i}$ for exactly $\defi_i$ indices (the ones where $i$ is in the last $\defi_i$ positions), and $0$ for all other indices. Therefore (using $\Wcycapx(\per)\ge1$): 
    \[\sum_{j=1}^n\portij
    =\sum_{j=1}^n\frac{\walpha_i(\percyc{j})}{\Wcycapx(\per)}
    =\defi_i\cdot \frac{1}{\defi_i\cdot \Wcycapx(\per)}=\frac{1}{\Wcycapx(\per)}\le 1\]
    as required.
    
    We now apply the algorithm in \Cref{lem:implement-export} on $\per$, the strict ordinal preferences, and $\portmat$, to obtain a distributional allocation $\D$ (and $n$ additional distributional allocations $\Dnum{1} ,\ldots, \Dnum{n}$). The mechanism simply outputs $\D$.

    First, note all steps above are done in polynomial time. \Cref{lem:implement-export} implies that the generated $\D$ has support of size at most $n\cdot m$, and that $\D$ induces $\cuqalloc(\valprof,\per)$. Thus overall our mechanism induces the cyclic-$\per$-unit-quota fractional mechanism, and so by \Cref{lem:cuq-tie} is TIE and proportional ex-ante. It remains to show the ex-post fairness guarantee.

    Let $\alloc\in\supp{\D}$, then since $\supp{\D}\subseteq\bigcup_{j=1}^n\supp{\Dj}$ there must exist $j\in[n]$ such that $\alloc\in \supp{\Dj}$. Fixing an agent $i\in \N$, \Cref{lem:implement-export} guarantees that there exists $r\le \percycij$ such that:
    \[\val_i(\allocset_i)\ge \max\left(\val_i(\eiat{r})\;,\;\portij\cdot \val_i(\M\setminus\{\eiat{1},\ldots,\eiat{2n-r}\})\right)\]

    If $\val_i(\eiat{r})\ge\defapx\cdot\TPS_n(\val_i)$ then we are done. Otherwise, we get $\val_i(\eiat{r})< \defapx \cdot \TPS_n(\val_i)$. By definition of $\defapx$-deficiency we have $\val_i(\eiat{k})\le \defapx \cdot \TPS_n(\val_i)$ for every $k\ge n-\defi_i+1$. Therefore $\eiat{r}$ is not $\defapx$-sufficient to agent $i$, so $\defi_i\ge n-r+1\ge n-\percycij+1$. Using this in conjunction with \Cref{lem:tps-r-monotone}:
    \[\val_i(\M\setminus\{\eiat{1},\ldots,\eiat{2n-r}\})= \val_i(\M\setminus\{\eiat1,\ldots,\eiat{n-\defi_i}\})-\sum_{k=n-\defi_i+1}^{2n-r}\val_i(e_k)\ge \defi_i\cdot\TPS_n(\val_i)-(n-r+\defi_i)\cdot \defapx\cdot \TPS_n(\val_i)\]

    Because $\defi_i\ge  n-\percycij+1$, we have $\portij=\frac{\walpha_i(\percyc{j})}{\Wcycapx(\per)}=\frac{1}{\defi_i\cdot \Wcycapx(\per)}$. Therefore, using $n-r\le \defi_i-1$:
    \[\frac{\val_i(\allocset_i)}{\TPS_n(\val_i)}
    \ge \portij\cdot \frac{\val_i(\M\setminus\{\eiat{1},\ldots,\eiat{2n-r}\})}{\TPS_n(\val_i)}
    \ge \frac{\defi_i-(2\defi_i-1)\cdot \defapx}{\defi_i\cdot\Wcycapx(\per)}\ge \frac{1-(2-\frac{1}{n})\cdot\defapx}{\Wcycapx(\per)}\]

    Using $\defapx\le \frac{1}{2+\Wcycapx(\per)-\frac{1}{n}}$:

    \[\frac{\val_i(\allocset_i)}{\TPS_n(\val_i)} 
    \ge \frac{1-(2-\frac{1}{n})\cdot\defapx}{\Wcycapx(\per)}
    \ge \frac{1-(2-\frac{1}{n})\cdot\frac{1}{2+\Wcycapx(\per)-\frac{1}{n}}}{\Wcycapx(\per)}
    =\frac{1}{2+\Wcycapx(\per)-\frac{1}{n}}\ge \defapx\]
    as required.
\end{proof}

\subsection{Deficiencies as Public Knowledge: a TIE Mechanism Obtaining 1/4-MMS}
\label{subsec:public-defi}

In this section we will make use of \Cref{lem:general-weight} to design a TIE mechanism obtaining $1/4$-$\MMS$ ex-post. In order to do this, we first show there always exists an ordering $\per\in\PerN$ with cyclic weight bounded by $2$.

\begin{restatable}{lemma}{orderedweight}
    \label{lem:ordered-weight}
    Fix a factor $\defapx>0$. Let $\valprof\in\Vaddprofs$ be an additive valuation profile and let $\per\in\PerN$ denote an ordering of the agents so that the $\defapx$-deficiencies are sorted in non-decreasing order. Then $\walpha(\per)\leq1$, and $\Wcycapx(\per)\leq 2$.
\end{restatable}

\begin{proof}
    W.l.o.g. rename the agents so $\per$ orders them $1,\ldots,n$. Then by definition their deficiencies are non-decreasing: $\defi_{1} \le \dots \le \defi_{n}$. We first establish a general bound for any contiguous subsequence of these sorted agents.
    
    \begin{claim}
        \label{claim:helper-ordered-weight}
        Let $\per$ and $\defi_1,\ldots,\defi_n$ as above. Let $1\le l\le r\le n$ be two agents. Then for any index $j$ such that $j\le l$ or $j> r$, the total demand imposed on $\percyc{j}$ by the subsequence $l,\ldots,r$ is at most $\sum_{i=l}^r \walpha_{i}(\percyc{j}) \le 1$.
    \end{claim}
    \begin{proof}
        Note that because $j \le l$ or $j > r$, the sequence $l, \ldots, r$ is mapped to a strictly increasing block of positions $\percyc{j}$. Specifically, if $j\le l$ then $\percycij = i - j + 1$, if $j>r$ then $\percycij = i - j + 1 + n$. Recall that an agent $i$ imposes a demand $\walpha_i(\percyc{j}) = \frac{1}{\defi_i}$ only if their position satisfies $\percycij > n - \defi_i$, otherwise, their demand is $0$. Because the positions are contiguous, as $i$ increases, the position $\percycij$ strictly increases, while $\defi_i$ is non-decreasing  by the assumption. Thus, the condition $\percycij > n - \defi_i$ is monotonic.
        
        If the condition never holds for any agent in the subsequence, the total demand is $0 \le 1$, and we are done. Otherwise, let $i^*\in\{l, \ldots, r\}$ be the smallest index for which $\percyc{j}_{i^*} \ge n - \defi_{i^*}+1$. For all $l \le i < i^*$, the demand is $0$. For all $i^* \le i \le r$, the demand is at most $\frac{1}{\defi_{i^*}}$ because $\defi_i \ge \defi_{i^*}$. The total demand of the subsequence is therefore bounded by:
        \[ \sum_{i=l}^r \walpha_i(\percyc{j}) = \sum_{i=i^*}^r \frac{1}{\defi_{i}} \le \sum_{i=i^*}^r \frac{1}{\defi_{i^*}} = \frac{r - i^* + 1}{\defi_{i^*}}\le \frac{\percyc{j}_{r}-\percyc{j}_{i^*}+1}{ n - \percyc{j}_{i^*}+1}\le 1 \]
        as required
    \end{proof}

    We now apply this claim to prove the lemma. For the initial ordering $\per$, applying \Cref{claim:helper-ordered-weight} (with $j=1$) on the entire sequence $1,\ldots,n$ proves that $\walpha(\per)=\sum_{i\in\N} \walpha_{i}(\per) \le 1$, as required.

    To bound the cyclic weight $\Wcycapx(\per)$, let $j\in\{2,\ldots,n\}$ and consider the cyclic shift $\percyc{j}$. This shift splits the sorted agents into two ordered blocks: the first block contains agents $j,\ldots,n$ and the second block contains agents $1, \dots, j-1$. Within each block, agents are ordered increasingly (so their deficiencies are ordered as well). Therefore, applying \Cref{claim:helper-ordered-weight} twice, once for each block:
    \[\walpha(\percyc{j})=\sum_{i=j}^n \walpha_{i}(\percyc{j})+\sum_{i=1}^{j-1} \walpha_{i}(\percyc{j}) \le 1+1=2\]
    This proves that $\Wcycapx(\per) =\max(\max_{j\in[n]}\left(\walpha(\percyc{j})\right),1) \le 2$.
\end{proof}

We now make use of the combination of \Cref{lem:general-weight} and \Cref{lem:ordered-weight} to design a TIE mechanism which makes use of the public information to obtain a $1/4$ approximation of the $\MMS$ ex-post.

\begin{restatable}{proposition}{constantpublicdeficiency}
    \label{prop:constant-with-public-defi}
    For any factor $0<\defapx\le\frac{1}{4}$, when agents' \emph{true} $\defapx$-deficiencies $\defi_1,\ldots,\defi_n$ are given to the mechanism, there exists a 
    distributional ordinal mechanism which is \emph{truthful-in-expectation (TIE)}, \emph{ex-ante proportional}, and every supporting allocation is \emph{$\defapx$-$\TPS$}. Moreover, the distribution has support of size at most $n\cdot m$, and can be computed explicitly in polynomial time. 
\end{restatable}

\begin{proof}
    Fix approximation factor $0<\defapx \le \frac{1}{4}$. The mechanism receives as input from the agents an additive valuation profile $\valprof\in\Vaddprofs$, and the agents' true $\defapx$-deficiencies $\defi_1, \ldots, \defi_n$. Henceforth the mechanism will only make use of the agents' strict orderings over the goods, with ties broken lexicographically (and their $\defapx$-deficiencies, {which have already been given}).

    Let $\per \in \PerN$ be an ordering of the agents such that their $\defapx$-deficiencies are in non-decreasing order. \Cref{lem:ordered-weight} shows this ordering guarantees that the cyclic weight is bounded by $\Wcycapx(\per) \le 2$. Then $\frac{1}{2+\Wcycapx(\per)-\frac{1}{n}}\ge \frac{1}{4}\ge \defapx$, so we can apply \Cref{lem:general-weight} using the factor $\defapx$ and the ordering $\per$, to obtain a distributional allocation $\D$ which induces the cyclic-$\per$-unit-quota fractional allocation $\cuqalloc(\valprof,\per)$. The mechanism outputs $\D$.

    We now prove the mechanism satisfies all the requirements of the theorem. The mechanism induces the cyclic-$\per$-unit-quota fractional mechanism at every input profile, so \Cref{lem:cuq-tie} implies the mechanism is TIE and proportional ex-ante. Additionally, \Cref{lem:general-weight} guarantees that the generated distribution $\D$ has support of size at most $n\cdot m$ which consists entirely of allocations which are $\defapx$-$\TPS$. Finally, note that all the steps of the mechanism can be computed in polynomial time.
\end{proof}

\subsection{\texorpdfstring{Proof of \Cref{thm:loglog-apx}}{Proof of Theorem~\ref{thm:loglog-apx}}}
\label{subsec:epsilon-tie}

This section is dedicated to proving \Cref{thm:loglog-apx}. First, we prove in \Cref{lem:random-ordering} that an ordering sampled uniformly at random from $\PerN$ will have weight at most $O(\log\log n)$ with probability $1-n^{-\log n}$. We note that using \Cref{lem:general-weight}, we can use this immediately to construct a TIE mechanism which is $\Omega(1/\log \log n)$-$\TPS$ ex-post with high probability. Using the fact that the gain from a lie is bounded, we show this mechanism can be converted into a $(1-\varepsilon(n))$-TIE mechanism that is  always $\Omega(1/\log \log n)$-$\TPS$.

We begin with the following probabilistic result, which shows that regardless of the deficiencies of the agents, a random sampled ordering will have weight at most $O(\log\log n)$ with high probability.

\begin{restatable}{lemma}{randomOrdering}
    \label{lem:random-ordering}
    For any factor $\defapx>0$, regardless of the $\defapx$-deficiencies of the agents, there is only negligible probability of $\frac{n^{-\log n}}{2n}$ that a uniformly sampled ordering $\per\leftarrow\PerN$ has cyclic-weight $\Wcycapx(\per)$ above $2\log\log n+23$.
\end{restatable}
\begin{proof}
Fix factor $\defapx>0$, let $\defi_1,\ldots,\defi_n$ be the agents' $\defapx$-deficiencies. Partition the agents into $k+1$ groups $I_0,\ldots,I_{k}$ where $k=\ceil{\log_2 (n+1)}$. Group $I_0$ contains agents with $\defi_i=0$, and group $I_j$ ($1\le j\le k$) contains agents with $\defi_i \in [2^{j-1}, 2^j - 1]$.

Let $c$ be a constant to be determined later, and set $t=\floor{c\log\log n}$. For any ordering $\per$, the weight from agents in the small groups $I_0,\ldots,I_t$ is bounded by the sum of demands imposed by agents in the last $2^t-1 \le \floor{\log^c n}$ positions. Since at most one agent from these groups occupies the $d$-th position from the end (contributing at most $\frac{1}{d}$), their total weight is at most:

\[\sum_{i\in \bigcup_{j=0}^{t}I_j}\walpha_i(\per)\le \sum_{d=1}^{\floor{\log^c n}}\frac{1}{d}=\harmonic{\floor{\log^c n}}\le 1+\ln(\log^c n)\le 1+ c\log\log n\]

We now move to bound the weight imposed by agents in the remaining groups $I_{t+1},\ldots,I_{k}$. Fixing a group $j\in\{t+1,\ldots,k\}$, let $s_j=2^j-1$ be the maximum deficiency in $I_j$, and $n_j=|I_j|$. Agents in $I_j$ contribute at most $\frac{1}{2^{j-1}} \le \frac{2}{s_j}$ to the weight, and only if placed in the last $s_j$ positions. Let $X_j$ be the random variable for the number of agents from $I_j$ randomly placed in these last $s_j$ positions. Then the total demand imposed by all agents from groups $I_{t+1},\ldots,I_{k}$ is at most $\sum_{j=t+1}^k \frac{2}{s_j} X_j$.

Notice that $X_j$ is distributed according to a hypergeometric distribution, with population size $n$, with $n_j$ success states, and $s_j$ draws without replacement. The expectation is known to be $\expect{X_j}=s_j\frac{n_j}{n}$. Applying the standard tail bound for the hypergeometric distribution, for any $\delta_j > 0$, the probability that $X_j$ deviates from its expected value by more than $\delta_j s_j$ is at most $e^{-2 \delta_j^2 s_j}$. That is, $\prob{X_j > \left(\frac{n_j}{n}+\delta_j\right) s_j} \le e^{-2 \delta_j^2 s_j}$.

    We set $\delta_j=\frac{2\log n}{\sqrt{s_j}}$ for every $j\in\{t+1,\ldots,k\}$. Applying a union bound over all groups $j$, the probability that \emph{any} $X_j$ exceeds its deviation bound is at most:
    \[ \sum_{j=t+1}^k e^{-2 \left(\frac{4\log^2 n}{s_j}\right) s_j} \le k \cdot e^{-8\log^2 n} \le (2\log n) n^{-8\log n}\]

    Assuming no $X_j$ exceeds the bound, the total demand imposed from these groups is at most:
    \[\sum_{i\in \bigcup_{j=t+1}^{k}I_j}\walpha_i(\per) \le \sum_{j=t+1}^k\frac{2}{s_j} \left(\frac{n_j}{n}+\frac{{2}\log n}{\sqrt{s_j}}\right) s_j = \sum_{j=t+1}^k \frac{2n_j}{n} + 4\log n \sum_{j=t+1}^k\frac{1}{\sqrt{s_j}}\]

    Since $\sum n_j \le n$, the first term is at most $2$. We bound the second sum by evaluating the tail of a geometric series:
    \[ \sum_{j=t+1}^k \frac{1}{\sqrt{s_j}} \le \sum_{j=t+1}^\infty \left(\frac{1}{\sqrt{2}}\right)^{j-1} = \frac{2^{-t/2}}{1 - 1/\sqrt{2}} = (2+\sqrt{2})2^{-t/2}\]

     Substituting $t = \floor{c\log\log n} \ge c\log\log n - 1$ gives $2^{-t/2} \le \sqrt{2} (\log n)^{-\frac{c}{2}}$. Thus, the total weight from these groups is bounded by:
    \[ 2 + 4\log n \cdot \sqrt{2}(2+\sqrt{2})(\log n)^{-c/2} = 2 + 4(2\sqrt{2}+2) (\log n)^{1-c/2}\]

    Setting $c=2$ yields a bound of $2 + 4(2\sqrt{2}+2) \le 22$. The weight of these groups is therefore at most $22$. Combining this with the bound for groups $I_0,\ldots,I_t$, a single random ordering $\per$ satisfies $w^\alpha(\per) \le 23+2\log\log n$ with probability at least $1 - (\log n) n^{-8\log n}$. Finally, the cyclic-weight $\Wcycapx(\per) = \max_{r \in [n]} w^\alpha(\percyc{r})$ considers all $n$ cyclic shifts. Taking a union bound over these $n$ shifts, the probability that \emph{any} shift exceeds the weight limit is at most:
    \[ \prob{\exists r \in [n] : w^\alpha(\percyc{r}) > 23+2\log\log n} \le n \cdot ({2}\log n) n^{-8\log n} \le \frac{n^{-\log n}}{2n}\]
    as required.
\end{proof}

We remark that the analysis of the bound in the lemma is tight. For a fixed factor $\defapx>0$, give each agent a $\defapx$-deficiency uniformly at random, among the values $\{1, 2, \ldots,\ell\}$ for $\ell = \frac{\log n}{2 \log\log n}$. Consider a randomly sampled ordering $\per$. With overwhelming probability, the ordering will include a consecutive decreasing subsequence of length $\ell$ (as there are $n/\ell$ disjoint segments of length $\ell$, and each has probability roughly $\frac{1}{\sqrt{n}}$ of being such a subsequence). Some cyclic shift will move this subsequence to the end, and then we pay the harmonic sum over $\ell$ terms, which is roughly $\log\log n$.

We now use \Cref{lem:random-ordering} in conjunction with \Cref{lem:general-weight} to prove \Cref{thm:loglog-apx}.

\loglognapx*

\begin{proof}
    
    Fix factor $\defapx = \frac{1}{25 + 2\log\log n-\frac{1}{n}}=\Omega\left(\frac{1}{\log\log n}\right)$. The mechanism receives as input from the agents an additive valuation profile $\valprof\in\Vaddprofs$, and computes their $\defapx$-deficiencies $\defi_1, \ldots, \defi_n$. This is the only use of cardinal information: henceforth the mechanism will only use these deficiencies, along with each agent's strict ordering over the goods, with ties broken lexicographically.

    Sample randomly an ordering of the agents $\per \leftarrow \PerN$, and compute the cyclic weight $\Wcycapx(\per)$. We say ordering $\per$ is \emph{good} if $\Wcycapx(\per)\le 23+2\log\log n=\frac{1}{\defapx}-2+\frac{1}{n}$. Note that this implies $\frac{1}{2+\Wcycapx(\per)-\frac{1}{n}}\ge \defapx$. If $\per$ is good, then we can apply \Cref{lem:general-weight} with factor $\defapx$ and the sampled ordering $\per$, to obtain a distributional allocation $\D$ which induces the cyclic-$\per$-unit-quota fractional allocation $\cuqalloc(\valprof,\per)$, then output $\D$. Otherwise, $\per$ is \emph{bad}. In this case we apply the mechanism of \Cref{prop:constant-with-public-defi}, using factor $\defapx$ defined above (note that $\defapx\le \frac{1}{4}$ for all $n$), and output the resulting distributional allocation.

    Note that in the above construction, when $\per$ is bad, the mechanism of \Cref{prop:constant-with-public-defi} reorders the agents by order $\per'$ of non-decreasing deficiencies, and then implements $\cuqalloc(\valprof,\per')$, which is generally different from $\cuqalloc(\valprof, \per)$. Thus overall, this mechanism does \emph{not} implement the cyclic-$\per$-unit-quota fractional allocation.

    We now prove the mechanism satisfies all the requirements of the theorem. First, note all steps above can be done in polynomial time. Also, the mechanism always outputs a distributional allocation which has support of size at most $n\cdot m$. We now show the mechanism is $\defapx$-$\TPS$ ex-post. If the sampled ordering $\per$ is good, the algorithm applies \Cref{lem:general-weight}, and so $\D$ is guaranteed to be $\defapx$-$\TPS$ ex-post. Otherwise, if $\per$ is bad, the algorithm applies \Cref{prop:constant-with-public-defi} with factor $\defapx$. Because always $\defapx < \frac{1}{4}$ \Cref{prop:constant-with-public-defi} guarantees $\defapx$-$\TPS$. In both cases, the mechanism is $\defapx=\Omega\left(\frac{1}{\log\log n}\right)$-$\TPS$ ex-post.

    For ex-ante fairness, as agents are assumed to report truthfully,
    both the cyclic-$\per$-unit-quota fractional allocation for the sampled ordering $\per$, and the mechanism of \Cref{prop:constant-with-public-defi} are ex-ante proportional. As the mechanism always executes one of these two strictly proportional algorithms, it remains ex-ante proportional overall.

    It remains to prove the mechanism is $(1-\varepsilon(n))$-TIE for $\varepsilon(n)=n^{-\log n}$. Let $\nu(n) = \frac{n^{-\log n}}{2n}$. Consider a fixed agent $i\in\N$, a true valuation profile $\valprof$, and an alternative reported valuation $\valp_i$. Let $X_i$ and $X'_i$ denote the random variables for the expected value of agent $i$ when reporting truthfully and falsely, respectively (both for the fixed reports $\val_{-i}$ of the others).

    By \Cref{lem:random-ordering}, the probability that a uniformly sampled ordering $\per$ is bad under the true profile $\valprof$ is at most $\nu(n)$, and under the misreported profile $\valprofp=(\valp_i, \val_{-i})$ is also at most $\nu(n)$. By the union bound, with probability at least $1 - 2\nu(n)$, the sampled ordering $\per$ is good under both reports. Conditioned on $\per$ being good for both reports, the mechanism implements the cyclic-$\per$-unit-quota fractional allocation. Then, by \Cref{lem:cuq-tie}, in this case truthful reporting maximizes the agent's expected utility on this event.

    In the event that the ordering is bad under either report (which occurs with total probability at most $2\nu(n)$), the agent's expected value under a false report is bounded from above by her total value for all items, $\val_i(\M)$. Conversely, because the mechanism is ex-ante proportional, her expected value under truth-telling is bounded from below by her proportional share, $\expect{X_i} \ge \frac{\val_i(\M)}{n}$.
    Therefore, the maximum potential expected gain from misreporting is bounded as follows:
    \[ \expect{X'_i} \le \expect{X_i} + 2\nu(n) \cdot \val_i(\M) \le \expect{X_i} + 2\nu(n) \cdot n \cdot \expect{X_i} = (1 + 2n\cdot \nu(n)) \expect{X_i} \]
    Rearranging this inequality yields:
    \[ \expect{X_i} \ge \frac{1}{1 + 2n\cdot \nu(n)} \expect{X'_i} \ge (1 - 2n\cdot \nu(n)) \expect{X'_i} \]
    Setting $\varepsilon(n) = 2n\cdot \nu(n) = n^{-\log n}$, which is negligible, we establish that the mechanism is $(1-\varepsilon(n))$-TIE, completing the proof.
\end{proof}

\section{A TIE Mechanism Obtaining 2/3-MMS for Two Agents}
\label{sec:two-agents}

We now focus on the special case where there are $n=2$ agents. In the algorithmic setting, where agents' valuations are known, it is well known that it is possible to obtain proportionality ex-ante and full $\MMS$ ex-post using random cut-and-choose. Meanwhile, the best known ex-post $\MMS$ approximation achieved by a TIE mechanism for this case is $\frac{1}{2}$ \cite{BT-24}. This result was obtained with a distributional mechanism that induces the uniform fractional allocation for every input profile. We show in \Cref{prop:n-approximation-impossibility} that an ex-post implementation of the uniform fractional mechanism cannot guarantee a fraction of the $\MMS$ that is larger than $\frac{1}{2}$. As in previous sections, we seek to overcome this impossibility by using the cyclic-unit-quota fractional mechanism, which makes use of agents' reports. In \Cref{sec:logn}, we established a general result of $\frac{1}{\harmonic{n-1}+2}$ utilizing a generic faithful implementation of the cyclic-unit-quota fractional allocation. Unfortunately, for $n=2$ agents, this gives an approximation of only $\frac{1}{3}$. In this section, we develop a new implementation for the case of two agents, obtaining a significantly better ex-post $\MMS$ approximation of $\frac{2}{3}$.

The substantial difference between this approach and the general case lies in the \emph{implementation} technique. Both methods start from the same underlying fractional allocation (CUQ). However, while the general result in \Cref{thm:logn-apx} relies on the inherently ordinal faithful implementation technique (\Cref{lem:faithful}), the implementation in \Cref{thm:two-agents} accounts for the agents' cardinal values.

We now present our main result for the two-agent case.

\twoagents*

We note the theorem holds for the stronger fairness metric of the $\TPS$ as well. Notably, $\frac{2}{3}$ is known to be the best possible approximation of the $\TPS$ for two agents, as there exists an instance with two agents in which no better than $\frac{2}{3}$-$\TPS$ allocation exists \cite{BEF-22}. This leads us to the following observation:

\begin{observation}
    The $\frac{2}{3}$ $\TPS$ approximation of \Cref{thm:two-agents} is best possible for two additive agents.
\end{observation}

Moreover, we prove the following negative result which shows that any randomized (or distributional) mechanism (even non-truthful ones) which induces an ordinal fractional mechanism cannot improve on the $2/3$-$\MMS$ ex-post guarantee of \Cref{thm:two-agents}. In particular, since the cyclic-unit-quota fractional mechanism is ordinal, this proves that improving our $2/3$-$\MMS$ result for two agents requires designing more intricate fractional allocations that incorporate cardinal information.

\begin{restatable}{proposition}{negtwo}
\label{prop:neg-two-agents}  
    Consider the problem of allocating $m$ indivisible goods to two ($n=2$) \emph{additive} agents. Let $\randMech$ be a randomized (or distributional) mechanism which induces fractional mechanism $\fracMech$. If $\fracMech$ is ordinal, then $\randMech$ is at most $\frac{2}{3}$-$\MMS$ ex-post.  
\end{restatable}

The proof of \Cref{prop:neg-two-agents} is deferred to \Cref{app:two-agents}.

\subsection{\texorpdfstring{Proving \Cref{thm:two-agents}}{Proving Theorem~\ref{thm:two-agents}}}

We now move to prove \Cref{thm:two-agents}, the main result of this section. The proof is structured as follows. We first discuss the cyclic-unit-quota fractional mechanism for the special case of two agents (\Cref{obs:2-agents-cuq-fractional} and the accompanying \Cref{fig:2-agents-cuq}). We then move to prove \Cref{thm:two-agents}, by separating into several cases. We first prove some ``easy'' cases (\Cref{lem:easy-cases}), and then the remaining ``hard'' case (\Cref{lem:hard-case}), all together completing the proof.

When there are only two agents, we observe that the cyclic-unit-quota fractional mechanism can be expressed
much more simply, as shown in the following observation and the accompanying \Cref{fig:2-agents-cuq}.

\begin{observation}
    \label{obs:2-agents-cuq-fractional}
    Consider the problem of allocating $m$ indivisible goods to two ($n=2$) \emph{additive} agents. Letting $y_1,y_2$ denote their favorite goods respectively, the cyclic-unit-quota fractional mechanism takes the following form.
    \begin{itemize}
        \item If both agents have the same favorite item $(y_1=y_2)$ then $\cuqalloc(\val_1,\val_2)$ is equal to the uniform fractional allocation, that is $\cuqix=\frac{1}{2}$ for every agent $i\in\{1,2\}$ and good $\ix\in\M$.
        \item If the agents have different favorite goods ($y_1\neq y_2$) then each agent $i\in\{1,2\}$ is fully allocated $y_i$, while the remaining goods are split evenly, that is $\cuqix=\frac{1}{2}$ for every agent $i\in\{1,2\}$ and good $\ix\in\M\setminus\{y_1,y_2\}$.
    \end{itemize}

\end{observation}
   
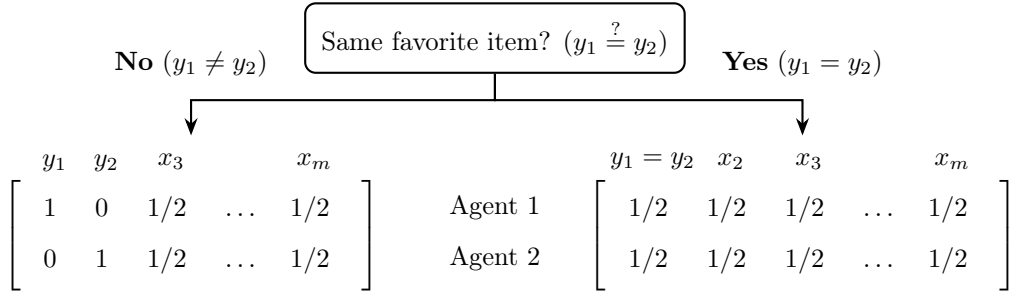
\begin{figure}[htbp]
    \centering
    
    \begin{tikzpicture}[
    >=Stealth, 
    node distance=1cm and 1.4cm, 
    every node/.style={font=\small} 
]

    \node[draw, thick, rectangle, rounded corners, align=center, inner sep=1.3ex] (decision) {Same favorite item? ($y_1 \stackrel{?}{=} y_2$)};

    \node[below left=-0.45cm and 0.35cm of decision, align=center] (no_label) {\textbf{No} ($y_1 \neq y_2$)};
    
    \matrix (M1) [below=1.2cm of no_label, matrix of math nodes, left delimiter={[}, right delimiter={]}, nodes={minimum width=0.8em, align=center}, row sep=0.07cm, column sep=0.14cm] {
        $1$ & $0$ & $1/2$ & \dots & $1/2$ \\
        $0$ & $1$ & $1/2$ & \dots & $1/2$ \\
    };
    \node[above=0.1cm of M1-1-1] {\small $y_1$};
    \node[above=0.1cm of M1-1-2] {\small $y_2$};
    \node[above=0.1cm of M1-1-3] {\small $x_3$};
    \node[above=0.1cm of M1-1-5] {\small $x_m$};
    \node[] at ($(M1-1-1.east -| decision.south)$) {Agent 1};
    \node[] at ($(M1-2-1.east -| decision.south)$) {Agent 2};

    \node[below right=-0.45cm and 0.35cm of decision, align=center] (yes_label) {\textbf{Yes} ($y_1 = y_2$)};
    
    \matrix (M2) [below=1.2cm of yes_label, matrix of math nodes, left delimiter={[}, right delimiter={]}, nodes={minimum width=0.8em, align=center}, row sep=0.07cm, column sep=0.14cm] {
        $1/2$ & $1/2$ & $1/2$ & \dots & $1/2$ \\
        $1/2$ & $1/2$ & $1/2$ & \dots & $1/2$ \\
    };
    \node[above=0.1cm of M2-1-1] {\small $y_1=y_2$};
    \node[above=0.1cm of M2-1-2] {\small $x_2$};
    \node[above=0.1cm of M2-1-3] {\small $x_3$};
    \node[above=0.1cm of M2-1-5] {\small $x_m$};

    \draw[->, thick] (decision.south) -- +(0,-0.35) -| ([yshift=+0.55cm]M1.north);
    \draw[->, thick] (decision.south) -- +(0,-0.35) -| ([yshift=+0.55cm]M2.north);

\end{tikzpicture}

    \caption{Visual representation of the cyclic-unit-quota fractional allocation for two agents. The fractional allocation depends on whether agents agree on the highest-value item.}
    \label{fig:2-agents-cuq}
\end{figure}

It follows from \Cref{obs:2-agents-cuq-fractional} that when there are two agents, the cyclic-$\per$-unit-quota fractional allocation does not depend on ordering $\per$. Intuitively, this occurs because for two agents there are only two possible orderings, and each is the cyclic-shift of the other. Therefore, throughout this section we do not write the ordering $\per$ anywhere.

We now prove our main result for this section.

\twoagents*
\begin{proof}
    Before describing our mechanism, we first separate between several ``easy'' cases and one ``hard'' case.

    \begin{restatable}{definition}{defeasycases}
        \label{def:easy-cases}
        We say valuation profile $(\val_1,\val_2)\in\Vadd^2$ is \emph{easy} if, letting $y_1,y_2\in \M$ denote the agents' favorite items respectively, one of the following three cases occurs:
        \begin{enumerate}
            \item $y_1\neq y_2$.
            \item $y=y_1=y_2$ and $\val_i(y)\le \frac{2}{3}\cdot \TPS_2(\val_i)$ for both $i\in \{1,2\}$.
            \item $y=y_1=y_2$ and $\val_i(y)\ge \frac{2}{3}\cdot \TPS_2(\val_i)$ for both $i\in \{1,2\}$.
        \end{enumerate}
    \end{restatable}
    
    We next show that in each of these cases, we can compute a distributional allocation satisfying our requirements.

    \begin{restatable}{lemma}{easycases}
        \label{lem:easy-cases}
        There exists a polynomial time algorithm that, given an \emph{easy} additive valuation profile $(\val_1,\val_2) \in\Vadd^2$, outputs a distributional allocation $\D$ which induces the cyclic-unit-quota fractional allocation $\cuqalloc(\val_1,\val_2)$, is $\frac{2}{3}$-$\TPS$ ex-post, and is supported on two allocations.
    \end{restatable}

    The first and second cases can be solved by directly applying the faithful implementation technique (see \Cref{lem:faithful}) on the cyclic-unit-quota fractional allocation. The third case can be solved by the distribution which flips a fair coin, and assigns one of the agents the shared favorite good $y$ and the other all remaining goods $\M\setminus\{y\}$. We defer the formal proof of \Cref{lem:easy-cases} to \Cref{app:two-agents}.

    We say valuation profile $(\val_1,\val_2)\in\Vadd^2$ is \emph{hard} if it is not easy, that is, if none of the three cases above holds. For the hard case we present the following  lemma:

    \begin{lemma}
        \label{lem:hard-case}
        There exists a polynomial time algorithm that, given a \emph{hard} additive valuation profile $(\val_1,\val_2) \in\Vadd^2$, outputs a distributional allocation $\D$ which induces the cyclic-unit-quota fractional allocation $\cuqalloc(\val_1,\val_2)$, is $\frac{2}{3}$-$\TPS$ ex-post, and is supported on two allocations.
    \end{lemma}

    The theorem clearly follows from the two lemmas, as together they handle every possible valuation profile. {For completeness, we formalize this in \Cref{claim:two-agents-completenesss} in \Cref{app:two-agents}.}

    Thus, all that remains is to present the proof of \Cref{lem:hard-case}.

    \begin{proof} [Proof of \Cref{lem:hard-case}]
        As we are in the hard case, we can assume both agents have a shared favorite item $y\in \M$, and w.l.o.g., $\val_1(y) > \frac{2}{3}\cdot \TPS_2(\val_1)$ and $\val_2(y) <\frac{2}{3}\cdot \TPS_2(\val_2)$.

        Because both agents share the same favorite item $y$,  \Cref{obs:2-agents-cuq-fractional} implies the cyclic-unit-quota fractional allocation will be equal to the uniform fractional allocation. Therefore, we will construct a partition of the goods $\M$ into two disjoint sets $(A_1,A_2)$ so that for both agents $i\in\{1,2\}$ and both indices $j\in \{1,2\}$ we have $\val_i(A_j) \ge \frac{2}{3}\cdot \TPS_2(\val_i)$. After doing so, we construct $\D$ by flipping a fair coin over which agent gets each set, or explicitly $\D=\left\{\left(\frac{1}{2},(A_1,A_2)\right), \left(\frac{1}{2},(A_2,A_1)\right)\right\}$. Then $\D$ induces $\cuqalloc(\val_1,\val_2)$, is supported on exactly two allocations, and has the required ex-post fairness guarantees, thus satisfying all the requirements of the theorem.
    
        All that remains is to show the construction of the partition $(A_1,A_2)$ with the required fairness constraints. A visualization of the construction is shown in \Cref{fig:two-agent-hard-case}. We construct $A_1,A_2$ as follows:
        \begin{enumerate}
            \item Let $\ix_1,\ldots,\ix_m$ denote the strict ordering of the items according to $\val_1$ (with ties broken lexicographically). If $m\mod 3\neq 0$, add $1$ or $2$ dummy goods at the end worth zero to both agents to make $m$ divisible by three. Note that $x_1=y$ is the shared favorite item of both agents.
            \item Initialize bundles $\allocset_1\leftarrow \{\ix_1\}, \allocset_2\leftarrow \{\ix_2, \ix_3\}$. For every index $s=1$ up to $\frac{m}{3}-1$:
            \begin{enumerate}
                \item Consider the consecutive triplet $\{\ix_{3\cdot s+1}, \ix_{3\cdot s+2}, \ix_{3\cdot s+3}\}$
                \item Select the item from this triplet which is the best item according to $\val_2$ and add it to $\allocset_1$, and add the remaining two items to $\allocset_2$.
                \item If after the addition, $\val_2(\allocset_1)\geq \frac{2}{3}\cdot \TPS_2(\val_2)$, halt this loop.
            \end{enumerate}
            \item Add all remaining unassigned goods to $\allocset_2$.
        \end{enumerate}
        
        \begin{figure}[htbp]
        \centering

        \begin{tikzpicture}[
    font=\normalsize,
    A1/.style={fill=blue!15, draw=blue!70!black, thick, minimum width=1cm, minimum height=0.9cm, align=center, rounded corners=2pt},
    A2/.style={fill=orange!15, draw=orange!70!black, thick, minimum width=1cm, minimum height=0.9cm, align=center, rounded corners=2pt},
    dots/.style={minimum width=0.8cm, align=center, font=\Large},
    mybrace/.style={decorate, decoration={brace, amplitude=6pt, raise=2pt}, thick},
    mybracedown/.style={decorate, decoration={brace, amplitude=8pt, mirror, raise=2pt}, thick}
]

\node[A1] (x1) at (0,0) {$x_1$};
\node[A2, right=0.1cm of x1] (x2) {\small $x_2$};
\node[A2, right=0.1cm of x2] (x3) {\small $x_3$};

\node[A2, right=0.5cm of x3] (x4) {\small $x_4$};
\node[A1, right=0.1cm of x4] (x5) {\small $x_5$};
\node[A2, right=0.1cm of x5] (x6) {\small $x_6$};

\node[dots, right=0cm of x6] (d1) {\small $\dots$};

\node[A1, right=0cm of d1] (xs1) {\small $x_{3s+1}$};
\node[A2, right=0.1cm of xs1] (xs2) {\small $x_{3s+2}$};
\node[A2, right=0.1cm of xs2] (xs3) {\small $x_{3s+3}$};

\node[A2, right=0.8cm of xs3] (xt1) {\small $x_{3s+4}$};
\node[dots, right=0.1cm of xt1] (d2) {\small $\dots$};
\node[A2, right=0.1cm of d2] (xm) {\small $x_m$};

\draw[->, dashed, >=latex] ([yshift=1cm, xshift=2cm]x1.north west) -- ([yshift=1cm, xshift=-2cm]xm.north east) node[midway, above=0.1cm] {Goods sorted in descending value by agent $1$ ($v_1$)};

\draw[->,thick] ([yshift=-2pt]x1.south) -- ([yshift=-9pt]x1.south) node[midway, below=3pt, align=center] {\small $v_1(y) > \frac{2}{3}\hyp \TPS$};
\draw[mybrace] (x1.north west) -- (x3.north east);
\draw[mybrace] (x4.north west) -- (x6.north east);
\draw[mybrace] (xs1.north west) -- (xs3.north east);

\draw[->,thick] ([yshift=-2pt]x5.south) -- ([yshift=-9pt]x5.south) node[midway, below=3pt, align=center] {\small Best good for $\val_2$ \\goes to $A_1$};
\draw[->,thick] ([yshift=-2pt]xs1.south) -- ([yshift=-9pt]xs1.south) node[midway, below=3pt, align=center] {\small Best good for $\val_2$ \\goes to $A_1$};

\draw[mybracedown] (xt1.south west) -- (xm.south east) node[midway, below=10pt, align=center] {\small Remaining tail\\goes to $A_2$};

\draw[dashed, very thick, red!80!black] ([xshift=0.4cm, yshift=0.6cm]xs3.north east) -- ([xshift=0.4cm, yshift=-0.7cm]xs3.south east) node[below left=-1pt and -30pt, align=center, font=\sffamily\bfseries, red!80!black] {\small Halt Loop:\\$v_2(A_1) \ge \frac{2}{3}\hyp\TPS$};

\node[A1, minimum width=0.6cm, minimum height=0.4cm, above left=1.15cm and -15pt of x1] (leg1) {};
\node[right=2pt of leg1] {\small $A_1$};
\node[A2, minimum width=0.6cm, minimum height=0.4cm, below=0.15cm of leg1] (leg2) {};
\node[right=2pt of leg2] {\small $A_2$};



\end{tikzpicture}

            \caption{Visual intuition for constructing the partition construction in the ``hard'' case where both agents share the same favorite item ($y=x_1$), but $\val_1(y)>\frac{2}{3}\hyp\TPS$ and $\val_2(y)<\frac{2}{3}\hyp\TPS$. Items are processed in triplets in descending order of agent $1$'s valuation. By assigning agent $2$'s preferred item from each triplet to $A_1$ until $v_2(A_1)\ge \frac{2}{3}\hyp\TPS$, both bundles are worth at least $\frac{2}{3}\hyp\TPS$ to both agents.}
            \label{fig:two-agent-hard-case}
        \end{figure}
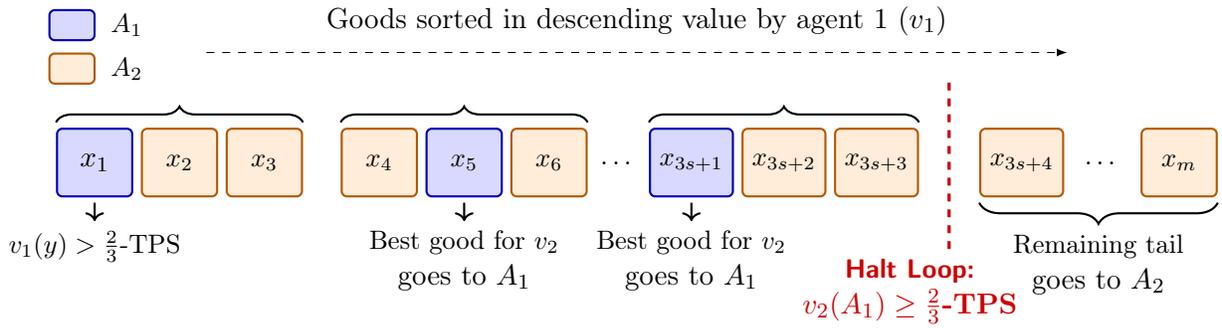
        
        The sets $(\allocset_1,\allocset_2)$ now partition all the goods. Remember we assumed $\val_1(y)>\frac{2}{3}\cdot \TPS_2(\val_1)$ and $\val_2(y)<\frac{2}{3}\cdot \TPS_2(\val_2)$.
        It remains to prove that each of these bundles is worth at least $\frac{2}{3}\hyp\TPS$ for both agents:
        \begin{itemize}
            \item $\val_1(\allocset_1)\geq \frac{2}{3}\cdot \TPS_2(\val_1)$ because $y\in\allocset_1$.
            \item $\val_1(\allocset_2)\geq \frac{2}{3}\cdot \TPS_2(\val_1)$: observe that $\allocset_2$ contains $2$ out of every $3$ consecutive items according to the descending ordering of $\val_1$. Because the items are sorted according to $\val_1$, the value of the missing item in each triplet is bounded above by the average value of the items kept from the previous triplet. Thus \[\val_1(\allocset_2)\ge\sum_{s=0}^{\frac{m}{3}-1}\val_1(\{x_{3s+2},x_{3s+3}\})\ge\frac{2}{3}\sum_{s=0}^{\frac{m}{3}-1}\val_1(\{x_{3s+2},x_{3s+3},x_{3s+4}\})\ge \frac{2}{3}\cdot\val_1(\M\setminus\{y\})\]
            By \Cref{lem:tps-r-monotone} we have $\val_1(\M\setminus\{y\})\ge \TPS_2(\val_1)$, so overall $\val_1(\allocset_2)\ge \frac{2}{3}\cdot \TPS_2(\val_1)$, as required.
            \item $\val_2(\allocset_1)\ge \frac{2}{3}\cdot \TPS_2(\val_2)$: if the loop halted part-way this must be true by the halting condition. Otherwise, set $\allocset_1$ contains the highest value item of $\val_2$ from every triplet (including the first triplet $\{y,x_2,x_3\}$ since $y$ is her favorite item). Therefore $\val_2(\allocset_1)\ge \frac{1}{3}\cdot\val_2(\M)\geq \frac{2}{3}\cdot \TPS_2(\val_2)$, where the last inequality follows from \Cref{lem:share-inequality}.
            \item $\val_2(\allocset_2)\geq \frac{2}{3}\cdot \TPS_2(\val_2)$: note that the loop halts at the first moment when $\val_2(\allocset_1)\geq \frac{2}{3}\cdot \TPS_2(\val_2)$. Thus, whether the loop halts part-way or not, $\val_2(\allocset_1)$ exceeds $\frac{2}{3}\cdot \TPS_2(\val_2)$ by at most one item. Since all items are worth at most $\frac{2}{3}\cdot \TPS_2(\val_2)$ to agent 2, this implies $\val_2(\allocset_1)\leq \frac{4}{3}\cdot \TPS_2(\val_2)$. \Cref{lem:share-inequality} implies $\val_2(\M)\ge 2\cdot \TPS_2(\val_2)$. Thus, $\val_2(\allocset_2)=\val_2(\M)-\val_2(\allocset_1)\ge 2\cdot \TPS_2(\val_2)-\frac{4}{3}\cdot \TPS_2(\val_2)=\frac{2}{3}\cdot \TPS_2(\val_2)$, as required.
        \end{itemize}
    
        We remark that the set $\allocset_1$ is in fact worth $\frac{2}{3}$ of the proportional share to both agents.
    \end{proof}
    We have completed the proof of \Cref{lem:hard-case}, thus completing the proof of the theorem.
\end{proof}

\section*{Acknowledgments}
This research was supported by the Israel Science Foundation (grant No. 301/24 and grant No. 1122/22). Moshe Babaioff's research is also supported by a Golda Meir Fellowship.

\expandafter\def\csname blx@maxbibnames\endcsname{99}\printbibliography

\appendix

\section{\texorpdfstring{Missing Proofs from \Cref{sec:design}}{Missing Proofs from Section~\ref{sec:design}}}
\label{app:design}

\tiefrac*
\begin{proof}
    We will show that for every \emph{additive} valuation profile $\valprof\in\Vaddprofs$ and for every agent $i\in\N$, agent $i$'s expected value from distribution $D=\randMech(\valprof)$ is equal to her value from the induced fractional allocation $\fracD=\fracMech(\valprof)$. Writing $D=\{(\pk,\allock)\}_{k\in [K]}$, by additivity of $\val_i$ and linearity of expectation:
    \[\expect[\alloc\sim \D(\valprof)]{\val_i(\allocset_i)}=
    \sum_{k=1}^K \pk\cdot\left(\sum_{x\in\allocki} \val_i(x)\right)\]
    \[=\sum_{x\in\M} \val_i(x)\cdot\left(\sum_{k=1}^K \pk\cdot 1[x\in \allocki]\right)=\sum_{\ix\in\M}\val_i(x) \cdot \fracDix=\val_i(\fracD)\]

    Assume for contradiction there exists a profile where an agent can improve her expected outcome in $\randMech$ by misreporting (contradicting TIE), then for this same profile she can also improve her outcome in $\fracMech$ by the same misreport. The same argument also holds in the reverse direction. This proves that $\randMech$ is TIE iff $\fracMech$ is truthful, as required.
\end{proof}

\unitotie*
\begin{proof}
    Let $\fracMech$ denote the fractional mechanism induced by both mechanisms. We will show that the universal-truthfulness of $\univMech$ implies it is also TIE. Beforehand, we show this is sufficient to complete the proof. We apply \Cref{lem:tie-fractional} twice: the first time from $\univMech$ to $\fracMech$, proving that $\fracMech$ is a truthful fractional mechanism, and then a second time from $\fracMech$ to $\randMech$, proving that $\randMech$ is TIE, as required.

    All that remains to be shown is that universal-truthfulness implies truthfulness-in-expectation. Since $\univMech$ is universally-truthful it can be represented with $\ell$ truthful deterministic mechanisms $\mech^{(1)},\ldots,\mech^{(\ell)}$, with respective probabilities $\lambda_1,\ldots,\lambda_\ell$, so that for given input valuation profile $\valprof\in\Vaddprofs$, mechanism $\univMech$ samples each mechanism $\mech^{(j)}$ with probability $\lambda_j$ and outputs $\mech^{(j)}(\valprof)$. This implies that for profile $\valprof\in\Vaddprofs$ and agent $i\in\N$, the expected value of agent $i$ from $\univMech$ is $\expect[\alloc\sim \univMech(\valprof)]{\val_i(\allocset_i)}=\sum_{j=1}^\ell \lambda_j\cdot v_i(\mech^{(j)}(\valprof))$.
    
    Assume for contradiction $\univMech$ is not TIE, then there exists a valuation profile $\valprof\in\Vaddprofs$, agent $i\in\N$ and alternative valuation $\valp_i\in\Vadd$ such that agent $i$ can strictly increase her expected value by reporting $\valp$, thus changing the overall profile to $\valprofp=(\valp_i,\val_{-i})$. Therefore:
    \[\sum_{j=1}^\ell \lambda_j\cdot v_i(\mech^{(j)}(\valprof))=\expect[\alloc\sim \univMech(\valprof)]{\val_i(\allocset_i)}<\expect[\alloc\sim \univMech(\valprof')]{\val_i(\allocset_i)}=\sum_{j=1}^\ell \lambda_j\cdot v_i(\mech^{(j)}(\valprofp))\]

    Yet, this implies there must exist at least one index $1\le j\le \ell$ s.t. $v_i(\mech^{(j)}(\valprof))<v_i(\mech^{(j)}(\valprofp))$, contradicting the truthfulness of $\mech^{(j)}$.    
\end{proof}

\cuqTieProp*
\begin{proof}
    Note that for any starting index $j\in[n]$, the $\percyc{j}$-unit-quota allocation mechanism (see \Cref{def:det-unit-quota}) is deterministic and truthful. Therefore, since the cyclic-$\per$-unit-quota distributional mechanism is constructed by uniformly sampling one of these $n$ mechanisms and running it, the mechanism is \emph{universally-truthful}. Thus \Cref{prop:universal-to-tie} implies that any distributional mechanism $\distMech$ that induces the cyclic-$\per$-unit-quota fractional mechanism is TIE, proving the first requirement of the lemma.

    It remains to show that $\distMech$ is proportional ex-ante. Importantly, valuations are additive, a distributional mechanism is proportional ex-ante if and only if its induced fractional mechanism is proportional. This holds because, for additive valuations, the expected value of an agent from a distributional allocation is equal to her value from the induced fractional allocation. Therefore, all that remains to be shown is that the  will show that the cyclic-$\per$-unit-quota fractional mechanism is proportional.

    Fix an agent $i\in\N$ with valuation $\val_i\in\Vadd$, let $e_1,\ldots,e_m$ denote agent $i$'s ordering of the goods (ties broken lexicographically). Notice that, irrespective of the ordering $\per$, across the $n$ cyclic shifts agent $i$ will pick in all $n$ possible positions, from first to last. For each $j\in[n-1]$, in the cyclic ordering where agent $i$ picks $j$-th in the ordering, she is guaranteed at least the value of her $j$-th best good $\val_i(e_j)$. Additionally, in the cyclic ordering where she is last (and receives all unpicked goods), she is guaranteed at least the total value of all her goods besides the first $n-1$, that is at least value $\val_i(\M\setminus \{e_1,\ldots,e_{n-1}\})$. Thus over all the $n$ cyclic orderings the agent receives a total value of at least $\sum_{j=1}^{n-1}\val_i(e_j)+\val_i(\M\setminus \{e_1,\ldots,e_{n-1}\})=\val_i(\M)$. Therefore, her value from the fractional allocation generated by the cyclic-$\per$-unit-quota mechanism will be at least $\frac{\val_i(\M)}{n}$, her proportional share.
\end{proof}

\section{\texorpdfstring{Missing Proofs from \Cref{sec:logn}}{Missing Proofs from Section~\ref{sec:logn}}}
\label{app:logn}

\greedyCompletion*

\begin{proof}
    We construct the completed matrix $\matr'$ iteratively using a greedy algorithm. Initialize $\matr' \leftarrow \matr$. Let $a_u = \sum_{v=1}^c \matr'_{uv}$ denote the current sum of row $u$, and $b_v = \sum_{u=1}^r \matr'_{uv}$ denote the current sum of column $v$. Define the sets of unsatisfied rows and columns as $R = \{u \in [r] \mid a_u < s_u\}$ and $C = \{v \in [c] \mid b_v < 1\}$.

    While both $R \neq \emptyset$ and $C \neq \emptyset$, we perform the following steps:
    \begin{enumerate}
        \item Select any row $u \in R$ and any column $v \in C$.
        \item Compute the maximum allowed increase $\delta = \min(s_u - a_u, 1 - b_v)$. Since $u \in R$ and $v \in C$, it is guaranteed that $\delta > 0$.
        \item Update $\matr'_{uv} \leftarrow \matr'_{uv} + \delta$.
        \item Update the running sums $a_u \leftarrow a_u + \delta$ and $b_v \leftarrow b_v + \delta$.
        \item If $a_u = s_u$, remove $u$ from $R$. If $b_v = 1$, remove $v$ from $C$.
    \end{enumerate}

    In every iteration of the loop, either a row reaches its target and is removed from $R$, or some column reaches its target and is removed from $C$ (or both). Therefore, the loop terminates after at most $r + c$ iterations.

    We now show that $\matr_{uv} \le\matr'_{uv} \le 1$. Because initially $\matr\in {[0,1]}^{r\times c}$, and our algorithm only adds strictly positive values $\delta$, all entries in $\matr'$ do not decrease from their original values in $\matr$. Moreover, observe that for any column $v$, the total sum $b_v$ never exceeds its target of $1$, therefore since elements are non-negative, no individual element $\matr'_{uv}$ can ever exceed $1$ for any $(u,v)\in[r]\times[c]$.   

    We next show that, throughout the entire run of the loop, the row and column targets are never exceeded. At the beginning this is the case by requirements (1) and (2) on the input matrix $\matr$. This invariant is then maintained at each iteration by the construction of $\delta$.

    Finally, by the initial assumption, the total row target equals the total column target ($\sum_{u=1}^r s_u = c = \sum_{v=1}^c 1$). Since every update step increases the total sum across all rows and the total sum across all columns by the exact same amount $\delta$, the last column and row must both reach their target simultaneously. Therefore, in the last iteration of the loop both $R$ and $C$ are updated to $\emptyset$. Then when the loop terminates both $R=C=\emptyset$, meaning that all target row and column sums are met exactly, as required.
\end{proof}

\section{\texorpdfstring{Missing Proofs from \Cref{sec:two-agents}}{Missing Proofs from Section~\ref{sec:two-agents}}}
\label{app:two-agents}

\defeasycases*
\easycases*
\begin{proof}

    We solve each case separately.

    \textit{Case 1:} Note from \Cref{obs:2-agents-cuq-fractional} that in this case, in $\cuqalloc(\val_1,\val_2)$ each agent $i\in\{1,2\}$ is fully allocated $y_i$, and receives $1/2$ of every remaining item $\M\setminus\{y_1,y_2\}$. For $i\in\{1,2\}$ let $\ei=(\eiat{1},\ldots,\eiat{m})$ be a strict ordinal preference consistent with $\val_i$. Let $\D$ be distributional allocation generated by applying the faithful implementation technique in \Cref{lem:faithful} on fractional allocation $\cuqalloc(\val_1,\val_2)$ and ordinal preferences $(\en{1},\en{2})$. Then $\D$ is guaranteed to induce $\cuqalloc(\val_1,\val_2)$. Moreover, because all strictly fractional values in $\cuqalloc(\val_1,\val_2)$ are exactly $\frac{1}{2}$, $\D$ is guaranteed to have support of size exactly $2$. This implies there must exist a partition $A_1,A_2$ of the goods $\M\setminus\{y_1,y_2\}$ such that $\D$ randomizes between allocations $(A_1\cup \{y_1\},A_2\cup\{y_2\})$ and $(A_2\cup \{y_1\},A_1\cup\{y_2\})$ with probability $\frac{1}{2}$.

    Fix agent $i\in\{1,2\}$. If $\val_i(y_i)\ge \frac{2}{3}\cdot \TPS_2(\val_i)$ then agent $i$ is guaranteed $\frac{2}{3}$ of her $\TPS$ ex-post in $\D$, as required. Otherwise, to agent $i$, all goods are worth less than $\frac{2}{3}$-$\TPS$. Using the additional property of \Cref{lem:faithful} that the difference in value between two ex-post allocations is at most the value of one (strictly fractionally assigned) good, and since the sum of both ex-post allocations is worth at least $2\cdot\TPS_2(\val_i)$, we get that both ex-post bundles must be worth at least $\frac{2}{3}$-$\TPS$ to agent $i$, as required.
    
    \textit{Case 2:} Note from \Cref{obs:2-agents-cuq-fractional} that in this case $\cuqalloc(\val_1,\val_2)$ is equal to the uniform fractional allocation, and also for both agents all goods are worth at most $\frac{2}{3}$-$\TPS$. For $i\in\{1,2\}$ let $\ei=(\eiat{1},\ldots,\eiat{m})$ be a strict ordinal preference consistent with $\val_i$. Let $\D$ be distributional allocation generated by applying the faithful implementation technique in \Cref{lem:faithful} on fractional allocation $\cuqalloc(\val_1,\val_2)$ and ordinal preferences $(\en{1},\en{2})$. Then $\D$ is guaranteed to induce $\cuqalloc(\val_1,\val_2)$. Moreover, because all strictly fractional values in $\cuqalloc(\val_1,\val_2)$ are exactly $\frac{1}{2}$, $\D$ is guaranteed to have support of size exactly $2$. Using the additional property of \Cref{lem:faithful} that the difference in value between two ex-post allocations is at most the value of one (strictly fractionally assigned) good, since the sum of both ex-post allocations is worth at least $2\cdot\TPS_2(\val_i)$, and all goods are worth $\frac{2}{3}$-$\TPS$, for both agents $i\in\{1,2\}$ we get that both ex-post bundles must be worth at least $\frac{2}{3}$-$\TPS$ to agent $i$, as required.

    \textit{Case 3:} In this case we will construct the output distribution explicitly. Let $\D$ be the distribution which flips a fair coin to assign the shared favorite item $y$ to one agent and all remaining items to the other agent. Formally:
    \[\D=\left\{\left(\frac{1}{2}\;,\;\left(\{y\},\M\setminus\{y\}\right)\right)\;,\;\left(\frac{1}{2}\;,\;\left(\M\setminus\{y\},\{y\}\right)\right)\right\}\]
    We remark that $\D$ is the same distribution generated by the universally-truthful cyclic-unit-quota distributional mechanism for this profile.

    We now prove the required fairness guarantee. Fix an agent $i\in\{1,2\}$. If she receives item $y$ alone she will be satisfied because we assumed $\val_i(y)\ge \frac{2}{3}\cdot\TPS_2(\val_i)$ in this case. Otherwise, if she receives all remaining items she actually receives her full $\TPS$ because by \Cref{lem:tps-r-monotone}, $\val_i(\M\setminus\{y\})\ge \TPS_2(\val_i)$.
\end{proof}

\begin{claim}
    \label{claim:two-agents-completenesss}
     \Cref{lem:easy-cases} in conjunction with \Cref{lem:hard-case} suffice to complete the proof of \Cref{thm:two-agents}.
\end{claim}
\begin{proof}
    We design the mechanism as follows. Given input profile $(\val_1,\val_2)\in\Vadd^2$, the mechanism constructs the output distributional allocation as follows. If $(\val_1,\val_2)$ is \emph{easy}, we apply \Cref{lem:easy-cases}, and output the resulting distribution $\D$. Similarly, if $(\val_1,\val_2)$ is \emph{hard} we apply \Cref{lem:hard-case}, and output the resulting distribution $\D$.
    
    We prove the above mechanism satisfies all the requirements of the theorem. \Cref{lem:easy-cases} and \Cref{lem:hard-case} both run in polynomial time and output distributional allocations which are supported on exactly two allocations, both of which are $\frac{2}{3}$-$\TPS$, so the same holds for our mechanism in general. Moreover, since both output a distributional allocation that induces the cyclic-unit-quota fractional allocation, overall our mechanism will induce the cyclic-unit-quota fractional mechanism, so by \Cref{lem:cuq-tie}, this guarantees that the mechanism will be both TIE and ex-ante proportional (which for $n=2$ agents implies ex-ante envy-freeness as well).
\end{proof}

\negtwo*
\begin{proof} 

    Assume for contradiction there exists a randomized mechanism $\randMech$, with \emph{ordinal} induced fractional mechanism $\fracMech$, such that every allocation in the support of $\randMech$ is $\apx$-$\MMS$ ex-post for some $\apx>\frac{2}{3}$. Consider the following valuation profile for $n=2$ agents and $m\ge 4$ items:

    \begin{center}
        \begin{tabular}{c|ccccccc}
         & $x_1$ & $x_2$ & $x_3$ & $x_4$ & $x_5$ & $\cdots$ & $x_m$ \\
         \hline
         $\val_1$ & $1$ & $\frac{1}{3}$ & $\frac{1}{3}$ & $\frac{1}{3}$ & $0$ & $\cdots$ & $0$ \\
         $\val_2$ & $\frac{1}{2}$ & $\frac{1}{2}$ & $\frac{1}{2}$ & $\frac{1}{2}$ & $0$ & $\cdots$ & $0$  \\
        \end{tabular}
    \end{center}

    Note that for both agents $i\in\{1,2\}$ in this profile $\MMS_2(\val_i)= \TPS_2(\val_i)= \PROP_2(\val_i)=1$. Consider the induced fractional allocation $\falloc=\fracMech(\val_1,\val_2)$ generated by the mechanism for profile $(\val_1,\val_2)$. Separate into cases.
    
    \textit{Case 1:} Assume that $\f_{1,x_1}<1$. Then there must exist an ex-post allocation $\alloc\in\supp{\randMech(\val_1,\val_2)}$ such that $x_1\notin \allocset_1$. In order for $\allocset_1$ to still be at least $\apx$-$\TPS$ for agent $1$, it must be that $\allocset_1\supseteq\{x_2,x_3,x_4\}$.  But now, the total value of all remaining goods for agent $2$ is $\frac{1}{2}$, so $\val_2(\allocset_2)=\val_2(\M\setminus\allocset_1)\le \frac{1}{2}<\apx=\apx\cdot \TPS_2(\val_2)$, as required.

    \textit{Case 2:} Otherwise, assume $\f_{1,x_1}=1$, then $\f_{2,x_1}=0$. Now consider the reversed instance $(\val_2,\val_1)$, in which the valuations of agents $1$ and $2$ are swapped. Note that while the profiles differ cardinally, from an ordinal perspective both profiles are identical, since $\val_1,\val_2$ both order the items identically (we assume ties are broken lexicographically). Therefore, since we assumed that the induced fractional mechanism $\fracMech$ is ordinal, we have that $\fracMech(\val_2,\val_1)=\fracMech(\val_1,\val_2)=\falloc$. But now agent $2$ with true valuation $\val_1$ receives item $x_1$ with probability zero, and so we get the same contradiction as in the first case.
\end{proof}

\subsection{A Note on the Communication Complexity of the Mechanism}
\label{app:twoAgentsComplexity}

For the mechanism of Theorem~\ref{thm:two-agents}, each agent needs to report ordinal information (a strict order over the items, consistent with the valuation function of the agent), and relatively few additional bits of information.
Each agent needs to send one additional bit specifying whether her top item has value at least $\frac{2}{3}$-MMS. In addition, one of the agents (say, agent~2) might need to send an integer $s$ in the range $[1, m]$ (determining in which triple to stop, as explained is Section~\ref{sec:overviewTwoAgents}). 

We assume that communication takes two rounds, and that all communications are public. That is, when agent~2 is requested to send $s$, she already knows the total order provided by agent~1. Under this assumption, agent~2 can figure out $s$ based on her own $v_2$, and send only $s$ ($\log m$ bits), and not her full valuation $v_2$. Having two rounds in the mechanism does not affect the TIE property, because the fractional allocation which determines the ex-ante value that agents get is fixed after the first round of communication. 

In fact, one can save on the number of bits needed in order to represent $s$. This is because given the orders over goods supplied by the two agents, there is a {\em hitting set} of $O(\log m)$ integers in the range $[1,m]$, such that for every $v_2$ consistent with the reported order, at least one member of the hitting set can qualify as $s$. Sending the index of this member requires only roughly $O(\log\log m)$ bits.

We now explain the construction of the hitting set. Scale the valuation $v_2$ so that $v_2(\goods) = 2$. 
Let $e_1^1, \ldots, e^1_m$ and $e_1^2, \ldots, e^2_m$ be the orderings over the items reported by agent~1 and agent~2 in the first round of communication. We assume that agent~2, who needs to send $s$ in the second round, can see this ordering. Recall that the mechanism partitions $e_1^1, \ldots, e^1_m$ into consecutive triples. Agent~2 needs to send the index $s$ of a triple for which the sum of $v_2$ values of the best item from each triple up to that point is between $\frac{2}{3}$ and $\frac{4}{3}$. 

Consider an auxiliary additive valuation $v'_2$, in which $v'_2(e^2_1) = v'_2(e^2_2) = \frac{2}{3}$, and  $v'_2(e^2_k) = \frac{2}{k}$ for every $k \ge 3$. Importantly, for every item $e$ it holds that $v'_2(e) \ge v_2(e)$. Note that unlike the private valuation $v_2$, the valuation $v'_2$ is public knowledge, given the reported order  $e_1^2, \ldots, e^2_m$. The total value of all items under $v'_2$ is roughly $2\ln m$. 

Extract from each triple in $e_1^1, \ldots, e^1_m$ the highest value item according to $v_2$. (This information is available, by the ordering supplied by agent~2). Let $\sigma = e_1, \ldots, e_{m/3}$ denote the sequence of items obtained by this process. Partition $\sigma$ into $O(\log m)$ consecutive blocks, each of value at most $\frac{2}{3}$ according to $v'_2$. The hitting set referred to above is the indices in $\sigma$ in which blocks end. By construction, it is public knowledge.

Within this hitting set, $s$ is chosen to be the first index by which the sum of $v_2$ values of items in $\sigma$ is at least $\frac{2}{3}$. Reporting $s$ can be done using $\log\log m + O(1)$ bits. As no item has $v'_2$ value larger than $\frac{2}{3}$, and $v'_2 \ge v_2$, the $v_2$ value of the first $s$ items in $\sigma$ is at least $\frac{2}{3}$, and at most $\frac{4}{3}$, as desired.

{If we insist on a one-round mechanism, then each agent needs to send in advance, without knowing the order of the other agent over the items, sufficient information that would allow one to compute $s$, if the need will arise. This can be done using $O((\log m)^2)$ bits. Fixing some small $\epsilon > 0$ ($\epsilon$ is independent of $m$), for every $r \in [0, \frac{2}{\epsilon}\log m]$, each agent reports also how many goods have value between $\frac{1}{m^2}(1 + \epsilon)^r$ and $\frac{1}{m^2}(1 + \epsilon)^{r+1}$. Using also the ordinal information, this allows to approximate the value of any set of goods up to a small error. This information suffices so that if the need arises, one can extract from the sequence $\sigma$ above a set of indices (not necessarily consecutive) such that the set of items in these locations has value  between $\frac{2}{3}$ and $\frac{4}{3}$. Further details are omitted.}

\end{document}